%% file: main_arxiv.tex
\documentclass[11pt]{article}
\usepackage[letterpaper, left=1in, right=1in, top=1in,bottom=1in]{geometry}
\usepackage{amsfonts}       %
\usepackage{nicefrac}       %
\usepackage{bbm}
\usepackage[ruled,noend]{algorithm2e}
\usepackage{bm}
\usepackage{tikz}
\usetikzlibrary{positioning,chains,fit,shapes,calc,arrows}
\usetikzlibrary{arrows.meta}
\usepackage{graphicx}
\usepackage{pgf}
\usepackage{mathtools}
\usepackage{mathrsfs}
\usepackage{subfigure}
\usepackage{amsmath,amsthm,amssymb}
\usepackage{comment}
\usepackage{xfrac}
\usepackage{accents}
\usepackage{cancel}
\usepackage[normalem]{ulem}

\usepackage[round]{natbib}

\usepackage[pagebackref=true]{hyperref} \PassOptionsToPackage{pagebackref=true}{hyperref}

\DeclareMathAlphabet{\pazocal}{OMS}{zplm}{m}{n}

\usepackage{auxiliary}

\usepackage{float}

\usepackage{url}      

\usepackage[pagebackref=true]{hyperref} \PassOptionsToPackage{pagebackref=true}{hyperref}   
\usepackage[capitalize]{cleveref}

\NewDocumentEnvironment{myproof}{o}
  {\IfNoValueTF{#1}{\paragraph{{{\small Proof.}}}} {\paragraph{{{\small #1.}} }} }
  {\hfill$\qed$}

\usepackage{color}              %
\usepackage{color-edits}
\addauthor{jb}{magenta}
\addauthor{grace}{blue} %

\title{
Strategic Hiring under Algorithmic Monoculture
}
\author{
Jackie Baek\thanks{Stern School of Business, New York University, \texttt{baek@stern.nyu.edu}} 
\and 
Hamsa Bastani\thanks{Wharton School, University of Pennsylvania, \texttt{hamsab@wharton.upenn.edu}}
\and 
Shihan Chen\thanks{Graduate Group in Applied Mathematics and Computational Science, University of
Pennsylvania, \texttt{gracsh@sas.upenn.edu}}
}

\begin{document}

\maketitle

\thispagestyle{empty}

\begin{abstract}
\input{abstract}
\end{abstract}

\newpage
\setcounter{page}{1}

\input{paper_contents}

\bibliographystyle{abbrvnat}
\bibliography{references} 

\newpage
\appendix

\input{appendix}

\end{document}

%% file: abstract.tex
We study the impact of strategic behavior in labor markets characterized by \textit{algorithmic monoculture}, where firms compete for a shared pool of applicants using a common algorithmic evaluation. In this setting, ``naive'' hiring strategies lead to severe congestion, as firms collectively target the same high-scoring candidates. We model this competition as a game with capacity-constrained firms and fully characterize the set of Nash equilibria. We demonstrate that equilibrium strategies, which naturally diversify firms' interview targets, significantly outperform naive selection, increasing social welfare for both firms and applicants. Specifically, the \textit{Price of Naive Selection} (welfare gain from strategy) grows linearly with the number of firms, while the \textit{Price of Anarchy} (efficiency loss from decentralization) approaches 1, implying that the decentralized equilibrium is nearly socially optimal. Finally, we analyze convergence, and we show that a simple sequential best-response process converges to the desired equilibrium. 
However, we show that firms generally cannot infer the key input needed to compute best responses, namely congestion for specific candidates, from their own historical data alone.
Consequently, to realize the welfare gains of strategic differentiation, algorithmic platforms must explicitly reveal congestion information to participating firms.

%% file: paper_contents.tex
\section{Introduction}
\label{sec:intro}

The rise of online job platforms has dramatically transformed the hiring landscape. Applicants can now apply to a vast number of jobs with minimal effort, resulting in job postings often receiving an overwhelming volume of applications. To manage this influx, many employers rely on algorithmic tools to aid in the screening process---three-fourths of employers in the United States use automated tools to screen candidates \citep{fuller2021hidden}. These tools range from simple keyword filters to advanced machine learning models that evaluate answers to interview questions \citep{raghavan2020mitigating}. While research has shown that algorithmic recommendations can improve job fill rates \citep{horton2017effects}, demonstrating their potential to enhance labor market efficiency, the increasing reliance on and widespread adoption of these systems raise questions about their implications for labor market outcomes.

While larger firms may be able to develop such algorithmic systems in-house, many firms employ tools provided by external vendors (e.g., LinkedIn, HireVue, ThriveMap, etc.). This can result in multiple firms using the \textit{same} algorithm, which raises fundamental questions about competition and efficiency in hiring outcomes in such settings.   
This scenario where multiple agents leverage the same algorithm is referred to as \textit{algorithmic monoculture}, a term which draws parallels to the agricultural practice of cultivating a single crop species, which can lead to vulnerabilities such as disease outbreaks due to lack of diversity. Similarly, algorithmic monoculture can create inefficiencies and homogenization in hiring outcomes. Recent studies have highlighted potential risks associated with this phenomenon, including reduced firm utility and adverse effects on applicant welfare \citep{kleinberg2021algorithmic,bommasani2022picking,peng2023monoculture}. In this paper, we consider a setting where algorithmic monoculture exists, and we evaluate how firms can make \textit{strategic} interview and hiring decisions to mitigate inefficiencies.

Concretely, consider an algorithm that evaluates candidates based on their job application and assigns each a scalar score representing their predicted value to employers. When multiple firms use the same algorithm, they all see identical scores for each candidate. 
In a simplified setting where each firm can interview exactly one applicant, all firms might target the highest-scoring candidate.  This is an undesirable outcome from both the perspective of the firms (at most one firm will successfully hire) as well as for the applicants (at most one applicant will be hired). However, this inefficiency is not inevitable as firms can adopt \textit{strategic behaviors}, such as avoiding competing for top candidates when they anticipate high competition. The extent to which such strategies can mitigate inefficiencies and improve labor market outcomes is the focus of this paper.

We consider a model with $N$ firms and a shared pool of applicants.
Each firm makes interview decisions that are capacity-constrained, with the goal of maximizing the number of successful hires.
Each applicant is associated with a score $s \in [0, 1]$, visible to all firms, which represents the probability that the applicant will pass an interview.
Each firm decides which applicants to interview, and then gives an offer to all applicants who pass the interview.
If an applicant receives offers from multiple firms, they accept one of them at random.
We consider two variants of how the offer decisions are correlated across firms for the same applicant, either fully correlated or independent.  

We define social welfare to be the total number of applicants who accept a job. 
Under this model, we consider three different solutions that vary in the strategic nature of the firms:
\begin{enumerate}
    \item \textit{Naive}: All firms interview the highest scoring candidates.
    \item \textit{Nash equilibrium (NE)}: All firms are strategic in their interview decisions, and the strategies form a Nash equilibrium. 
    \item \textit{Centralized}: Every firm's decision is controlled by a centralized decision maker who chooses the outcome that maximizes social welfare.
\end{enumerate}

\subsection{Summary of Results}

\paragraph{Characterize Nash equilibria.} We parameterize strategies using \textit{thresholds} on applicant scores, where pairs of consecutive thresholds define the score ranges receiving a given number of interviews.
We show that a strategy profile forms an NE if and only if the corresponding thresholds satisfy a specific set of properties that we identify.
Interestingly, the resulting equilibrium strategies are often non-monotonic in candidate quality---for example, a firm may choose to interview the top and third deciles of candidates, but skip the second decile to avoid competition.

We also characterize how the number of interviews each applicant receives depends on their score under the NE. 
We show that when the firms' hiring decisions are independent, all applicants who receive interviews get nearly the same number of interviews, while under correlated decisions, the number of interviews increases steadily with applicant scores.

\paragraph{Social welfare comparisons.} 
We show that social welfare, the total number of applicants hired, increases as one goes down the list from Naive, to NE, to Centralized.
We analyze the \textit{Price of Naive Selection (PoNS)}, the ratio of social welfare between the NE and Naive solutions.
We establish that PoNS is high in regimes where there is a large number of firms and the interview capacity is small; specifically, when the interview capacity goes to 0, the PoNS goes to $N$.
We also show that for any interview capacity, the PoNS is higher when the offer decisions are independent across firms, compared to when they are correlated.
Therefore, under these regimes, naive strategies result in large inefficiencies, and there is a huge welfare gain when firms behave strategically. 

We also evaluate the \textit{Price of Anarchy (PoA)}, the ratio of welfare between the Centralized and NE solutions.
In the same regimes where we show that the PoNS is high, we show that the PoA goes to 1.
That is, there is essentially no loss in welfare due to the lack of a centralized decision-maker.
Therefore, the Nash equilibrium yields essentially the best possible social welfare, and hence strategic behavior by firms is desirable from both the firms' and the applicants' perspectives.

\paragraph{Convergence to NE.}
Though we show that the NE is desirable, it is unclear whether firms will be able to converge to an NE in practice.
To this end, we first show that best response dynamics always converge to an NE in a discretized version of the game. Moreover, when interview capacities are small, a simple set of ``one-turn'' best response dynamics---where each firm sequentially chooses the best response assuming they are the last to act---also converges to an NE.
Therefore, provided firms can compute best responses, an NE can arise relatively easily.

\paragraph{Best response is difficult to estimate.}
However, computing a best response requires information on competitors’ strategies, which is typically unavailable in practice. 
We show that it is difficult for a firm to estimate these strategies using only their own historical data. 
Specifically, in a simplified setting where a firm must choose between two applicant pools differing only in success probabilities, distinguishing which pool yields higher expected utility requires a prohibitively large number of samples (e.g., on the order of hundreds).
This implies that without direct congestion information provided by the platform, firms are unlikely to reliably estimate best responses, making convergence to equilibrium behavior unrealistic.

\paragraph{Extension: flexible capacity.}
We consider an extension where firms have a fixed target hiring yield rather than a fixed interview capacity, where they can adjust their capacity to meet this target. 
We show that for any target welfare level, the NE solution achieves the goal with far less total interview capacity than the Naive solution. Thus, strategic behavior enables firms to meet hiring needs much more efficiently.

\paragraph{Extension: applicant preferences over tiered firms.}
We also consider an extension where firms belong to two tiers, and all applicants prefer firms from the higher tier.
We provide sufficient conditions for an NE in this setting and show that social welfare is higher under the tiered NE compared to the non-tiered NE. This implies that it is even more important in the tiered regime for firms to be strategic and arrive at an NE.

\subsection{Implications}

\paragraph{Value of strategic behavior.} First, we demonstrate the substantial value of strategic hiring under algorithmic monoculture.
Strategic behavior improves welfare for both sides of the market.
Since simple best response dynamics converge to the Nash equilibrium, the optimal outcome is attainable provided firms can compute their best response.

\paragraph{Value of congestion information.}
For firms to compute their best response, they require information on the congestion level (the number of other firms interviewing an applicant).
In settings where firms operate on a common algorithmic hiring platform, the platform can play a crucial role by providing this data. With access to congestion information, firms can make informed, strategic decisions that reduce the inefficiencies caused by overlapping interview selections. 
Thus, platforms that prioritize transparency and provide insights into applicant competition levels enable more effective use of algorithmic hiring systems, benefiting both firms and applicants.

Conversely, if sharing congestion information is infeasible, our results imply that the utility of a common algorithmic hiring platform may be limited. Without such information, firms will likely default to naive strategies, leading to high congestion and low social welfare. This highlights a critical consideration for firms adopting algorithmic hiring platforms: if the platform cannot facilitate coordination or mitigate competition, its adoption may fail to deliver intended benefits and could exacerbate inefficiencies. 
This aligns with existing literature raising concerns regarding algorithmic monoculture \citep{kleinberg2021algorithmic,peng2023monoculture}.

\subsection{Related Work}

\paragraph{Algorithmic monoculture and homogenization.}
The concept of algorithmic monoculture was formalized by \citet{kleinberg2021algorithmic}, who showed that firms relying on a common algorithm may hire weaker applicants than when each firm uses an independent, but less individually accurate, hiring method. 
\citet{peng2023monoculture} incorporate two-sided preferences and competition to compare outcomes between monoculture and polyculture (when firms make independent decisions).
They leverage a two-sided matching model  and evaluate the stable matching outcome under monoculture and polyculture.
They show that monoculture can reduce firm utility compared to polyculture, but monoculture can improve average utility for the applicants.
\citet{kleinbergprice} generalize the model in \cite{kleinberg2021algorithmic} and quantify the social welfare loss.

Compared to the above works, the goal of this paper is not to evaluate the benefits or downsides of monoculture. Rather, we simply assume that monoculture exists, and then we evaluate how firms should make decisions in a setting with congestion effects.
We show that social welfare can drastically improve when firms make \textit{strategic} decisions based on the algorithm's output, compared to when they naively follow the algorithm's recommendation.

\cite{besbes2025impact} study the impact of public rankings (e.g., college rankings), which can be thought of as analogous to algorithmic monoculture. They show that under supply constraints, the value of public rankings are limited, whereas personalized recommendations can provide substantial welfare gains.
Though \cite{besbes2025impact} do not study strategic behavior, they arrive at a similar insight as us---global rankings may reduce efficiency by creating congestion for the top items.

One consequence of  algorithmic monoculture is \textit{outcome homogenization}, the idea that certain individuals systematically experience undesirable outcomes by many algorithmic systems. There is growing line of work that study homogenization caused by algorithms \citep{ajunwa2019paradox,bommasani2022picking,jain2024algorithmic,toups2024ecosystem}, and the recent advances in generative AI has sparked studies on its impact on diversity of outcomes \citep{padmakumar2023does,anderson2024homogenization,doshi2024generative,raghavan2024competition,zhou2024generative}.
Our paper studies outcome homogenization in the hiring context. 
Indeed, under the naive baseline where all firms interview the top-scoring candidates, every applicant receives the same outcome (interview decision) from every firm. 
We study whether this homogenization can be mitigated through strategic behavior. 

\paragraph{Congestion in matching markets.}
A key difference of our model to that of \citet{peng2023monoculture} is that the latter uses \textit{stable matching} as the solution concept, without specifying the \textit{process} in which the stable matching arises.
There are also papers that study stable matching in a market where interviews are conducted to learn the utility of match  
 \citep{beyhaghi2021randomness,allman2025signaling,ashlagi2025stable}.
A stable matching can be found, for example, iteratively using the deferred acceptance algorithm \citep{gale1962college}. 
In contrast, our paper fixes a \textit{one-step} process in which firms hire applicants: firms decide simultaneously which applicants to interview and gives offers to everyone who passes the interview.
This process is motivated by the fact that firms face screening costs and hence have capacity constraints on the number of applicants they can interview. 
This process causes issues due to congestion; i.e., reduced utility when multiple firms interview the same applicant.

Other papers have modeled congestion or search costs in matching models \citep{kadam2015interviewing,vohra2024matching, halaburda2018competing,arnosti2021managing,kanoria2021facilitating,lee2017interviewing,manjunath2023interview,donahue2025optimal}. 
\citet{kadam2015interviewing} analyzes how increases in student interview capacity affect total surplus and the number of matches, while \cite{vohra2024matching} focus on analyzing how the \textit{timing} of interviews and offer acceptances affects the total welfare, assuming firms are always strategic. The other papers show that various restrictions on the matching process can alter outcomes. For example, \citet{halaburda2018competing,arnosti2021managing,kanoria2021facilitating,lee2017interviewing} show that reducing the number of applications that an individual can send, or restricting which side can initiate a match can improve the matching outcome. Compared to these works, the main difference of our paper is that we assume all firms have access to an informative signal for all applicants, while these papers assume that each agent has no a priori knowledge about other agents (i.e., no algorithmic recommendation). \cite{manjunath2023interview} show that applicants can be potentially harmed if applicants can accept more interviews in the residency matches between doctors and hospitals, where they assume that hospitals always make naive decisions. 
\cite{donahue2025optimal} studies a single decision-maker where candidates who are higher quality tend to be less likely to be available, and they show that algorithms with higher accuracy can decrease social welfare. 

Several empirical studies have evaluated various interventions to improve efficiency under congestion.
\citet{gee2019more,besbes2023signaling,fradkin2023competition,filippas2024costly} empirically show that signaling the level of competition can improve efficiency.
The benefit of this intervention is also established in our paper, and therefore our work is complementary to these empirical findings.
\citet{manshadi2023redesigning} design a ranking algorithm to take congestion into account on an online platform to match volunteers to nonprofits.
Specific to the hiring application, \citet{horton2017effects} shows that leveraging algorithmic recommendations of applicants to firms substantially increase the fill rate, demonstrating the value of algorithmic hiring.
\citet{horton2024reducing} conduct a field experiment which showed that imposing a cap on the number of applications a job opening can receive can improve efficiency. \citet{coles2013preference} study a mechanism where applicants can send a signal of interest to firms, and this can improve applicant welfare.
\citet{dwork2024equilibria} studies congestion and incoordination in a social network where individuals can refer others for job opportunities.

\paragraph{Algorithmic hiring and discrimination.}
There is a large literature studying algorithmic hiring, on developing the algorithms themselves (e.g., \citet{purohit2019hiring,epstein2022order,aminian2023fair}), human perceptions of hiring algorithms (e.g., \citet{fumagalli2022ok,zhang2022examining}), as well as evaluating bias propagated or incited by such algorithms (e.g., \citet{cowgill2019bias,raghavan2020mitigating,li2020hiring,baek2023feedback,komiyama2024statistical,gaebler2024auditing,kim2025fair})

\cite{farajollahzadeh2025rooney} study interview-stage diversity interventions using a two-stage hiring model.
Our model has parallels to their two-stage framework of hiring, but their focus is on diversity and statistical discrimination, whereas we center on inter-firm competition and strategic behavior. \cite{castera2022statistical} study discrimination arising from differential correlation of firm's priority scores across groups, where this correlation can arise due to algorithmic monoculture.  \cite{parasurama2025algorithmic} show that a higher correlation between the algorithm's and the human hiring manager's screening criteria leads to lower diversity in hiring outcomes.

\section{Model}
\label{sec:model}

Consider a labor market with $N \geq 2$ firms hiring from a shared pool of applicants. We assume a continuum of applicants with total mass 1. All firms utilize a common hiring platform that evaluates applicants and assigns a score $s \in [0, 1]$ to each. These scores follow a continuous distribution $\cD$ on $[0, 1]$ with probability density function $\varphi(s)$. We assume $\varphi(s) > 0$ for all $s \in [0,1]$.

\paragraph{Firm Strategies.}
Firms make interview decisions based solely on applicant scores. Firm $i$ chooses a strategy function $f_i: [0,1] \rightarrow \{0,1\}$, where $f_i(s) = 1$ indicates that the firm interviews applicants with score $s$.
Each firm faces a capacity constraint $c \in (0, 1]$, such that the total mass of applicants interviewed cannot exceed $c$. Formally, a strategy $f_i$ is feasible if:
\[ \int_{0}^{1} f_i(s) \varphi(s) \, ds \leq c. \]
We assume the set of interviewed applicants, $\{s \in [0, 1]: f_i(s) = 1\}$, consists of a finite union of intervals. Let $\mathbf{f} = (f_1, \dots, f_N)$ denote a \textit{strategy profile}, and let $\mathcal{F}$ be the set of all feasible strategy profiles. We focus on pure strategies; we show in Appendix~\ref{appendix:B} that, even when mixed strategies are allowed, all Nash equilibria consist of pure strategies.

We define $M(s; \mathbf{f}) = \sum_{i=1}^N f_i(s)$ as the congestion level---the total number of firms interviewing an applicant with score $s$ under profile $\mathbf{f}$. We let $\Mmax(\mathbf{f}) = \max_{s\in [0,1]} M(s; \mathbf{f})$ denote the maximum competition for any candidate.

\paragraph{Interview and Hiring Process.}
The score $s$ represents the probability that an applicant passes an interview and receives an offer. We assume this probability is intrinsic to the applicant and identical across all firms. If an applicant receives multiple offers, they accept one uniformly at random.
We analyze two decision schemes, $\theta \in \{\textsc{corr}, \indep\}$, describing the correlation of interview outcomes across firms:

\begin{enumerate}
    \item \textbf{Correlated} ($\theta = \textsc{corr}$): Interview outcomes are perfectly correlated. An applicant with score $s$ interviewed by $n$ firms passes all interviews (receiving $n$ offers) with probability $s$, and fails all with probability $1-s$.
    \item \textbf{Independent} ($\theta = \indep$): Interview outcomes are independent. Each firm interviewing an applicant with score $s$ independently extends an offer with probability $s$.
\end{enumerate}

In practice, whether the decision scheme is correlated or independent (or somewhere in between) can depend on type of job and the interview process. For example, if the interview is very similar across all firms, then the results of the interview could be highly correlated for the same applicant across firms. On the other hand, if each firm uses a unique evaluation method, then the hiring decisions may be independent across firms.  

An instance of the model $\cI = (N, c, \cD, \theta)$ is specified by the number of firms $N$, the capacity $c$, the score distribution $\cD$, and the hiring decision scheme $\theta \in \{\textsc{corr}, \indep\}$.

\paragraph{Utility and social welfare.}
A firm's utility is the expected mass of applicants that they successfully hire.
The firm derives the same utility from any successful hire regardless of the applicant's score $s$; the score only affects the likelihood of an offer.
Let $U_n(s)$ denote the expected utility a firm derives from interviewing an applicant with score $s$, given that $n$ firms total are interviewing that applicant. This represents the probability that the firm successfully hires the applicant.
The form of $U_n(s)$ depends on the decision scheme $\theta$. When $\theta = \textsc{corr}$, $U_n(s) = s/n$. When $\theta = \indep$, $U_n(s) = (1-(1-s)^n)/n$.
Figure~\ref{fig:utility_lines} illustrates these functions for $n = 1, \dots, 6$. 

Some of our results rely on the decision scheme $\theta$, while others only rely on an assumption on the utility functions stated below.
\begin{assumption} \label{assump:utility_functions}
For all $n \in [N]$, $U_n(s)$ is continuous and strictly increasing in $s \in [0, 1]$, and $U_n(0) = 0$.
Moreover, for $n < n'$, we have $U_n(s) > U_{n'}(s)$ for all $s \in [0, 1]$.
\end{assumption}

\begin{figure}[ht]
    \centering

    \subfigure[$\theta =$\textsc{corr}]{
        \includegraphics[width=0.47\linewidth]{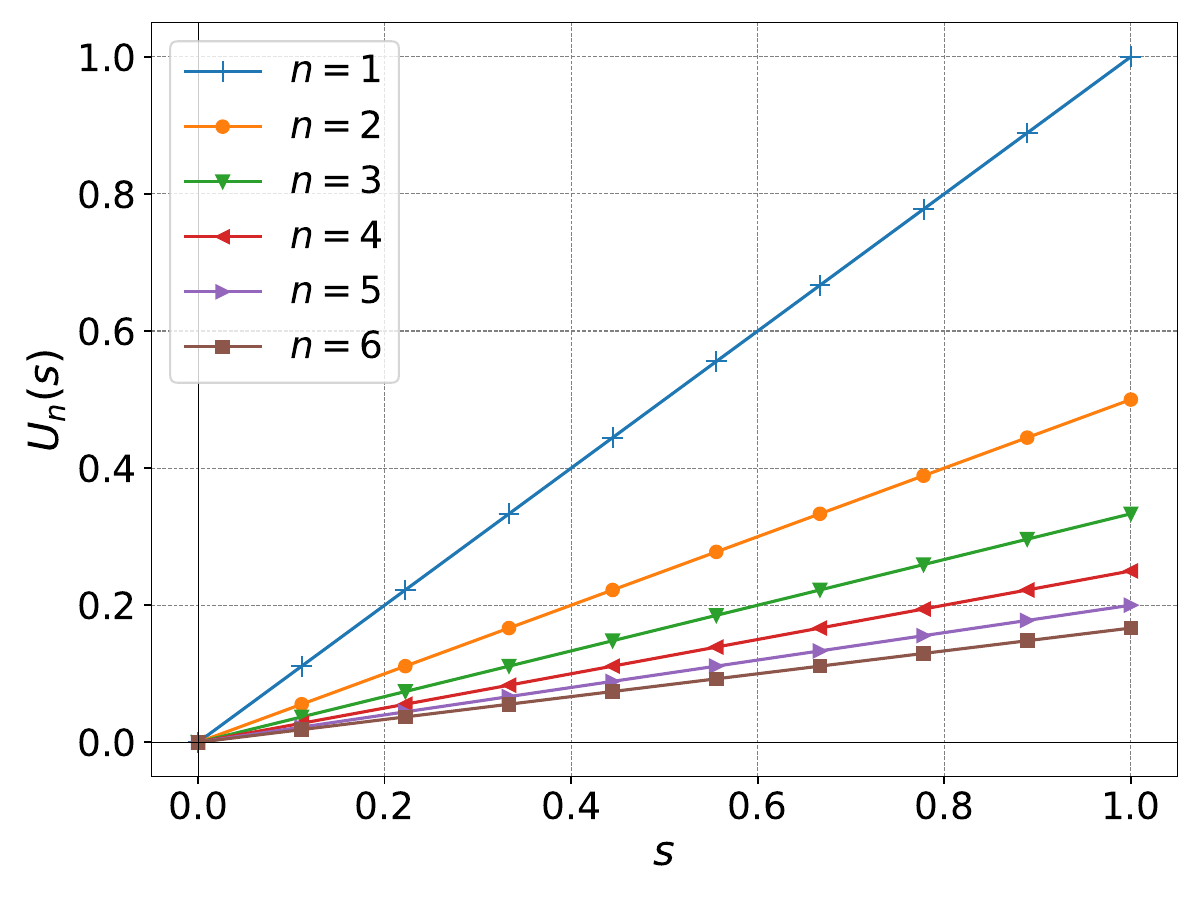}
        \label{fig:utility_correlated}
    }
    \hfill
    \subfigure[$\theta =$\textsc{indep}]{
        \includegraphics[width=0.47\linewidth]{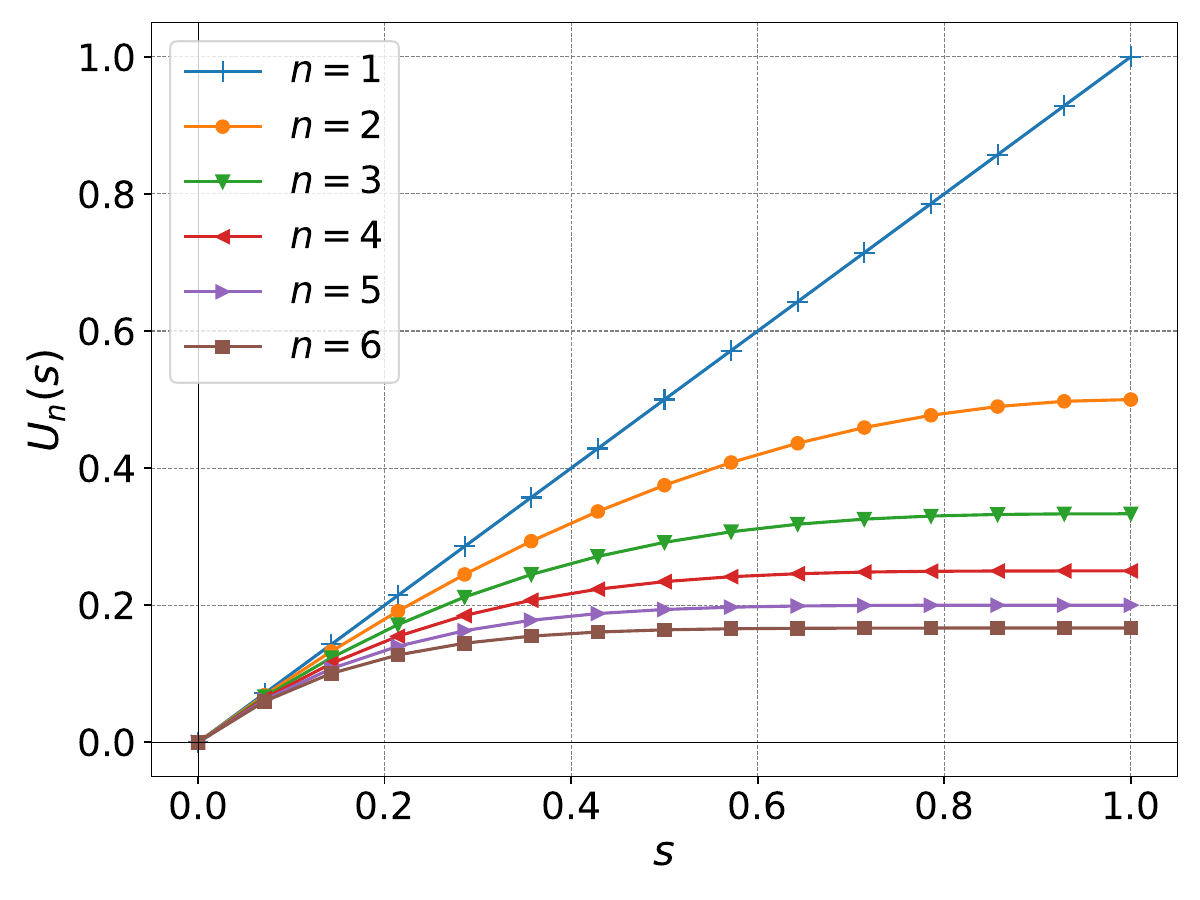}
        \label{fig:utility_indep}
    }

    \caption{Plots of the utility curves $U_n(s)$ when 
    $\theta \in \{\mathrm{corr}, \mathrm{indep}\}$ and $n = 1, \dots, 6$.}
    \label{fig:utility_lines}
\end{figure}

Then, given firm $i$'s action, $f_i$, and the actions of all other firms, $f_{-i}$, the utility for firm $i$ is
\begin{align*}
    u(f_i, f_{-i}) = \int_{0}^1 f_i(s) U_{M(s, \mathbf{f})}(s) \varphi(s)ds.
\end{align*} 
We define the social welfare to be the sum of the utilities for all firms:
\begin{align*}
    \SW(f_1, \dots, f_N) = \sum_{i=1}^N u(f_i, f_{-i}).
\end{align*}
This term also represents the social welfare on the applicant's side---since the firm's utility is its mass of successful hires, $\SW(f_1, \dots, f_N)$ equals the mass of applicants who get a job.

\paragraph{Solution concepts.}
For a problem instance $\cI$, we consider three different solution concepts that vary in how the strategy profiles $(f_i)_{i=1}^N$ are chosen.

\begin{enumerate}
\item  \textbf{Naive}: 
All firms non-strategically interview the highest-scoring applicants up to their capacity. Specifically, let $s_c \in [0, 1]$ be the threshold such that $\int_{s_c}^1 \varphi(s) ds = c$, and let $f^{\text{Naive}}(s) = 1$ for all $s \geq s_c$, and $f^{\text{Naive}}(s) = 0$ otherwise.
The social welfare is $\SW_\text{naive}(\cI) = \SW(f^{\text{Naive}}, \dots, f^{\text{Naive}})$.
\item \textbf{Nash equilibrium (NE)}:  A strategy profile $(f_1, \dots, f_N)$ is a Nash equilibrium if no firm can unilaterally deviate to another strategy to strictly increase their utility. 
Specifically, for every $i$ and for any strategy $f'$ that satisfies the capacity constraint,
\begin{align*}
   u(f', f_{-i}) \leq u(f_i, f_{-i}).
\end{align*}
Let $\mathcal{F}_{\text{NE}} \subseteq \mathcal{F}$ be the set of all strategy profiles that is a Nash equilibrium.
We define $\SW_\text{NE}$ to be the smallest social welfare from a Nash equilibrium solution:
\begin{align*}
    \SW_\text{NE}(\cI) = \inf_{(f_1, \dots, f_N) \in \mathcal{F}_{\text{NE}}}  \SW(f_1, \dots, f_N).
\end{align*}
\item \textbf{Centralized}: 
The centralized solution is one where all firms' decisions are controlled by a centralized decision maker who chooses the outcome that maximizes social welfare. We denote social welfare under this as
\begin{align*}
    \SW_\text{max}(\cI) = \sup_{(f_1, \dots, f_N) \in \mathcal{F}} \SW(f_1, \dots, f_N ).
\end{align*}
\end{enumerate}

We define the \textit{Price of Naive Selection (PoNS)} as the ratio between the social welfare under the Nash equilibrium solution to the Naive solution:
\begin{align*}
    \PoNS(\cI) = \frac{\SW_\text{NE}(\cI)}{\SW_\text{naive}(\cI)}.
\end{align*}
We define the \textit{Price of Anarchy} as the ratio between the social welfare under the Centralized solution to the Nash equilibrium.
\begin{align*}
    \PoA(\cI) = \frac{\SW_\text{max}(\cI)}{\SW_\text{NE}(\cI)}.
\end{align*}

We discuss several of the modeling assumptions and their limitations in \cref{sec:discussion_and_conclusion}.

\section{Characterizing the Nash Equilibrium}

In this section, we characterize structure of Nash equilibria, which we also show always exists.
We first explain this characterization before rigorously formalizing the result in \cref{thm:equilibrium}.

\paragraph{Thresholds.}
A key concept in our characterization are non-decreasing thresholds, 
$0 \leq \tau_1 \leq \tau_2 \leq \dots \leq \tau_N \leq 1$,
where applicants with score $s \in [\tau_i, \tau_{i+1})$ are interviewed by exactly $i$ firms.
We show that every Nash equilibrium is associated with a set of thresholds $(\tau_i)_{i=1}^N$, which must satisfy a specific set of conditions.
We differentiate between two classes of equilibria, dubbed ``equal-utility'' and ``variable-utility'', based on two different classes of conditions that the thresholds satisfy.

\paragraph{Equal-utility equilibria.}
These equilibria are those in which the thresholds satisfy $U_i(\tau_i) = U_j(\tau_j)$ for $\tau_i, \tau_j \in (0, 1)$.
That is, across the thresholds, the utility for a firm that is interviewing an applicant at the threshold is equal.
Graphically, an equal-utility equilibrium is characterized by a \textit{horizontal line} in the plot of the utility curves $U_n(s)$, where $\tau_i$ is the intersection of the horizontal line with the curve $U_i(s)$.
We explain the significance of the horizontal lines through an example.

\begin{figure}[ht]
    \centering

    \subfigure[Capacity of each firm is 0.2.]{
        \includegraphics[width=0.47\linewidth]{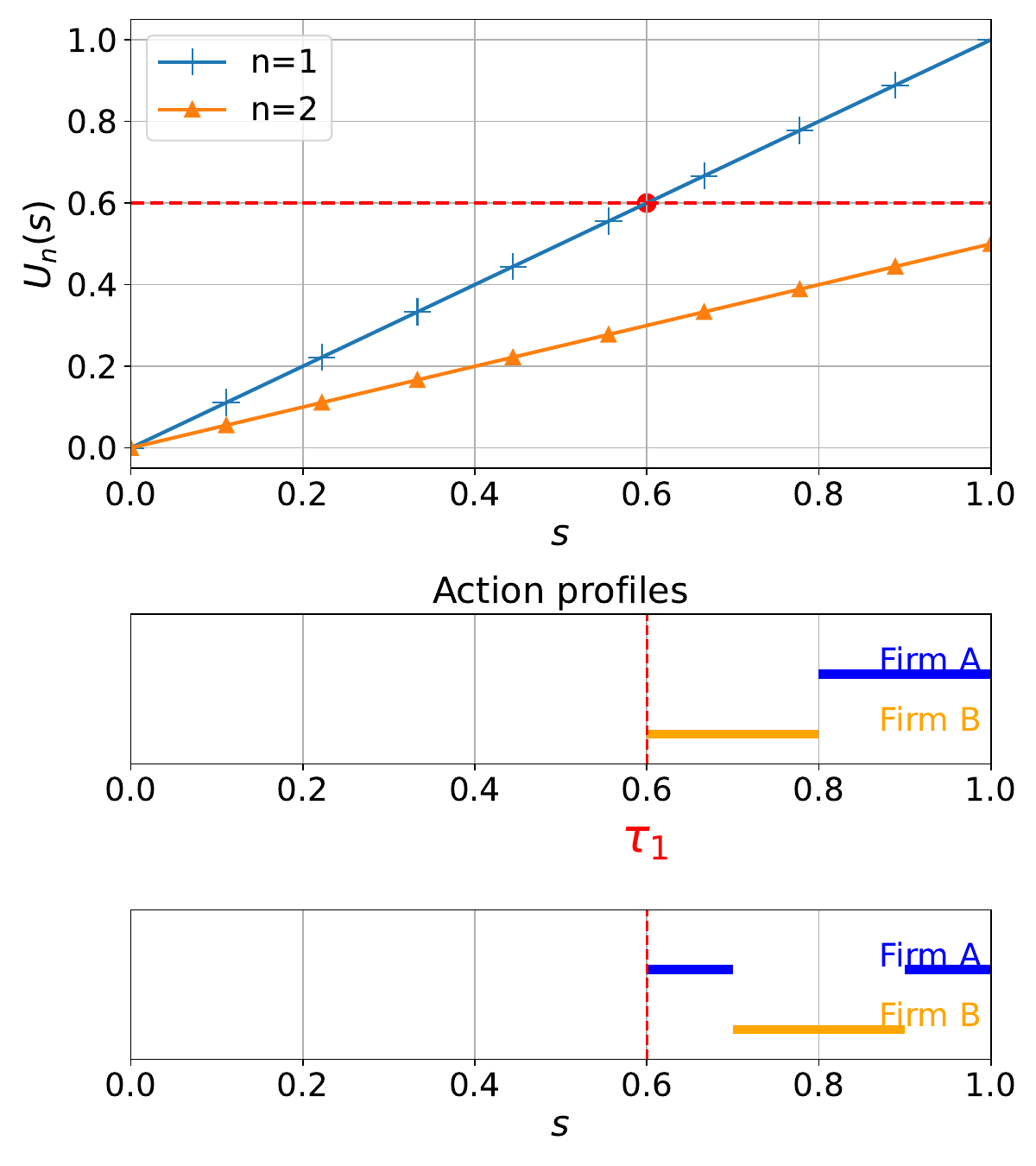}
        \label{fig:ex1}
    }
    \hfill
    \subfigure[Capacity of each firm is 0.35.]{
        \includegraphics[width=0.47\linewidth]{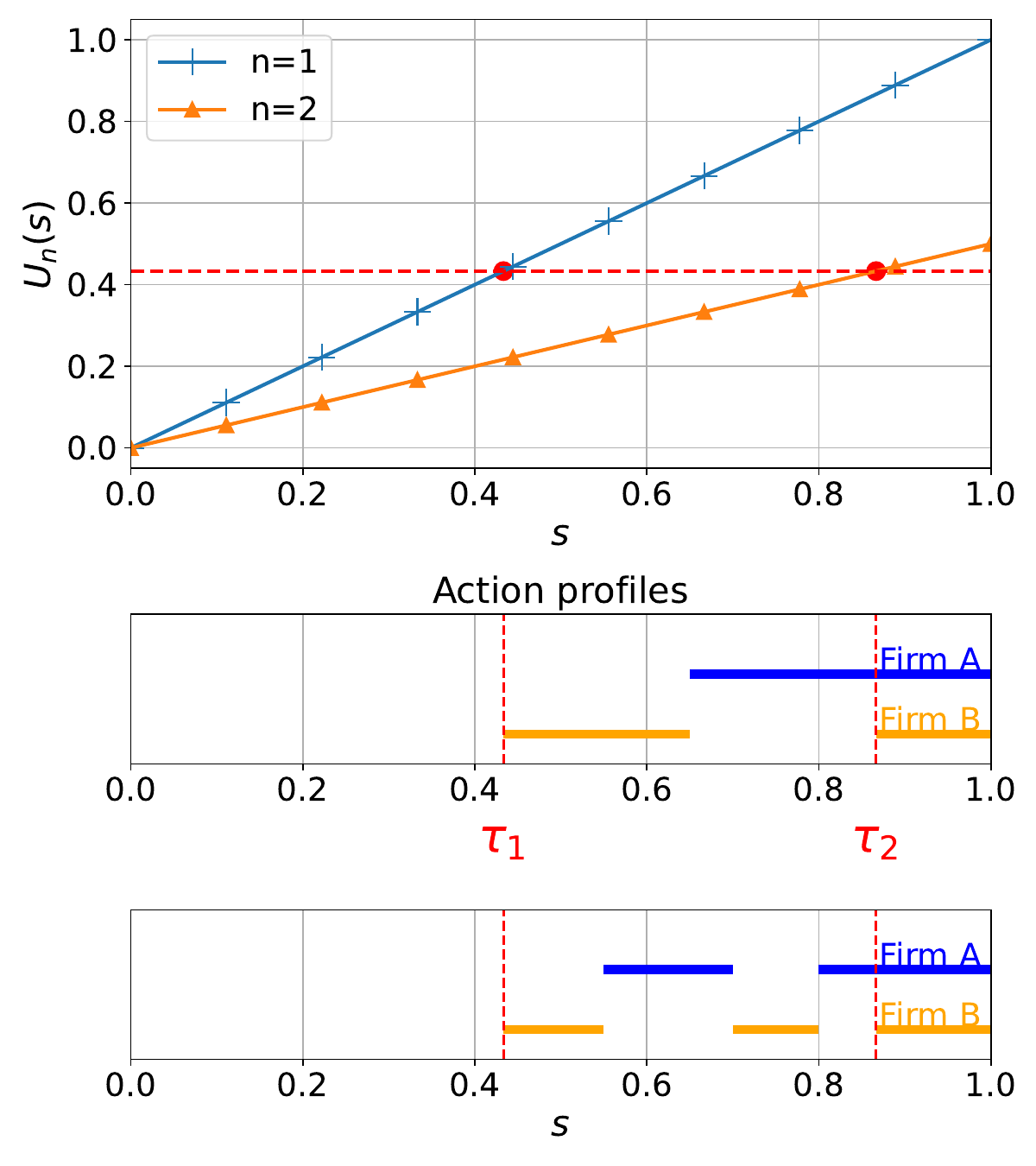}
        \label{fig:ex2}
    }

    \caption{Examples of equal-utility equilibria when $N = 2$ and $\theta = $\textsc{corr}. An equilibrium is characterized by the dashed red horizontal line in the upper plots. The lower plots depict two possible strategy profiles that form an NE with the corresponding thresholds. In Figure~\ref{fig:ex1}, $c=0.2$ and the equilibrium is characterized by the horizontal line $y=0.6$ with $\tau_1 = 0.6$. In Figure~\ref{fig:ex2}, the capacity is $c = 0.35$ and an equal-utility equilibrium is characterized by $y=13/30$ with $\tau_1 = 13/30$ and $\tau_2 = 13/15$.}
    \label{fig:side_by_side}
\end{figure}

Consider an instance with two firms A and B,
each with capacity $c = 0.2$, and suppose applicant scores are uniformly distributed from 0 to 1.
There are infinitely many equilibria for this instance, and we illustrate two of them in Figure \ref{fig:ex1}, under ``Action profiles''.
Specifically, any strategy profile in which all applicants with score $s \geq 0.6$ are interviewed by exactly one firm is an equilibrium. 
This is characterized by a horizontal line at $y = 0.6$, and the corresponding thresholds are $\tau_1 = 0.6$ and $\tau_2 = 1$.
Note that each firm derives a utility of at least 0.6 for each applicant interviewed.
There is no incentive for both firms to interview the same applicant (``double-interviewing''), as the highest utility that a firm can derive from double-interviewing an applicant is 0.5.  
Moreover, there is also no incentive for a firm to deviate to interviewing an applicant with score less than 0.6, since that also results in a lower utility.

Next, if the capacities of the firms are increased to $c = 0.35$, then there is a gain in double-interviewing the applicants whose score is close to 1, compared to single-interviewing an applicant whose score is less than 0.5. 
In this case, equilibria correspond to a horizontal line at $y = 13/30$, corresponding to thresholds $\tau_1 = 13/30$ and $\tau_2 = 13/15$ (see Figure \ref{fig:ex2}).

\begin{figure}[ht]
    \centering
        \includegraphics[width=0.47\linewidth]{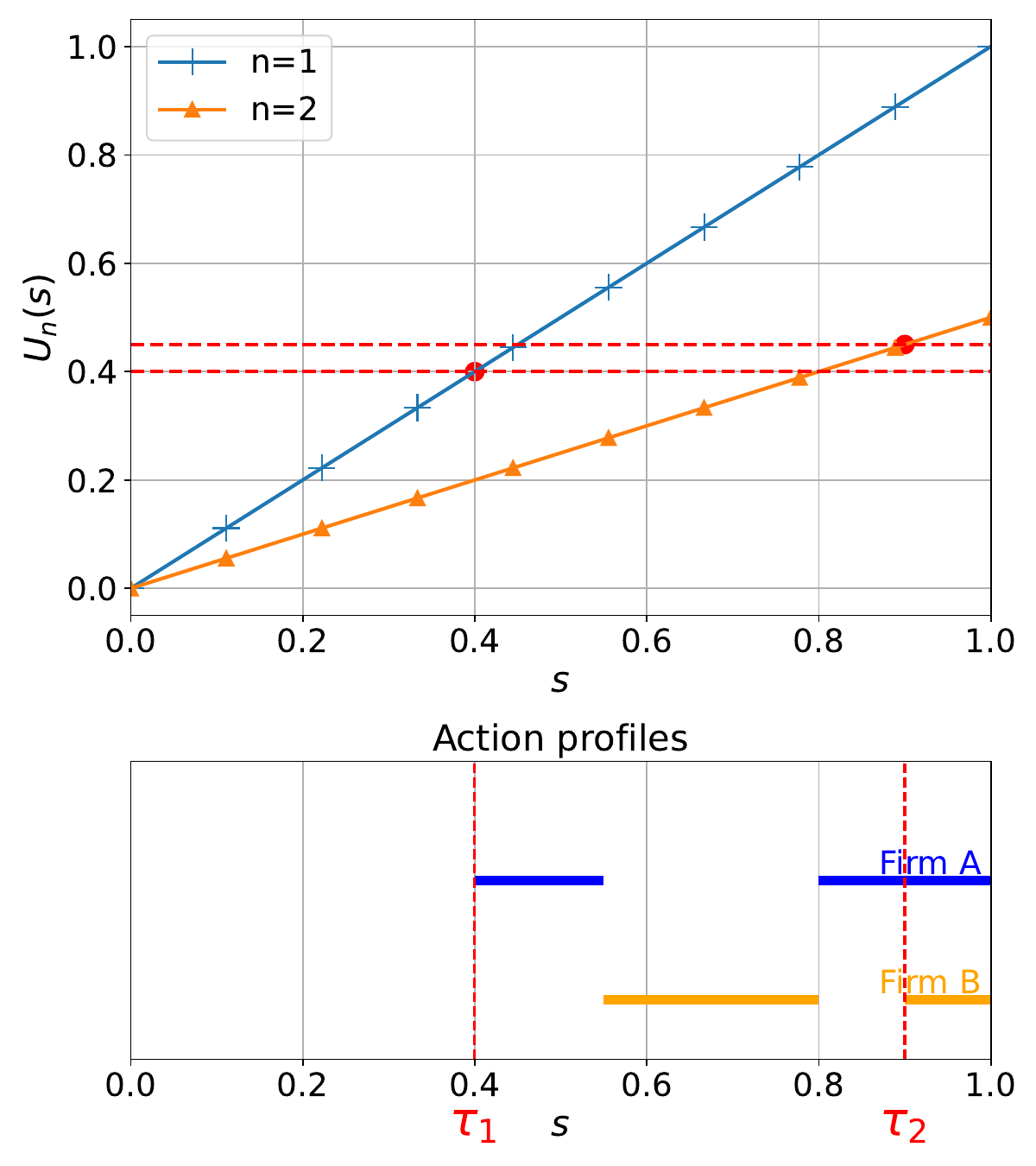}
    \caption{
    Example of a variable-utility equilibrium with $N = 2$, $\theta = \textsc{corr}$, and $c = 0.35$.
    $f_A(s) = 1$ if and only if $s\in[0.4, 0.55]\cup [0.8,1]$, and $f_B(s) = 1$ if and only if $s\in[0.55, 0.8]\cup [0.9,1]$. 
    The dashed horizontal lines highlight the utility at the two thresholds $\tau_1$ and $\tau_2$. Under this strategy profile, $\tau_1 = 0.4$ and $\tau_2 = 0.9$, and the utility at the thresholds are $U_1(\tau_1) = 0.4$ and $U_2(\tau_2) = 0.45$.}
    \label{fig:ex3}
\end{figure}

\paragraph{Variable-utility equilibria.}
Not all equilibria have thresholds described by a horizontal line. 
Using the same instance with the capacities $c = 0.35$ (same as Figure \ref{fig:ex2}), Figure \ref{fig:ex3} provides an example of strategy profiles that constitute an equilibrium, 
but the utility at the thresholds are not equal; $U_1(\tau_1) \neq U_2(\tau_2)$.

In this example, all applicants with score ranging from 0.4 to 0.9 are interviewed by \textit{one} firm. This means that double-interviewing an applicant of score 0.85 yields higher utility than single-interviewing an applicant of score 0.4. 
However, in this example, firm A is interviewing applicants from 0.4 to 0.45, but they are also already interviewing those from 0.8 to 0.9. Therefore, firm A cannot make any deviations to improve their utility.
Therefore, this is an equilibrium even though the utility at the thresholds differ. 

In this example, we have $U_1(\tau_1) < U_2(\tau_2)$, and firm A interviewed applicants right above $\tau_1$ as well as right below $\tau_2$. A variable-utility equilibrium has this property in general: when the utility at the thresholds differ, it must be that all firms who interview right above the lower threshold must also have interview right below the upper threshold.

\subsection{Formal Characterization}
We formalize the above ideas by characterizing the conditions of a Nash equilibrium in the following theorem.

\begin{theorem} \label{thm:equilibrium}
Under Assumption~\ref{assump:utility_functions},
the strategy profile $\bbf = (f_1, \dots, f_N)$ is a Nash equilibrium if and only if 
there exist thresholds $0 = \tau_0 \leq \tau_1 \leq \tau_2 \leq \dots, \leq \tau_N \leq \tau_{N+1}= 1$ that satisfy the following conditions:
\begin{enumerate}
    \item  $\Pr(f_i(S)=1) = c$ for all $i \in [N]$.
    \item  $M(s, \bbf) = m$ for all $s \in [\tau_m, \tau_{m+1})$ for all $m = 0, 1,  \dots, \Mmax(\bbf)$.
    \item  $U_n(\tau_n) \leq U_m(\tau_m)$ for all $n < m \leq \Mmax(\bbf)$.
    
    \item Consider any $n < m \leq \Mmax(\bbf) +1$ where 
    $U_n(\tau_n) < U_m(\tau_m)$, and consider any firm $i$ and score $s \in [\tau_n, \tau_{n+1})$ where $f_i(s) = 1$ and $U_n(s) < U_m(\tau_m)$. 
    For any $s' \in [\tau_{m-1}, \tau_m)$ where $U_m(s') > U_n(s)$, we have that $f_i(s') = 1$.
\end{enumerate}
\end{theorem}

The first condition states that every firm will use their entire capacity.
The second condition states that exactly $m$ firms will interview applicants between score $\tau_m$ and $\tau_{m+1}$.
The third condition states that the utilities at the thresholds are non-decreasing.
The last condition states additional conditions that need to be satisfied if the utilities at the thresholds are strictly decreasing. 
Specifically, if $U_n(\tau_n) < U_m(\tau_m)$, then it must be that any firm who interviews applicants slightly above $\tau_n$ must also interview those slightly below $\tau_m$. 
We now formally define and differentiate between equal-utility and variable-utility equilibria.

\begin{definition}
\label{def:types_of_NE}
A strategy profile $\bbf = (f_1, f_2, \cdots, f_N)$ is a \textit{equal-utility Nash equilibrium} if it is a Nash equilibrium where the thresholds defined in \cref{thm:equilibrium} satisfy $U_n(\tau_n) = U_m(\tau_m)$ for all $n,m \in[\Mmax(\bbf)]$.
Any other Nash equilibrium is a \textit{variable-utility Nash equilibrium}.
\end{definition}  

Lastly, we show that a equal-utility Nash equilibrium always exists.

\begin{proposition} \label{prop:ne_equal_existence}
For any instance $\cI = (N, c, \cD, \theta)$ and utility function $U_n(s)$ that satisfies Assumption \ref{assump:utility_functions}, there exists a strategy profile $\bbf = (f_1, f_2, \cdots, f_N)$ that is an equal-utility Nash equilibrium.
\end{proposition}

\subsection{Nash Equilibrium Properties}

We prove properties regarding the number of interviews received by each applicant in the NE solution, under both $\theta = \textsc{corr}$ and $\indep$.
These properties can be visualized in Figure~\ref{fig:utility_lines_large_n}.

\begin{figure}[ht]
    \centering
    \subfigure[$\theta = \textsc{corr}$]{
        \includegraphics[width=0.47\linewidth]{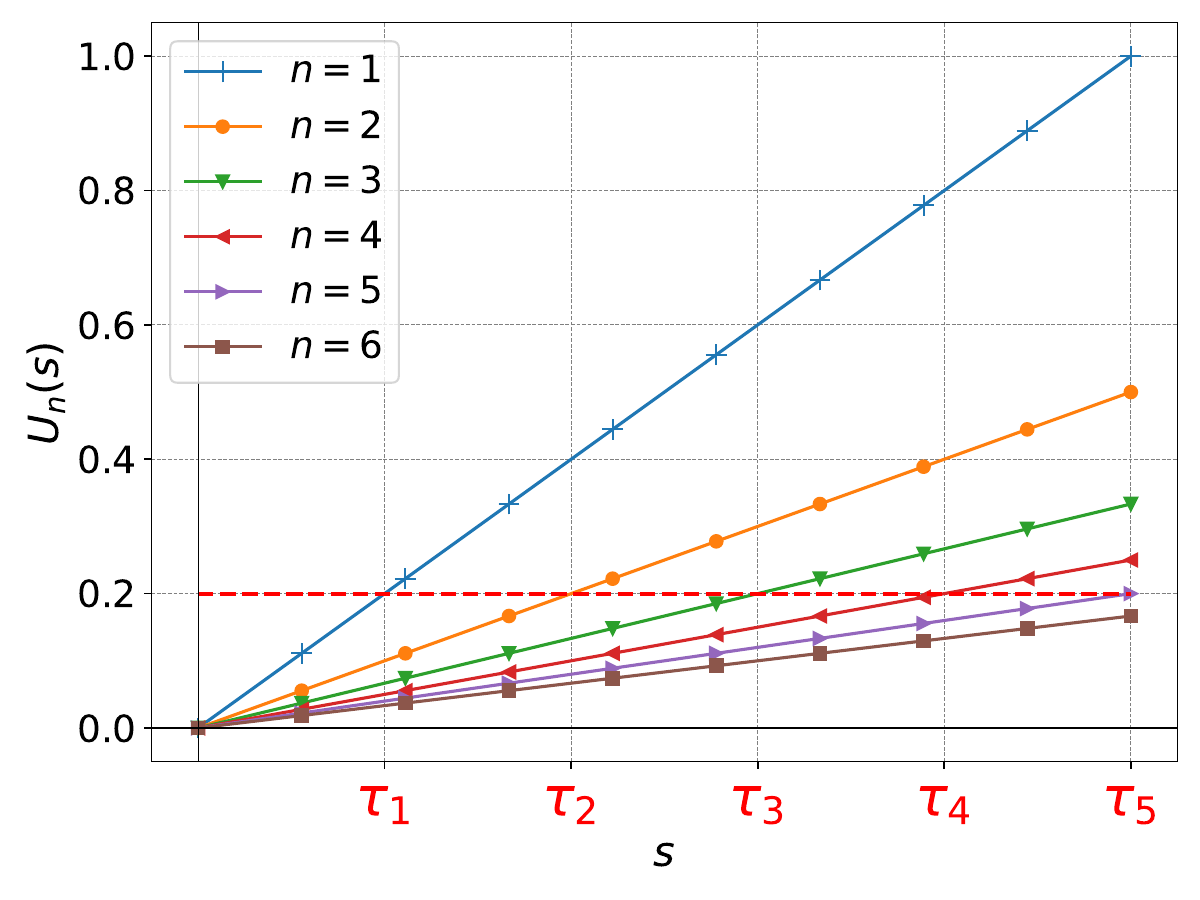}
        \label{fig:utility_correlated_limit}
    }
    \subfigure[$\theta = \indep$]{
        \includegraphics[width=0.47\linewidth]{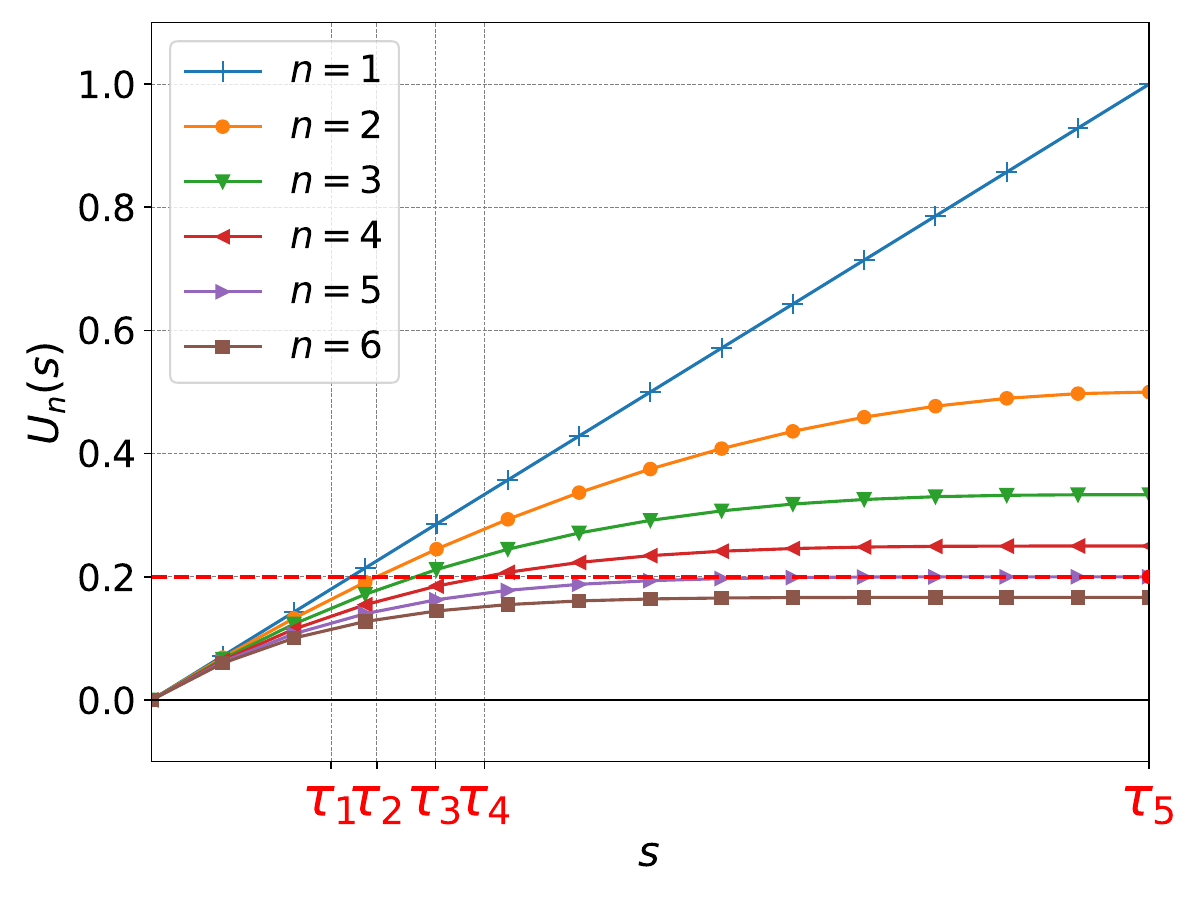}
        \label{fig:utility_indep_limit}
    }
    \caption{Comparison of utility curves  \( U_n(s) \) and the thresholds for \( n \in \{1, 2, \dots, 6\} \) under $\theta = \{\textsc{corr}, \indep\}$. The dashed horizontal line $y=0.2$ characterizes equal-utility equilibria. When $\theta = \textsc{corr}$, the thresholds are evenly spaced across the applicant score \( s \). In contrast, when $\theta = \indep$, most thresholds are concentrated near \( s = 0 \); hence a large portion of applicants who receive interviews will be interviewed by the same number of firms.}
    \label{fig:utility_lines_large_n}
\end{figure}

Notice that under the $\theta = \textsc{corr}$ scheme displayed in Figure \ref{fig:utility_correlated_limit}, the thresholds $\tau_1$ to $\tau_5$ are evenly spaced out. 
This means that the number of interviews received by each applicant increases with their score $s$ in regular intervals.
In the example in the figure, the set of scores from 0 to 1 is equally divided into five sets, where the first set represents scores that receive 0 interviews, the second set of scores receive 1 interview, and so on.

Contrastingly, under $\theta = \indep$ shown in Figure \ref{fig:utility_indep_limit}, the thresholds $\tau_1$ to $\tau_4$ are very close together, close to 0, while $\tau_5$ is at the maximum value. 
This implies that there is a large range of scores, $[0.331, 1]$, where every applicant in that range receive four interviews. 
It may seem unintuitive that an applicant with score $s = 0.5$ receives the same number of interviews as an applicant with score $s = 1$, but this is due to the  independent realizations of the offer decisions, 
Specifically, if four firms interview an applicant with score $s = 0.5$, on average only two of the firms will give them an offer, and therefore those two firms are only competing with each other to hire that applicant. 
However, if four firms interview an applicant with score $s = 1$, they will all give an offer, and therefore a firm competes with three other firms for that applicant.
In general, under the $\theta = \indep$ scheme, a large portion of applicants who receive interviews will be interviewed by the same number of firms.

We formalize these properties described in the following two propositions.
First, under $\theta = \textsc{corr}$, we show that under an equal-utility NE, every threshold $\tau_m$ will be an exact multiple of $\tau_1$. 

\begin{proposition}
    \label{prop:NE_shared}
    Fix a distribution $\cD$ and capacity $c \in (0, 1]$, and let $\cI(N) = (N, c, \cD, \textsc{corr})$ be the instance with $N$ firms and the correlated decision rule. Let $\mathbf{f_N}$ be a Nash equilibrium for instance $\cI(N)$. Then $\tau_m \geq m \tau_1$ for all $m\in [\Mmax(\mathbf{f_N})]$. If $\mathbf{f_N}$ is an equal-utility Nash equilibrium, then $\tau_m = m \tau_1$ for all $i\in [\Mmax(\mathbf{f_N})]$.
\end{proposition}

Next, under $\theta = \indep$, we show that as $N \to \infty$, the number of firms that interview an applicant will differ by at most 1. That is, every applicant will be interviewed either $\Mmax(\mathbf{f_N})$ times or $\Mmax(\mathbf{f_N}) - 1$ times.

\begin{proposition}
\label{prop:NE_indep}
Fix a distribution $\cD$ and capacity $c \in (0, 1]$, and let $\cI(N) = (N, c, \cD, \indep)$ be the instance with $N$ firms and the independent decision rule.
Let $\mathbf{f_N}$ be a Nash equilibrium for instance $\cI(N)$.
Then, $\lim_{N \to \infty} \mathbb{P}_{S\sim D}(|M(S; \mathbf{f_N}) - \Mmax(\mathbf{f_N})|>1) = 0$.
\end{proposition}

\section{Social Welfare}
In this section, we evaluate how the social welfare compares under the three different solution concepts: Naive, Nash equilibrium, and Centralized. 
We show that the Price of Naive Selection can grow with $N$ when the capacity is small, while the Price of Anarchy goes to 1 in the same regime.

We first establish that $\SW_{\text{naive}} < \SW_\text{NE} \leq \SW_\text{max}$.
The second inequality holds by definition of the Centralized solution, and we show that the first inequality holds under general conditions.
\begin{theorem}\label{thm:SW1}
Under Assumption~\ref{assump:utility_functions}, $\SW_{\text{NE}} > \SW_\text{naive}$.
\end{theorem}

\subsection{Price of Naive Selection}

We compare the social welfare across the Naive and Nash equilibrium solutions using the Price of Naive Selection (PoNS).
We first analyze the PoNS when the capacity $c$ is at its extremes. 
Keeping $N$ fixed, we show that the PoNS goes to $N$ when $c  \to 0$, and the PoNS goes to 1 when $c \to 1$.

\begin{theorem} \label{thm:SW_PoNS2}
For any distribution $\cD$, number of firms $N$, and decision rule $\theta \in \{\indep, \textsc{corr}\}$,
if $\cI(c) = (N, c, \cD, \theta)$ is the instance parameterized by capacity $c \in (0, 1)$,
$\lim_{c \to 0^+} \PoNS(\cI(c)) = N$ and $\lim_{c \to 1^-} \PoNS(\cI(c)) = 1$.
\end{theorem} 

\cref{thm:SW_PoNS2} establishes that there is substantial inefficiency of firms using the naive strategies when each firm has a small interview capacity.
This is more pronounced with a large number of firms, as the PoNS grows with $N$ when $c \to 0$.
Conversely, as $c \to 1$, the PoNS converges to 1, indicating that strategic behavior offers little advantage when firms can interview almost every applicant. 

This implies that the importance of strategic selection is thus most pronounced when capacities are low and the number of firms is large.
In this regime, firms must allocate their limited capacity carefully.
Naive strategies, which prioritize top scores uniformly, result in excessive competition for a small subset of candidates, leaving many applicants underutilized. Strategic selection diversifies this allocation, significantly improving social welfare.

Next, we consider a regime where the capacity $c$ is fixed, and the number of firms $N$ goes to infinity.
In this case, we derive two results for each of the two hiring schemes.
First, under the independent decision rule, we show that the PoNS goes to $\frac{1}{c}$ when $N \to \infty$.
\begin{theorem}\label{thm:SW_PoNS4}
    For any distribution $\cD$, capacity $c$, if $\cI(N) = (N, c, \cD, \indep)$ is the instance with $N$ firms and the independent decision rule, 
    then $\lim_{N \to \infty} \PoNS(\cI(N)) = \frac{1}{c}$
\end{theorem}

Next, under the correlated decision rule, the PoNS goes to $\frac{1}{c(2 - c)}$ when $\mathcal{D}$ is the uniform distribution. We generalize the result for any general distribution $\mathcal{D}$ in Lemma \ref{lemma:SW_PONS3}. %

\begin{theorem}\label{thm:SW_PoNS3}
    Let $\cD_0$ be the uniform distribution on $[0, 1]$.
    Let $\cI(N) = (N, c, \cD_0, \textsc{corr})$ be the instance with $N$ firms and the correlated decision rule. 
    Then, $\PoNS(\cI(N))$ increases with $N$, and $\lim_{N \to \infty} \PoNS(\cI(N)) =\frac{1}{c(2 - c)}$. In addition, if $Nc\leq 0.5$, $\PoNS(\cI(N))\geq 1.5$.
\end{theorem}

The results of \cref{thm:SW_PoNS4} and \cref{thm:SW_PoNS3} establish how the PoNS changes for intermediate values of $c$.
Specifically, as $c$ decreases, the PoNS increases at a rate of $O(1/c)$.
Next, comparing the two decision rules, when $N \to \infty$, for any fixed capacity, the PoNS is higher under the independent decision rule compared to the correlated decision rule. 
Hence, independence in interview outcomes exacerbates the inefficiencies in naive strategies and makes strategic selection more important.

We numerically examine the PoNS for finite values of $c$ and  $N$ under a uniform score distribution.
Figure~\ref{fig:PoNS_limit} plots the PoNS as a function of the capacity $c \in (0, 1)$ and number of firms $N \in \{1, \dots, 100\}$, for both $\theta\in\{\mathrm{CORR},\mathrm{INDEP}\}$.
The results show clear patterns consistent with our theoretical limits: for fixed 
$N$, the PoNS increases sharply as capacity decreases, and for fixed $c$, the PoNS grows with the number of firms before leveling off near its asymptotic value.
Therefore, the PoNS remains significant for reasonable values of $N$ and $c$.

\subsection{Price of Anarchy}
Next, we examine the \textit{Price of Anarchy (PoA)}, which measures the inefficiency of strategic and de-centralized decisions, compared to the socially optimal outcome that can be achieved by a central decision-maker.

\begin{theorem}
\label{thm:PoA}
    For any distribution $\cD$ and decision rule $\theta \in \{\indep, \textsc{corr}\}$,
    let $\cI(c, N) = (N, c, \cD, \theta)$ be the instance parameterized by the capacity $c$ and number of firms $N$.
    For any $N$, we have $\lim_{c\rightarrow 0} \PoA(\cI(c, N)) =1$, and for any $c \in (0, 1)$, we have $\lim_{N\rightarrow \infty} \PoA(\cI(c, N)) =1$.
\end{theorem}

\cref{thm:PoA} shows that, in the regime where $N \to \infty$ or when $c \to 0$, the social welfare when firms are strategic converges to the highest possible social welfare in the system. 
In other words, there is little inefficiency stemming from individual firms acting in their own self-interest.

\begin{figure}[H]
    \centering
    \subfigure[PoNS vs Capacity for Different N]{
        \includegraphics[width=0.47\linewidth]{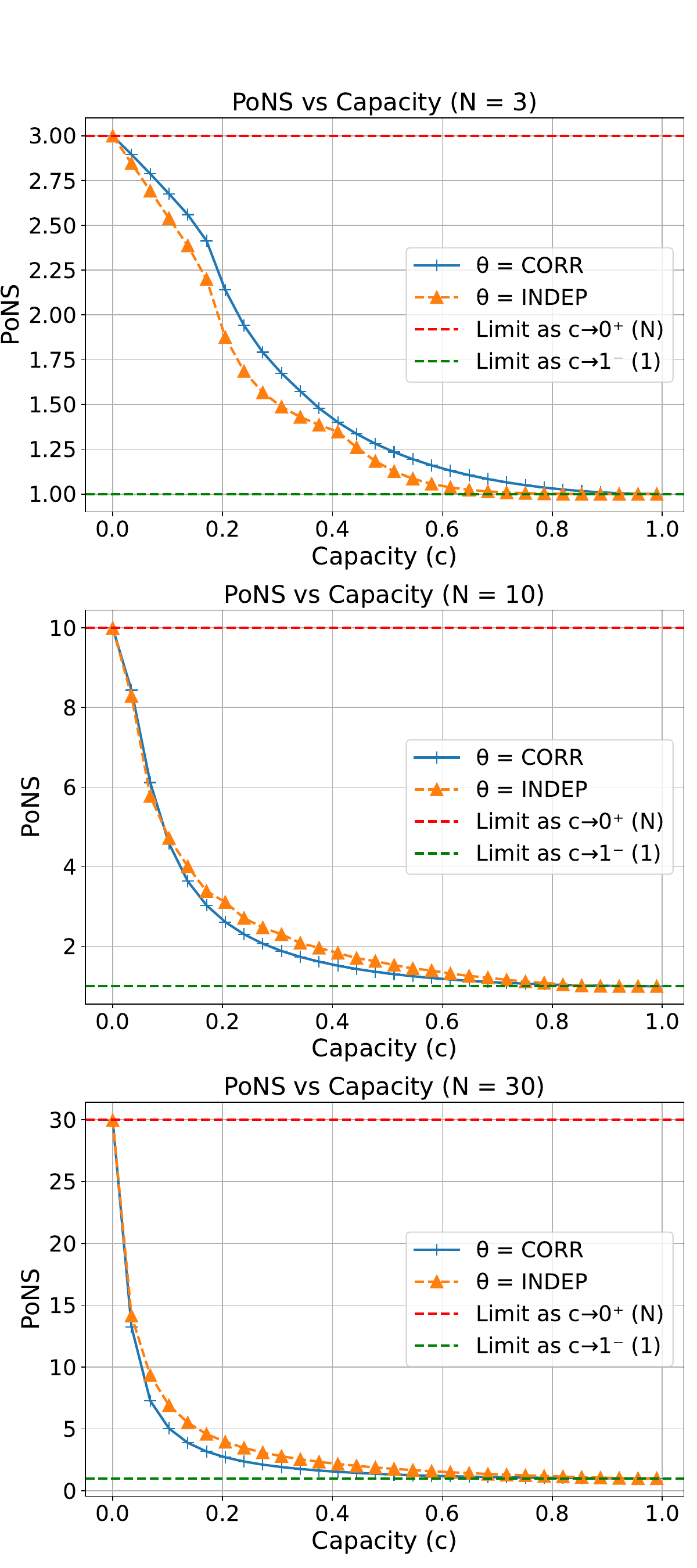}
        \label{fig:simu_pons_c}
    }
    \subfigure[PoNS vs N for Different Capacity]{
        \includegraphics[width=0.47\linewidth]{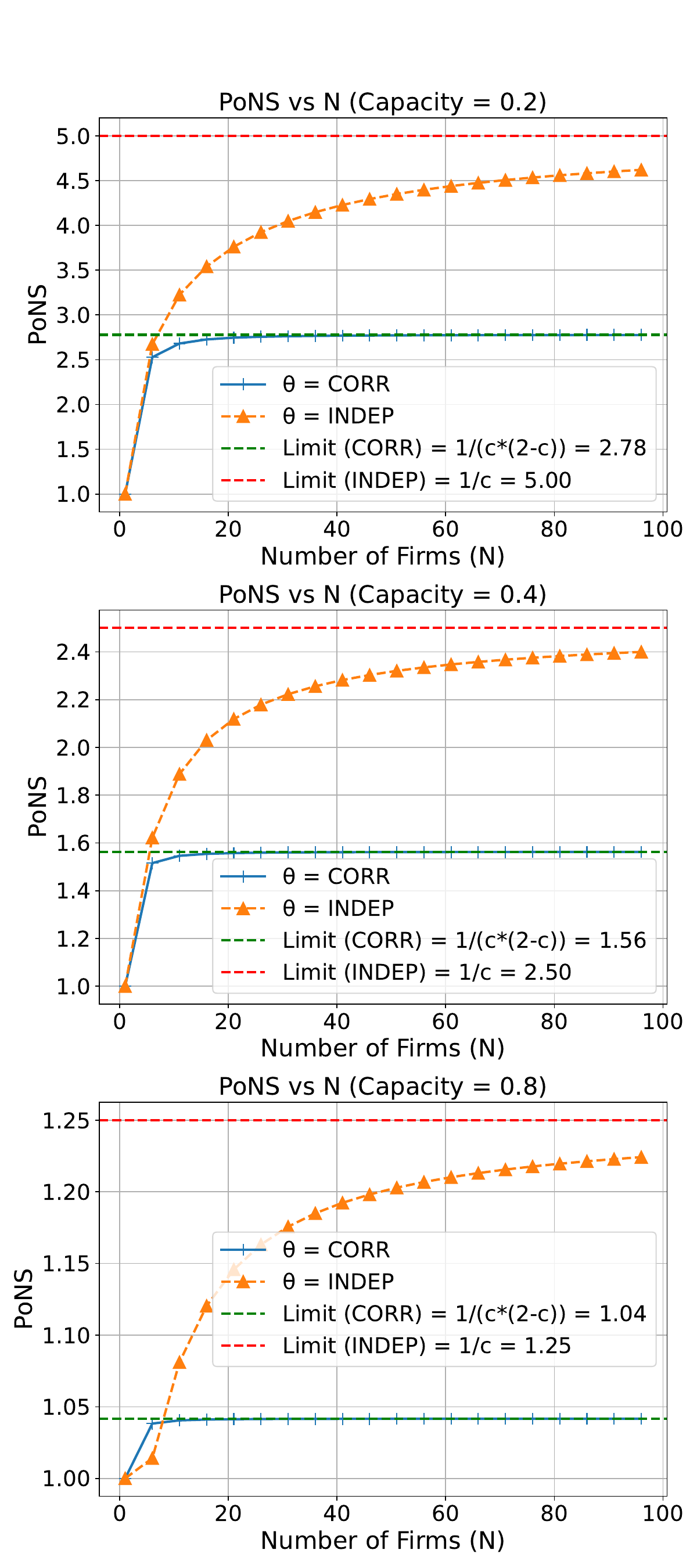}
        \label{fig:simu_pons_n}
    }
    \caption{Numerical evaluations of the PoNS when $\cD$ is uniform; the left plots fix $N$ and vary $c$, while the right plots fix $c$ and vary $N$. When $N$ is fixed, the PoNS goes to $N$ when $c\rightarrow0^+$, and goes to 0 when $c\rightarrow 1^-$ for $\theta \in \{\textsc{corr}, \indep\}$. 
    When $c$ is fixed $N\rightarrow \infty$, the PoNS converges to $\frac{1}{c}$ under $\theta = \textsc{indep}$ and $\frac{1}{c(2-c)}$ under $\theta = \corr$.
    }
    \label{fig:PoNS_limit} 
\end{figure}

\begin{figure}[H]
    \centering
    \subfigure[PoA vs Capacity for Different N]{
        \includegraphics[width=0.47\linewidth]{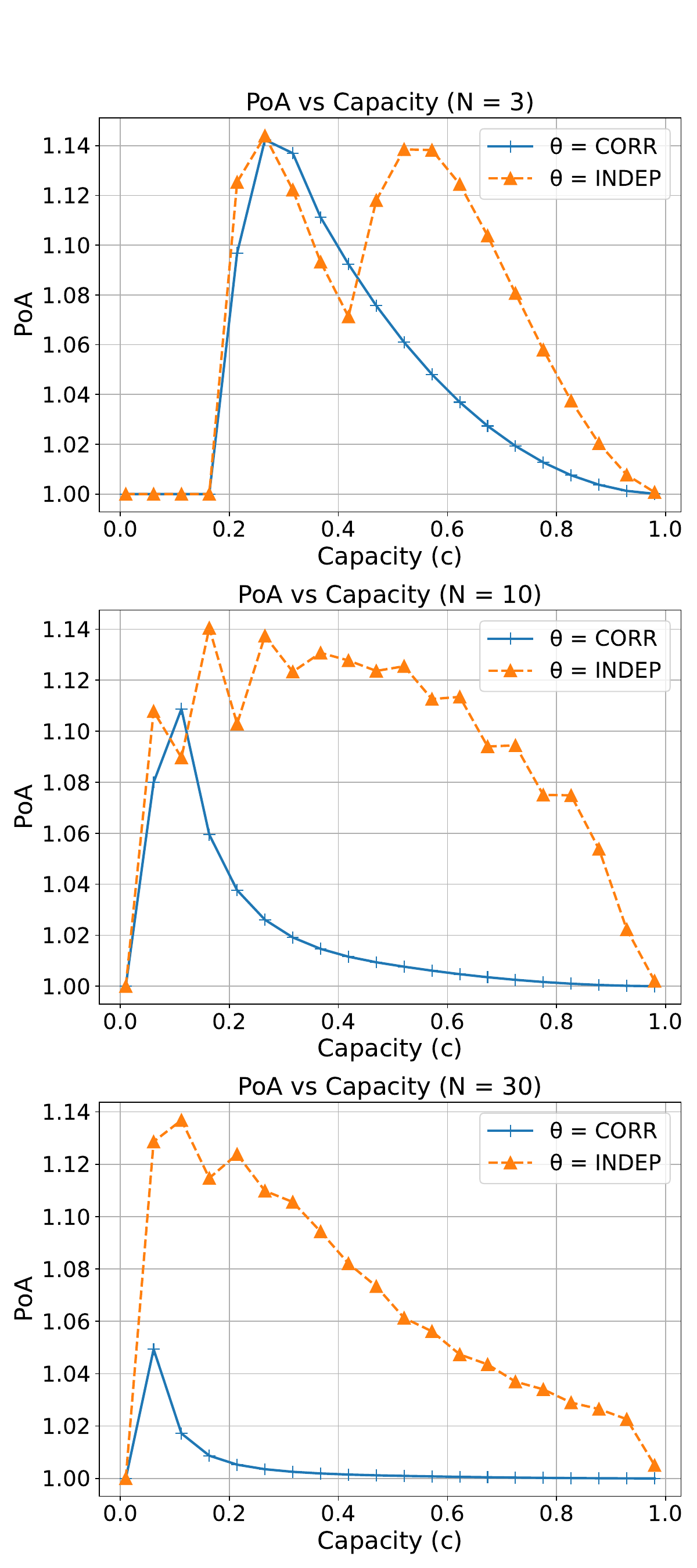}
        \label{fig:simu_poa_c}
    }
    \subfigure[PoA vs N for Different Capacities]{
        \includegraphics[width=0.47\linewidth]{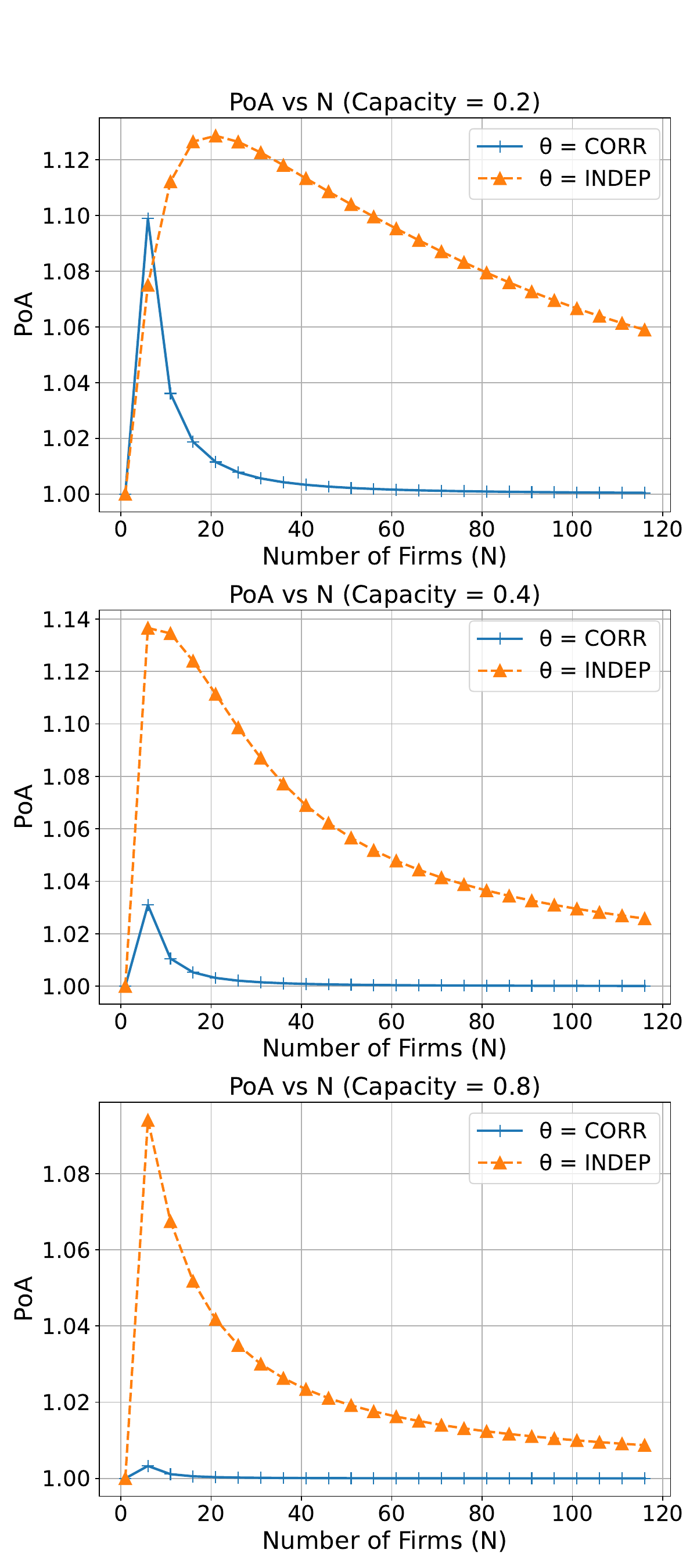}
        \label{fig:simu_poa_n}
    }
    \caption{Numerical evaluations of the PoA when $\cD$ is uniform: the left plots fix $N$ and vary $c$ while the right plots fix $c$ and vary $N$. When $N$ is fixed, the PoA goes to 1 when $c\rightarrow0^+$ for $\theta \in \{\textsc{corr}, \indep\}$. When $c$ is fixed, the PoA goes to $1$ when $N\rightarrow \infty$ for $\theta \in \{\textsc{corr}, \indep\}$.}
    \label{fig:PoA_limit} 
\end{figure}

As \cref{thm:PoA} only considers limiting values of the parameters, in Figure \ref{fig:PoA_limit}, we numerically evaluate the PoA for moderate values of $c$ and $N$ under a uniform score distribution. 
We see that the PoA is consistently very low; below 1.15 in all tested regimes. 
The non-smoothness of the curves are due to the fact that the social welfare of the Nash equilibrium solution increases continuously as the capacity increases but exhibits kinks at certain values of $c$. Overall, the numerical results demonstrate that the PoA remains low across finite parameter values.

\section{Equilibrium Convergence} \label{sec:eq_convergence}
The previous section establishes that the social welfare is substantially higher under the Nash equilibrium than under the Naive solution.
However, the relevance of this finding depends on whether firms can reasonably arrive at an NE.
To investigate this, we study whether best response dynamics converge to an NE.

Best response dynamics is a process in which firms iteratively choose their strategies to their current best responses, given the strategies of the other firms.
If this iterative process converges to a point where no firm can improve its payoff by a unilateral deviation, then this corresponds to an NE. However, best response dynamics, in general, is not guaranteed to converge.

We show two results regarding best response dynamics. First, we show that in a discretized version of the game, best response dynamics converge to a Nash equilibrium. Second, we show that when capacities are small, we show that ``one-turn'' best response dynamics converge to an NE in the original game, where each firm only makes one move.

First, we define a modified version of the original game where the strategy space is discretized.

\begin{definition}[Discretized Game]
\label{ass:quantization}
We define the discretized game as the same game as in the continuous setting except that the set of interview strategies is restricted to a discretized partition of the applicant space. 
Let $\{I_k\}_{k=1}^K$ be a partition of $[0, 1]$ into $k$ bins such that each $I_k$ corresponds to  a mass of $1/K$ with respect to the score distribution density function $\varphi$: 
\[
\int_{I_k} \varphi(s) \, ds = \frac{1}{K} \quad \text{for all } k = 1, \dots, K.
\]
Then, firm $i$ chooses a subset of bins $S_i \subseteq \{1,\dots,K\}$ such that $|S_i| \leq L := \lfloor cK \rfloor$. 
Then, $\bigcup_{k \in S_i} I_k$ corresponds to the set of applicants interviewed by firm $i$.
\end{definition}

We show that under the discretized game, best response dynamics converge to an NE in a finite number of steps, which is established by showing that this game is a potential game.
We then show that such equilibria approximate a Nash equilibrium of the original game.

\begin{theorem}\label{thm:congestion}
    In the discretized game (Definition \ref{ass:quantization}), best response dynamics terminates in finitely many steps at a Nash equilibrium. For any $\varepsilon > 0$, there exists $K_0 \in \mathbb{N}$ such that the corresponding continuous strategy profile $\bbf$ of the discretized game is an $\varepsilon$-Nash equilibrium in the original continuous game. That is, for all firm $i$ and $K \ge K_0$, $\bbf$ satisfies:
    \begin{align*}
        u_i(f_i, f_{-i}) \ge u_i(f_i', f_{-i}) - \varepsilon,
    \end{align*}
    for all deviations $f_i' \subseteq [0,1]$ with $\int_{f_i'} \varphi(s) ds \le c$.

\end{theorem}

These results show that discretizing the strategy space yields well-behaved best response dynamics: the process always converges, and the resulting equilibria approximate those of the continuous game arbitrarily well.
However, \cref{thm:congestion} does not specify how many iterations this process will take. 
We address this in the next result where we show that a simpler set of best response dynamics where each firm only takes one turn, will converge to an NE, under mild assumptions.

\paragraph{One-turn best response.}
We define the \textit{one-turn best response} dynamics to be the following. Sequentially, each firm takes one turn to decide their action, starting with firm 1 and ending with firm $n$.
On its turn, firm $i$ selects the action that maximizes its payoff, assuming the actions of firms 1 to $i-1$ are fixed, and treating itself as the last firm to act.  In other words, firm $i$ optimizes its action under the assumption that no subsequent firm will adjust its behavior in response to its choice. This process continues until all $n$ firms have taken their turn.

\begin{theorem}
\label{thm:convergence_corr}
Suppose there is a $\delta > 0$ such that $\varphi(s) \geq \delta \; \forall s \in [0,1]$ 
    and $c \leq 0.5 \delta$.
    If $\theta = \textsc{corr}$, one-turn best response dynamics converges to an equal-utility Nash equilibrium.
\end{theorem}

\begin{theorem}\label{thm:convergence_indep}
Fix a distribution $\cD$, number of firms $N$, and $\theta = \indep$. Suppose there is a $\delta > 0$ such that $\varphi(s) \geq \delta \; \forall s \in [0,1]$. There exists a $c_0$ such that if $c \leq c_0$, then one-turn best response dynamics converges to an equal-utility Nash equilibrium. 
\end{theorem}

The above results establish that when capacities are small, an NE can be reached even when each firm only takes one turn, and they take the best response assuming they are the last to act.
Hence, 
if each firm is able to compute their best response, it is not a strong assumption to assume that an NE would arise in practice.

\subsection{Computing the Best Response} \label{sec:compute_br}
We have established that a Nash equilibrium can be reached through a sequence of best response dynamics.
However, it is not clear how a firm would compute their best response, since a firm would not have access to the exact strategies of other firms in practice.

Specifically, to compute the utility $U_n(s)$ of interviewing an applicant, the firm needs both the score $s$, and the total number of firms who selected the applicant, $n$.
The firm has access to $s$ through the algorithm, but they do not know $n$ if they do not know the strategies of the other firms.
We consider a setting where a firm can aim to \textit{estimate} the utility $U_n(s)$ using historical data on acceptance decisions.
Concretely, we assume that there are pools of applicants where both $s$ and $n$ is the same for all applicants in the same pool.
Then, for example, if a firm has made job offers to applicants from this pool in the past that were always accepted, then the firm can estimate that $n$ is small.
We calculate how many samples a firm needs to be able to accurately determine which applicants yield higher utility.

\paragraph{Comparing two pools of applicants.}
We consider a simple setting where a firm needs to determine which pool of applicants, out of two, yield higher utility.
Suppose there are two pools of applicants, where every applicant in the same pool has the same score; $s_1$ and $s_2$ respectively.
Consider firm 1 deciding which pool to interview from, and assume that all other firms' strategies are fixed, but unknown to firm 1.
We assume that each firm has to either interview everyone from a pool or no one from a pool, and hence two applicants from the same pool are interviewed by the same set of firms.
Note that the expected utility of interviewing an applicant is equal to the probability that a firm successfully hires them, given that they get an interview.
Let $p_1, p_2 \in (0, 1)$ be the probabilities that firm 1 successfully hires if they interview applicants from pool 1 and 2 respectively.
Firm 1 would like to choose the pool with higher expected utility, but these values are unknown to the firm as they do not know how many other firms are interviewing.

Suppose firm 1 can collect data by interviewing $k$ applicants from each pool, and then selects the pool with a greater number of successful hires.
Then, how large does $k$ need to be for firm 1 to correctly determine which pool yields higher utility (with high probability)?
The answer depends on the exact values of $p_1$ and $p_2$, as well as the probability guarantee.
We provide the answer for a range of these values.
For a given $p_1 < p_2$ and probability $q \in [0, 1]$, we find the smallest number of samples $k$ needed that guarantees that if $X_1 \sim \text{Binom}(k, p_1)$ and $X_2 \sim \text{Binom}(k, p_2)$, we have $\Pr(X_1 < X_2) \geq q$.
We vary $p_1 \in \{0.1, 0.5\}$, $p_2 \in [p_1+0.05, p_1+0.3]$ and $q \in \{0.8, 0.9, 0.95\}$ and compute the corresponding $k$. The results are shown in Figure \ref{fig:computing_br}.

\begin{figure}[H]
    \centering
        \includegraphics[width=1\linewidth]{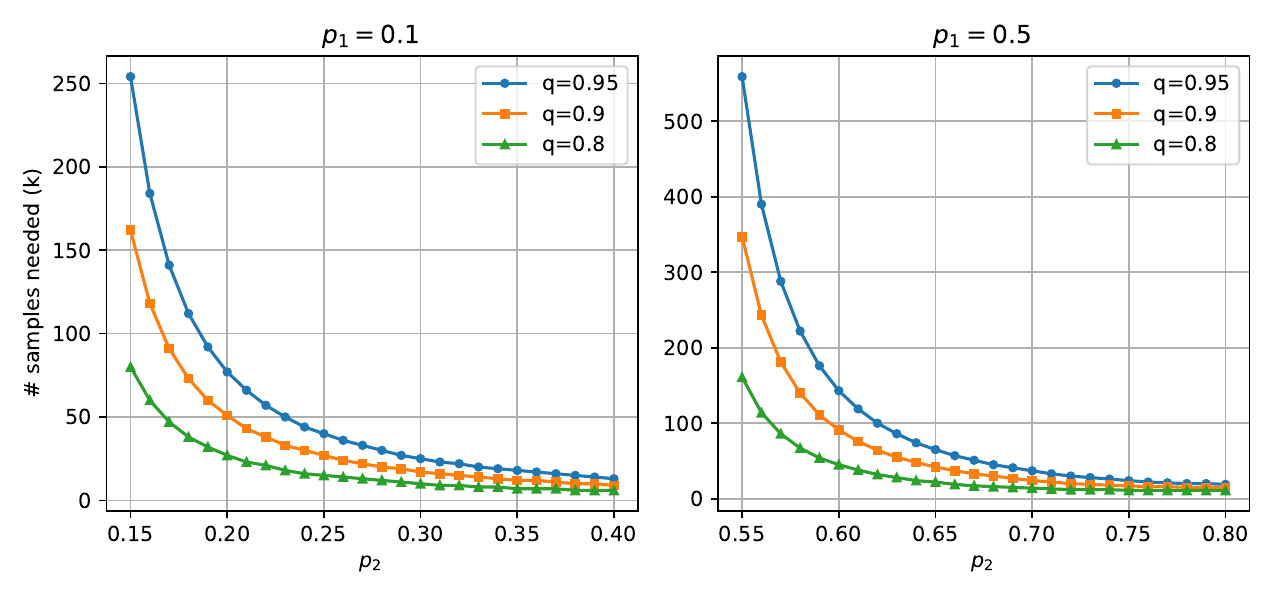}
    \caption{For $p_1 \in \{0.1, 0.5\}$, $p_2 > p_1$, and $q \in \{0.8, 0.9, 0.95\}$, we compute the smallest number of samples $k$ needed so that if $X_1 \sim \text{Binom}(k, p_1)$ and $X_2 \sim \text{Binom}(k, p_2)$, then $\Pr(X_1 < X_2) \geq q$. We vary $p_2$ from $p_1 + 0.05$ to $p_1+0.3$ in increments of 0.01.}
    \label{fig:computing_br}
\end{figure}

When $p_1 = 0.1$ and $p_2 = 0.15$, the number of samples needed to differentiate the two pools correctly with probability 80\%, 90\%, and 95\% are $k=80$, 162, and 254 respectively.
Note that $k$ is the number of samples needed from \textit{each} pool, hence $2k$ samples are needed in total.
This is a \textit{significant} number of samples that a firm needs to collect for the goal of distinguishing between two pools of applicants, which is practically infeasible for small to moderate sized firms.
The number of samples required decreases as the gap between $p_1$ and $p_2$ increases, however, it remains non-trivial. 
Even under a 0.2 gap between $p_1$ and $p_2$, to differentiate between $p_1 = 0.5$ and $p_2 = 0.7$ with 90\% probability, a firm needs 24 samples from each pool.

We note that the above calculations made several simplifications to the actual task of computing the best response.
First, the calculation was only to differentiate between \textit{two} pools of candidates.
If there are multiple pools, the firm will need to gather samples from all pools.
Second, we also assume that all other firms' strategies stay constant, while one firm is ``experimenting'' to collect data.
To converge an NE, all firms need to compute their best response and hence will need to experiment; if multiple firms experiment at the same time, this will interfere with each others' estimation.
Third, we also assumed that there are distinct ``pools'' of candidates with the same score, and there are many applicants in each pool.
In reality, there may not exist discrete ``pools'' of candidates, or there are not enough applicants from each pool needed to correctly estimate utilities.
Lastly, it may be unreasonable to assume that firms can ``experiment'' for the sake of collecting data for better estimation as interviewing is costly.

Given the large number of samples needed, as well as the reasons stated above, 
it seems unrealistic to hope that firms, using their own estimates of utility, can naturally arrive at an NE.
The bottleneck in computing the best response is in estimating the utility $U_n(s)$. 
Even though the algorithm provides a highly informative signal, $s$, about the applicants, the lack of knowledge on $n$ renders the information on $s$ almost useless.
If firms knew $n$, then they would be able to exactly compute the utility $U_n(s)$ and calculate their best response strategy.

\section{Extensions}

\subsection{Flexible Capacity with Fixed Welfare}

So far, we have assumed that all firms have a fixed capacity $c$, and we showed that there is a large gap in social welfare when comparing the naive solution to the Nash equilibrium.
However, in reality, firms may update their capacity to meet their hiring needs.
In this section, we fix the total social welfare, and then evaluate what the capacity needs to be under the different solution concepts to reach that level of social welfare.

We show that under both $\theta \in \{\textsc{corr}, \indep\}$,
when the social welfare is fixed, the total capacity needed to reach the same level of welfare \textit{scales with $N$} under the naive solution, while under the Nash equilibrium, the total capacity stays constant with $N$.

\begin{proposition} \label{prop:fixedW_corr}
Fix the social welfare $W \in (0, N \int_0^1 U_N(s)\varphi(s)ds)$ and decision scheme $\theta \in \{ \textsc{corr}, \indep\}$.
There exists a $c > 0$ such that:
\begin{itemize}
    \item With $N$ firms each with capacity $c$, the naive solution yields a social welfare of $W$.
    \item When $W \leq \int_{0.5}^1 s \varphi(s)ds$, with $N$ firms each with capacity $c/N$, there is a Nash equilibrium that yields a social welfare of $W$.
\end{itemize}
\end{proposition}

This result state that the each firm needs capacity $c$ under the Naive solution, while each firm only needs capacity $c/N$ under the Nash equilibrium to yield the same total welfare $W$.
This corroborates the claim that the Nash equilibrium is a much more desirable solution than the Naive solution.
Previously, we have established that when firm capacities are fixed, then the welfare is substantially higher until the NE compared to Naive.
In this section, we show that the same welfare can be achieved under the NE solution with substantially less capacity compared to the Naive solution.
Hence, even if the same number of applicants are eventually hired, the NE solution is more efficient in that firms can substantially reduce their interview capacity.

\subsection{Applicant Preferences over Tiered Firms}

We relax the baseline model by introducing two tiers of firms, A and B, to capture applicant preferences. Applicants always prefer an offer from any tier A firm to any tier B firm. If an applicant receives offers from both tiers, they reject all tier B offers and accept one tier A offer uniformly at random; if they only receive offers from one tier, they accept one of them uniformly at random. 

Our main result is that under this tiered model, the social welfare at the Nash equilibrium is higher than at the Nash equilibrium under the non-tiered model, as introducing tiers further reduces congestion for the top applicants. 
Therefore, it is even more important in the tiered setting for firms to act strategically and converge to an NE.

\paragraph{Utility.} 
We first derive the utility functions when when firms are differentiated by tier. 
Because tier A firms only compete among themselves, their utility functions remain identical to those in the non-tiered case. 
In contrast, the utility of tier B firms depends on the strategy profiles of the tier A firms.
Let $U_{n_B}(s, \mathbf{f_A})$ be the utility of tier B firms, \textit{given} the strategy profile $\mathbf{f_A}$ of tier A firms.
Letting $M(s, \mathbf{f_A})$ be the number of tier A firms that interview an applicant of score $s$ under the tier A strategy profile $\mathbf{f_A}$, when $\theta = \textsc{corr}$, the expected utility derived by one of the $n_B$ tier B firms from interviewing an applicant with score $s$ is 
\begin{align*}
    U_{n_B}(s, \mathbf{f_A}) = \begin{cases}
            s/n_B,&\text{if } M(s, \mathbf{f_A}) = 0\\
            0, & \text{otherwise}.
        \end{cases}
\end{align*}
When $\theta = \indep$, the utility becomes $U_{n_B}(s, \mathbf{f_A}) = (1-(1-s)^{n_B})(1-s)^{M(s, \mathbf{f_A})}/n_B$
(see Appendix \ref{appendix:deri_indep_tiered} for the derivation).  Figure \ref{fig:utility_tiers} displays an example of these utility functions for $n_b = 1, \cdots, 6$ under both $\theta\in \{\textsc{corr}, \indep\}$, when there are 7 tier A firms at Nash equilibrium.

\begin{figure}[H]
    \centering
    \subfigure[$\theta = \textsc{corr}$]{
        \includegraphics[width=0.47\linewidth]{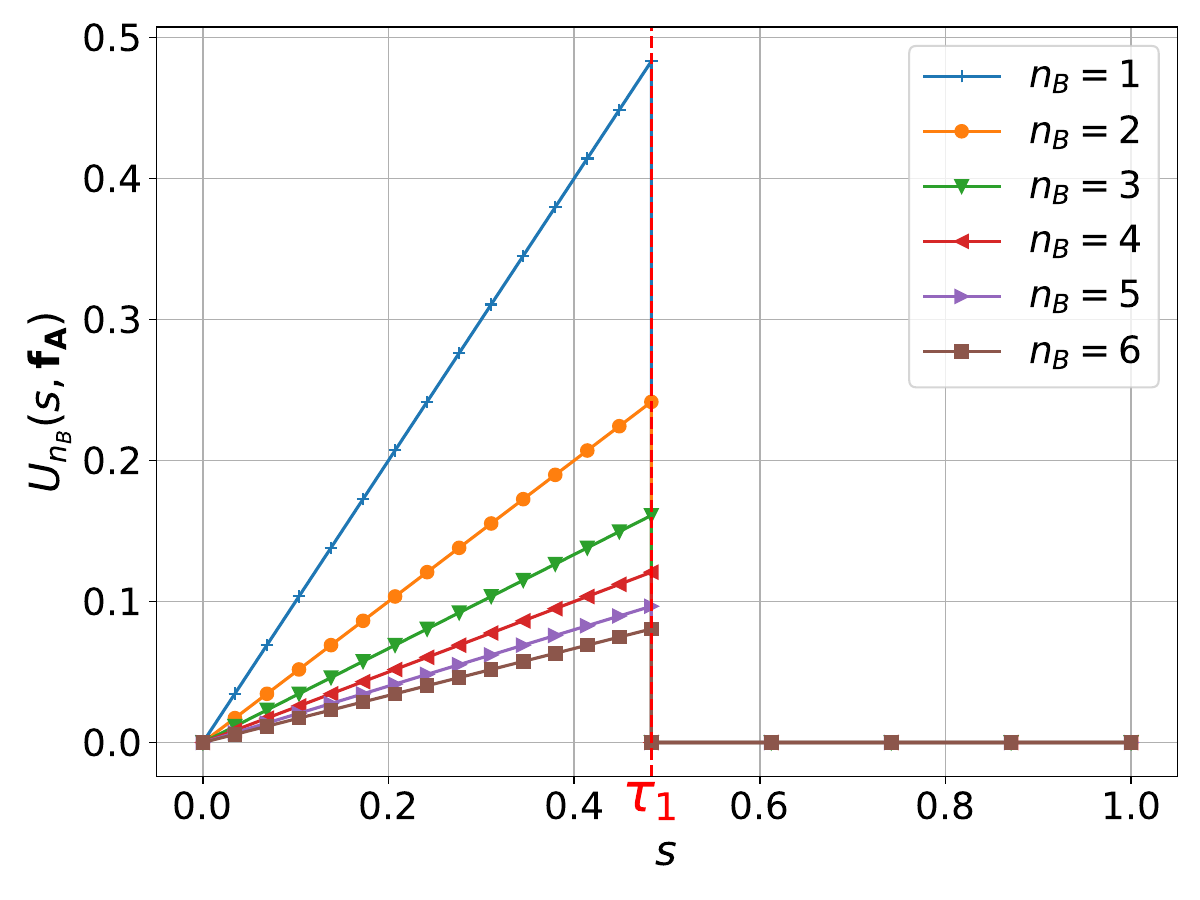}
        \label{fig:utility_tier_corr}
    }
    \subfigure[$\theta = \indep$]{
        \includegraphics[width=0.47\linewidth]{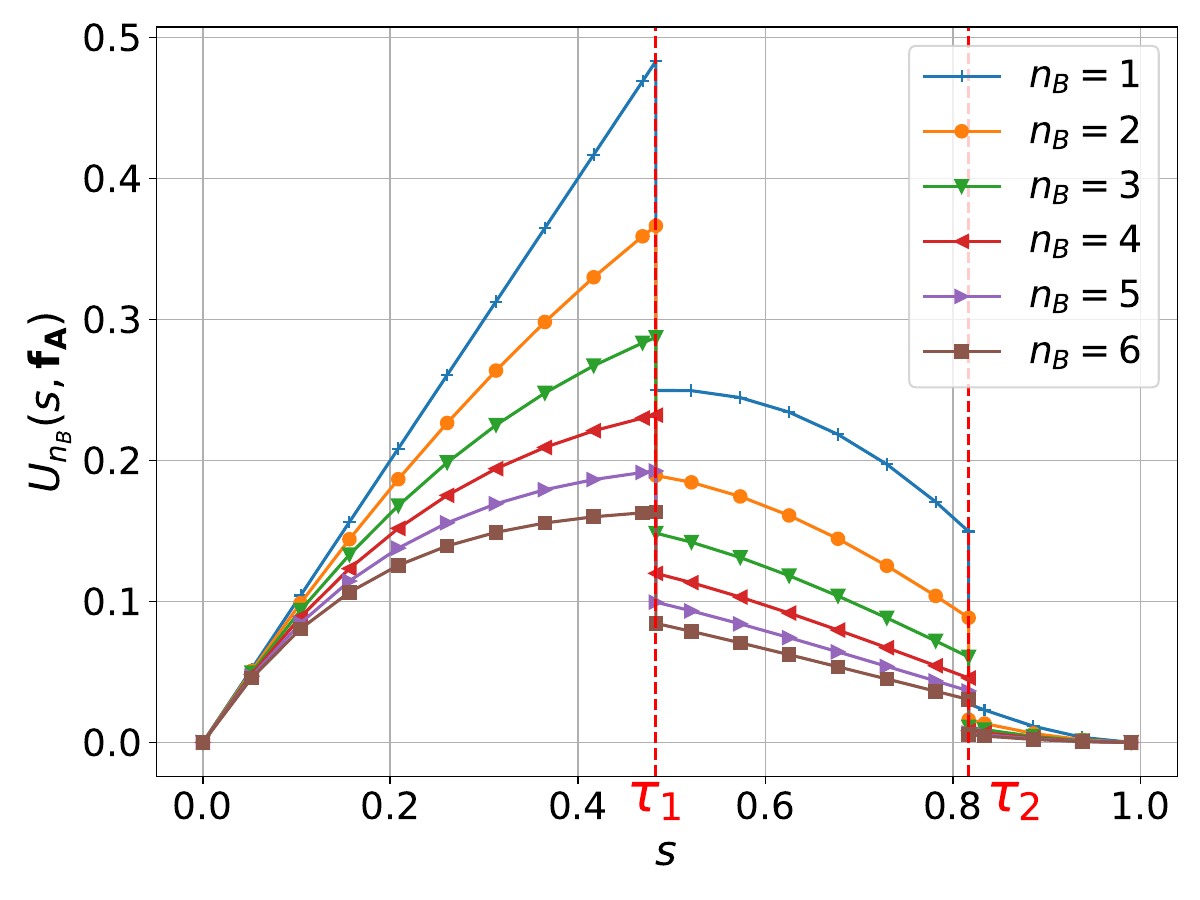}
        \label{fig:utility_tier_indep}
    }

    \caption{Utility curves $U_{n_B}(s, \mathbf{f_A})$ for tier B firms when tier A firms are at a Nash equilibrium, for both $\theta\in\{\textsc{corr}, \indep \}$. We have $c=0.1$, number of tier B firms $n_B = 1,\cdots, 6$ and number of tier A firms $n_A = 7$. 
    }
    \label{fig:utility_tiers}
\end{figure}

\paragraph{Nash equilibria.} Given that the utility of tier A firms is independent of the actions of tier B firms, we show that under a Nash equilibrium, there are thresholds $(\tau_i)_{i=1}^{N_A}$ for tier A firms that satisfy the conditions in Theorem \ref{thm:equilibrium}. When $\theta = \textsc{corr}$, tier B firms interview only applicants with scores below $\tau_1$. In this case, the equilibrium strategies of tier B firms can be characterized by another set of thresholds $(\tau'_j)_{j=1}^{N_B}$ with $0\leq \tau'_1\leq \tau'_2\cdots\leq \tau'_{N_B}\leq \tau_1$, which also satisfy the conditions in Theorem \ref{thm:equilibrium}. We formalize this in Proposition \ref{prop:tiered_eq_corr}, which establishes that when $\theta = \textsc{corr}$, the two tiers naturally separate at equilibrium: tier A firms compete for the top applicants, while tier B firms concentrate on those with lower scores.

\begin{proposition} \label{prop:tiered_eq_corr}
    Suppose there are $N_A$ tier A firms and $N_B$ tier B firms, and $\mathbf{f}_A = (f_1^A, \cdots, f_{N_A}^A)$ and $\mathbf{f}_B = (f_1^B, \cdots, f_{N_B}^B)$ are the corresponding strategy profiles. Under $\theta = \textsc{corr}$, the strategy profile $\mathbf{f}_A\cup \mathbf{f}_B$ is a Nash equilibrium if and only if the following conditions hold:

    \begin{enumerate}
     \item There exist thresholds $0=\tau_0\leq \tau_1 \leq \tau_2 \leq \cdots ,\leq \tau_{N_A}\leq  \tau_{N_A+1}=1$ such that $\mathbf{f}_A$ satisfies the same conditions as Theorem \ref{thm:equilibrium}.

    \item There exist thresholds $0=\tau_0\leq \tau'_1 \leq \tau'_2 \leq \cdots ,\leq \tau'_{N_B}\leq  \tau'_{N_B+1}=\tau_1$ such that $\mathbf{f}_B$ satisfies the same conditions as Theorem \ref{thm:equilibrium} except that the utility function $U_n(s)$ is replaced with $U_{n_b}(s, \mathbf{f_A})$.

    \end{enumerate}
    
\end{proposition}

When $\theta = \indep$, since interview outcomes are independent, applicants may accept offers from tier B firms even when they also receive interviews from tier A firms. 
As a result, tier B firms may interview applicants who are also interviewed by tier A; this makes the equilibrium characterization slightly more complex.
The next result describes a set of sufficient conditions for a strategy profile to be an NE.

\begin{proposition}\label{prop:tiered_eq_indep}
    Suppose there are $N_A$ tier A firms and $N_B$ tier B firms, and $\mathbf{f}_A = (f_1^A, \cdots, f_{N_A}^A)$ and $\mathbf{f}_B = (f_1^B, \cdots, f_{N_B}^B)$ are the corresponding strategy profiles. Under $\theta = \indep$, the strategy profile $\mathbf{f}_A\cup \mathbf{f}_B$ is Nash equilibrium if the following conditions are satisfied:
    \begin{itemize}
        \item There exist thresholds $0=\tau_0\leq \tau_1 \leq \tau_2 \leq \cdots ,\leq \tau_{N_A}\leq  \tau_{N_A+1}=1$ such that $\mathbf{f}_A$ satisfies the conditions in Theorem \ref{thm:equilibrium}.
        \item $\mathbb{P}(f_i^B(S)=1)=c$ for all $i\in [N_B]$.
        \item There exists a constant $H$ such that, for each $n\in [N_B]$,  $M(s, \mathbf{f_B}) = n$ on any subset $S_n \subseteq [0,1]$ with $U_{n}(s, \mathbf{f_A})\geq H$ and $U_{n+1}(s, \mathbf{f_A})< H$ for any $s\in S_n$. Moreover, $M(s, \mathbf{f_B}) = 0$ whenever $U_{1}(s, \mathbf{f_A})< H$.
    \end{itemize}
\end{proposition}

\paragraph{Social welfare.} Now we fix the total number of firms, and we compare the social welfare of the tiered and non-tiered settings. Proposition \ref{prop:extension_sw} establishes that, when $\theta=\textsc{corr}$, introducing tiers strictly increases social welfare under the NE.

\begin{proposition}\label{prop:extension_sw}
    Fix the total number of firms $N$. Let $SW_\text{tiered}$ be the social welfare when the firms have two tiers and are at Nash equilibrium. Similarly, let $SW_\text{non-tiered}$ be the social welfare under Nash equilibrium when all firms are equivalent. Then, under $\theta=\textsc{corr}$ and $c<1$, $SW_\text{tiered}>SW_\text{non-tiered}$. 
\end{proposition}

Next, for $\theta = \indep$, we provide numerical results for different values of $n_A$, $n_B$, and $c$, and plot the welfare gap ($SW_\text{tiered} - SW_\text{non-tiered}$) between the tiered and non-tiered equilibria in Figure \ref{fig:tiers}. We see that the social welfare is always higher under the tiered setting, where the welfare gap generally increases as capacity increases. When total capacity is sufficiently small ($Nc<0.5$), no applicant receives multiple interviews, so the welfare gap remains close to zero.

To provide more intuition for the $\theta = \indep$ regime, Figure \ref{fig:extension_NEs} plots the number of interviews received by applicants in the non-tiered and tiered regimes.
Relative to the non-tiered setting, the tiered equilibrium allocates more interviews to applicants with scores in the range $s\in[0.24, 0.42]$, and fewer interviews to applicants with scores above $0.76$. This reallocation improves welfare: in the high-score region, additional interviews contribute little because firms already face intense competition and applicants are unlikely to accept more offers. By contrast, in the lower-score region, an extra interview substantially increases the chance that an applicant finds a match, generating a larger marginal contribution to welfare. As a result, total social welfare is higher in the tiered setting (0.61) than in the non-tiered setting (0.55).

\begin{figure}[H]
    \centering
    \subfigure[$n_A = 5$]{
        \includegraphics[width=0.47\linewidth]{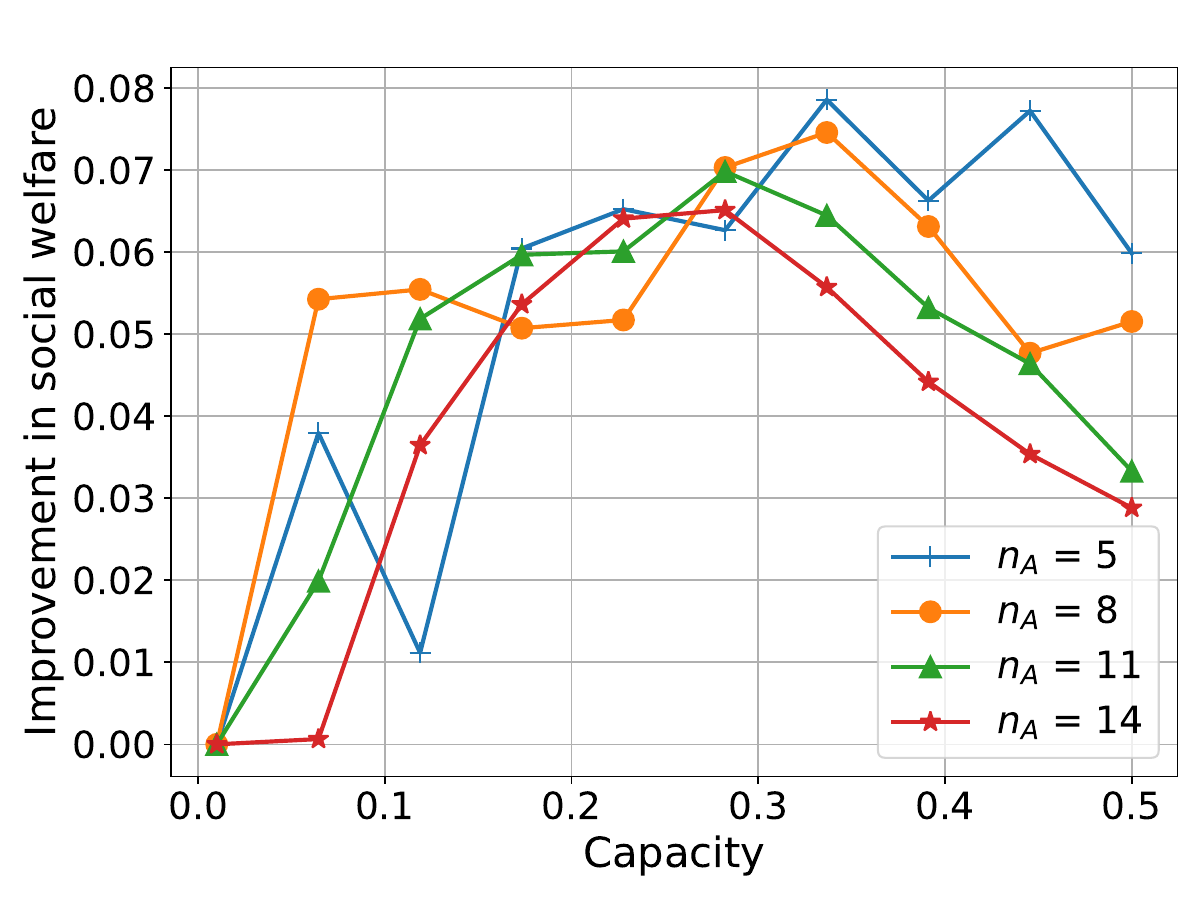}
        \label{fig:tier1}
    }
    \subfigure[$n_B = 5$]{
        \includegraphics[width=0.47\linewidth]{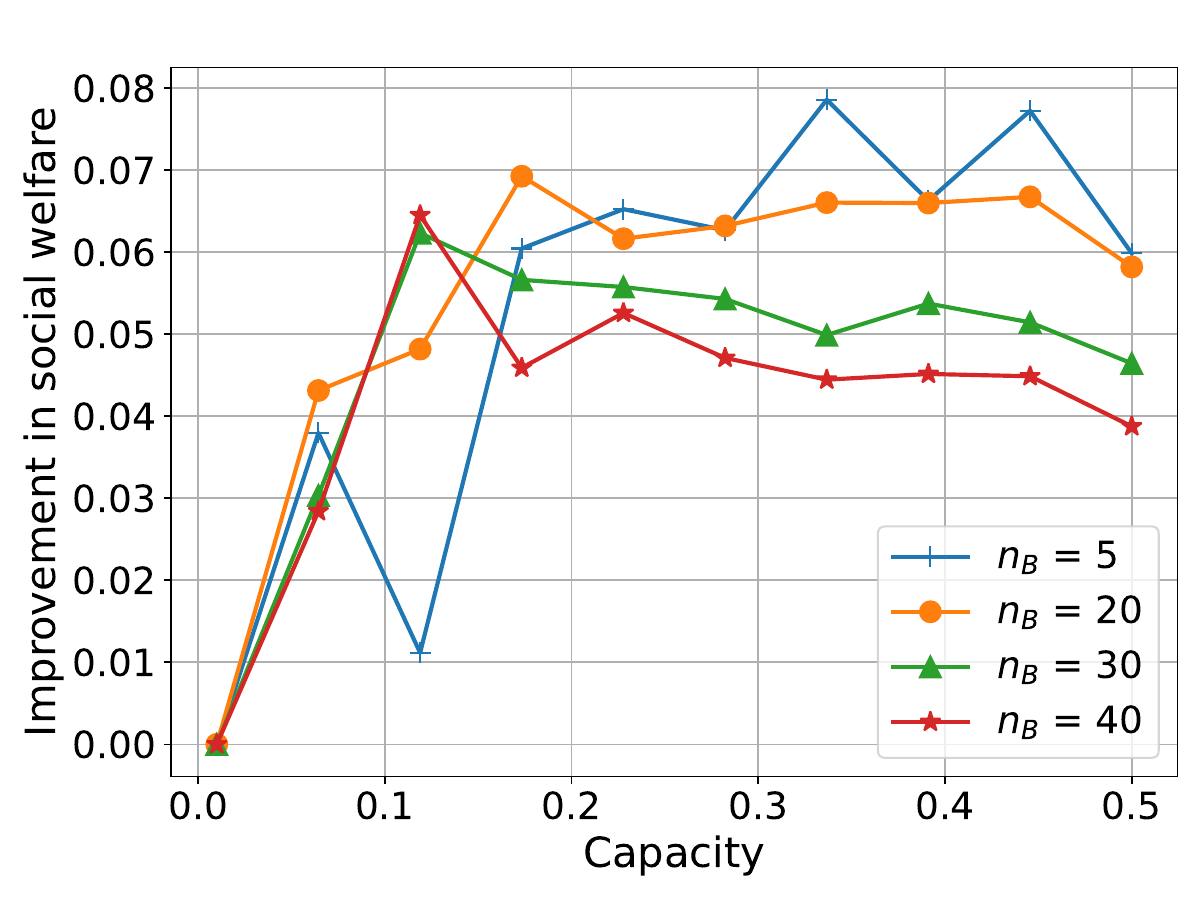}
        \label{fig:tier2}
    }

    \caption{Social welfare gap between the tiered and non-tiered nash equilibria, $SW_\text{tiered} - SW_\text{non-tiered}$, for varying capacity $c\in (0,0.5]$, $n_A \in \{5, 8, 11, 14\}$, and $n_B\in \{5, 8, 11, 14\}$. }
    \label{fig:tiers}
\end{figure}

\begin{figure}[H]
    \centering
    \subfigure[Non-tiered]{
        \includegraphics[width=0.47\linewidth]{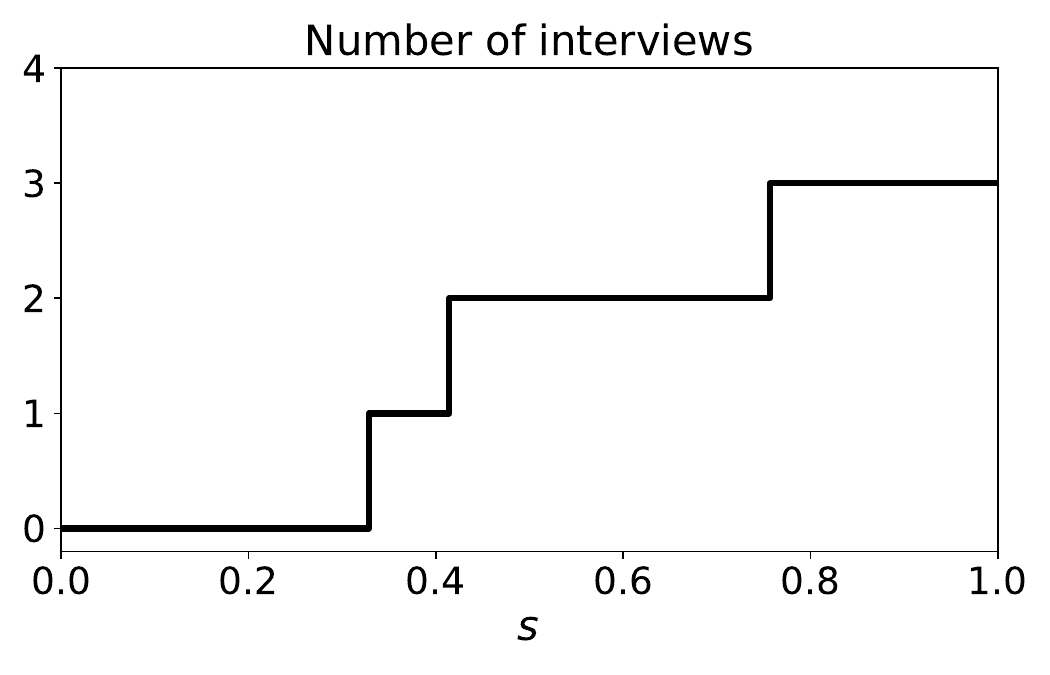}
        \label{fig:nontiered_NE}
    }
    \subfigure[Tiered]{
        \includegraphics[width=0.47\linewidth]{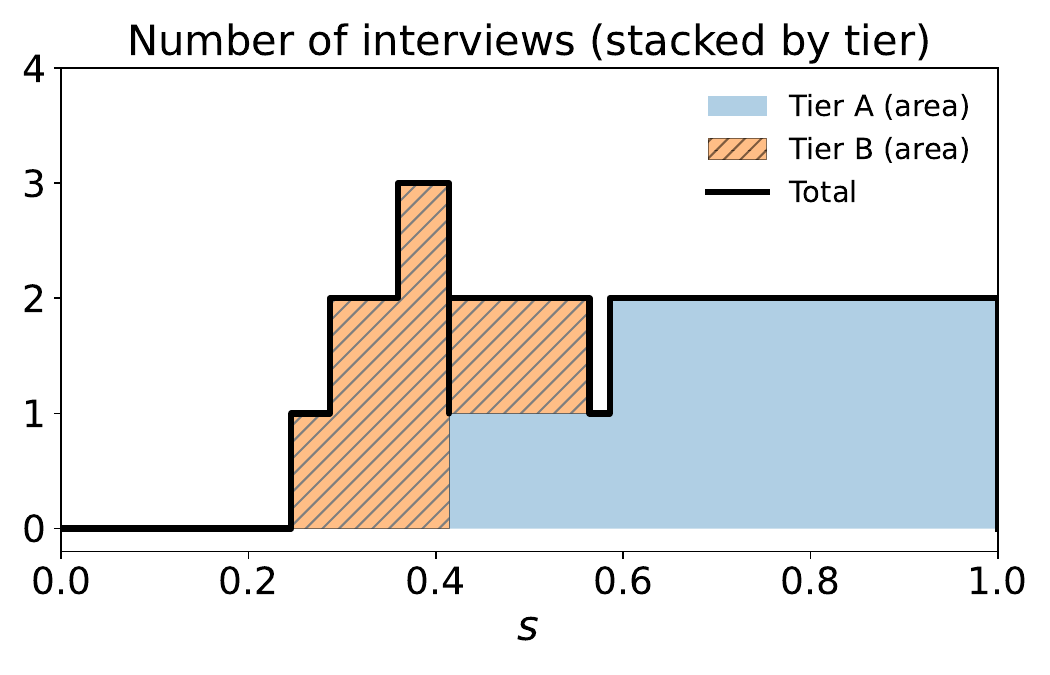}
        \label{fig:tiered_NE}
    }

    \caption{Number of interviews received by applicants across different score levels in tiered and non-tiered equilibrium with $N = 15$, $n_A = 10$, $n_B = 5$, and $c=0.1$. In Figure \ref{fig:tiered_NE}, tier A firms form a equilibria with $\tau_1 = 0.42, \tau_2 = 0.76$. The left shaded band (orange) and the right shaded band (blue) represent the number of interviews received by tier B and tier A firms, respectively. While tier B firms interview applicants with middle scores, tier A firms interview mostly applicants with high scores.} 
    \label{fig:extension_NEs}
\end{figure}

Overall, introducing tiers alters equilibrium interview patterns in ways that reduce congestion and raise welfare. When firms differ in desirability, interview congestion naturally separates across tiers, reducing competition among high-ranked firms and redirecting capacity toward lower-scored applicants.
Therefore, under tiered firms, it is even more important that firms act strategically and converge to an NE.

\section{Conclusion and Discussion} \label{sec:discussion_and_conclusion}

This paper studies the role of strategic behavior in algorithmic hiring environments under congestion and algorithmic monoculture. 
By analyzing Nash equilibrium strategies and their impact on social welfare, we demonstrate that strategic selection can mitigate the inefficiencies of congestion, which benefits both firms and applicants. 
Our findings highlight the value of providing firms with information about applicant congestion, as it facilitates convergence to a Nash equilibrium.
This result implies that the effectiveness of algorithmic hiring platforms is contingent on their ability to enhance coordination among firms, via congestion information or personalizing algorithmic scores.
Without mechanisms to address competition for top-scoring candidates, platforms risk exacerbating inefficiencies, limiting their value.

Lastly, we discuss several assumptions of our model, and corresponding limitations and future directions.

\paragraph{One-shot process.}
One of the main assumption of our model is that hiring is a one-shot process where all firms make one set of interview decisions, which then lead to hiring decisions. 
In reality, the hiring process may take multiple stages, where firms that fail to hire in the initial round may revisit the pool of candidates to interview additional applicants.
At the opposite extreme, hiring outcomes could result from a stable matching, achieved through numerous iterations (e.g., via the deferred acceptance algorithm). 
Our work focuses on the ``one-shot'' nature of the hiring process and the congestion it induces.
Real-world hiring likely falls between these two extremes, with limited rounds of iteration that neither resolve all congestion issues nor achieve full stability. 
We leave as an interesting direction to analyze a model that interpolates between these two regimes.

\paragraph{Homogeneous firm-side utility.}
We also assume that the firm receives equal utility from hiring any applicant, as long as they pass the interview.
In reality, an applicant's profile is multi-dimensional, and a firm may incur different utilities from different profiles. 
We believe that assuming equal utility from all hires is a good approximation of hiring practices for roles with many positions (e.g., entry-level software engineering), whereas this may not be a good assumption for a specialized position (e.g., a C-suite role).

\paragraph{Estimating utilities.}
Our analysis highlights the importance of providing firms with information about applicant congestion levels to improve decision-making.
In our model, since we assume that the utility for interviewing an applicant is $U_n(s)$, a firm can exactly compute this utility if they have access to both $n$ and $s$.
However, in reality, there are many factors that prohibit this exact calculation.
First, the exact form of $U_n(s)$ depends on the decision scheme (e.g., $\theta \in \{\textsc{corr}, \indep\}$).
Firms may know exactly which regime they are in, and the regime may lie in between the two we consider.
Even so, we believe our results are still significant in that the utilities heavily depend on \textit{both} $s$ and $n$, and therefore having access to $n$ can greatly help the firm make decisions.
For example, if two applicants have the same value of $n$ but the first applicant has a higher score, then clearly the first applicant dominates the second, even if the exact form of $U_n(s)$ is unknown.

\paragraph{Interview decision based on the score.}
In our model, we assume that firms make interview decisions solely based on the algorithmic score $s$.
In reality, firms may have access to additional features about the applicant which can be used in their interview decisions. For example, firm A might prioritize applicants with score $s = 0.6$ who graduated from college X, while firm B focuses on applicants with $s = 0.6$ from college Y. In such cases, there is no overlap between the two firms’ interview pools, despite their reliance on the same score.
Our model effectively accommodates this scenario by assuming that scores are continuous, allowing firms to select subsets of $[0, 1]$ with arbitrary precision. This can be interpreted as encoding additional applicant attributes in the decimal representation of $s$.
Mathematically, this equivalence ensures that the model captures the effects of feature-based selection without loss of generality.

\paragraph{Interpretation of the algorithmic score.}
We also interpret an applicant's score $s$ as the probability that the applicant will pass the interview.
In reality, the algorithm might simply give a ranking over the candidates, or a score that does not exactly represent the probability of passing an interview.
In this case, the firm may need time and experience to understand how to interpret the output of the algorithm.

%% file: appendix.tex
\section{Proof of Main Results}\label{appendix:proofs}
\subsection{Proof of Theorem \ref{thm:equilibrium}}
We organize the proof as follows. We first present some intermediate results before proving Theorem \ref{thm:equilibrium}. Next, we show that if the strategy profile $\bbf$ is a Nash equilibrium, it must satisfy each of the four conditions in Theorem \ref{thm:equilibrium}. Finally, we prove that any strategy profile that satisfies the four conditions is a Nash equilibrium. 

\begin{lemma}\label{lemma:A1}
    Let X be a measurable set on [0,1] with $\Pr(S\in X)=x$ where $x\in (0,1]$. Then for any $0<x'<x$, we can always find a subset of X such that $\Pr(S\in X')=x'$.
\end{lemma}

\begin{myproof}
    Consider the function $f(t) = Pr(S\in (X\cap [0,t]))$. We will first show that $f(t)$ is continuous. 
    First, suppose $t' < t$.
    Then, we can write
    \begin{align*}
    f(t) 
    &= \Pr(S\in (X\cap [0,t]))\\
    &= \Pr(S\in (X\cap [0,t'])) +  \Pr (S \in (X\cap [t',t]) )\\
    &= f(t') +  \Pr (S \in (X\cap [t',t]) ).
    \end{align*}
    Clearly, as $t' \to t$, the term $\Pr (S \in (X\cap [t',t]) )$ goes to 0.
    Therefore, $\lim_{t' \to t^-} f(t') = f(t)$.
    Similarly, when $t' > t$, we can write $f(t') = f(t) + \Pr (S \in (X\cap [t,t']) )$, which implies that $\lim_{t' \to t^+} f(t') = f(t)$.
    Therefore, $\lim_{t' \to t} f(t') = f(t)$, and hence $f(t)$ is continuous.
    
    Furthermore, $f(0) = 0$ and $ f(1) = x$. Hence, by the intermediate value theorem, for any $x'\in [0,x]$ there is some $t^*\in [0,1]$ such that $f(t^*)=x'$. Let $S' = X\cap [0,t^*]$, we have $\Pr(S\in X') = x'$ as required. 
\end{myproof}

\begin{myproof}[\text{Proof of Theorem \ref{thm:equilibrium}}]
Suppose the strategy profile $\bbf$ is a Nash equilibrium. We will show that there exist thresholds $0\leq \tau_1 \leq \tau_2\cdots, \tau_N\leq 1$ that satisfy the four conditions listed in the theorem.

\begin{enumerate}
    \item  $\Pr(f_i(S)=1) = c$ for all $i \in [N]$.
    \item  $M(S) = m;s \in [\tau_m, \tau_{m+1})$ for all $m = 0, 1,  \dots, N$.

    \item  $U_n(\tau_n) \leq U_m(\tau_m)$ for all $n < m \leq \Mmax(\bbf)$.
    
    \item Consider any $n < m \leq \Mmax(\bbf)+1$ where 
    $U_n(\tau_n) < U_m(\tau_m)$, and consider any firm $i$ and score $s \in [\tau_n, \tau_{n+1})$ where $f_i(s) = 1$ and $U_n(s) < U_m(\tau_m)$. 
    For any $s' \in [\tau_{m-1}, \tau_m)$ where $U_m(s') > U_n(s)$, we have that $f_i(s') = 1$.

\end{enumerate}

Let's start by proving the \textbf{first condition}. Given capacity c, we want to show that $\Pr(f_i(S) = 1) = c$ for all $i\in[N]$. Suppose, by contradiction, that $\Pr(f_i(S)=1)<c$ for some $i\in[N]$. We will show that firm $i$ can deviate from $f_i$ to strictly improve its strategy. Let $K = \{s\in [0,1]: f_i(s)=1\}$ be the applicants interviewed by firm $i$ under $\bbf$. Under the strategy profile $\bbf$, by equation (1), the utility derived by firm $i$ is 
\begin{align*}
    u(f_i, f_{-i}) &= \int_{0}^1 f_i(s) U_{M(s, \mathbf{f})}(s) \varphi(s)ds\\
    & = \int_K U_{M(s, \mathbf{f})}(s) \varphi(s)ds.
\end{align*}
Since $\Pr(f_i(S) = 1)<c$, by Lemma \ref{lemma:A1}, there is a subset $E\subseteq [0,1]\setminus K$ such that $Pr(S\in E)=c-Pr(f_i(S)=1)$.

Consider an alternative strategy $f'_i$ where $f'_i(s) = 1$ if and only if $s \in K \cup E$.
We have $\Pr(f_i'(S) = 1) = \Pr(S \in K) + \Pr(S \in E) = c$, hence $f_i'$ satisfies the capacity requirement.
We will now show $u(f'_i, f_{-i})>u(f_i, f_{-i})$.
Let $\bbf' = (f_1, \cdots, f'_i, \cdots, f_N)$ be the strategy profile after firm $i$ deviates from $f_i$.
By construction, $M(s, \mathbf{f'}) = M(s, \mathbf{f})$ on $K$. As a result, the utility derived by firm $i$ under $\bbf'$ is 

    \begin{align*}
        u(f'_i, f_{-i})
        & = \int_{0}^1 f'_i(s) U_{M(s, \mathbf{f'})}(s) \varphi(s)ds\\
        & = \int_{K} U_{M(s, \mathbf{f'})}(s) \varphi(s)ds + \int_{E} U_{M(s, \mathbf{f'})}(s) \varphi(s)ds\\
        & = \int_{K} U_{M(s, \mathbf{f})}(s) \varphi(s)ds + \int_{E} U_{M(s, \mathbf{f'})}(s) \varphi(s)ds\\
        & = u(f_i, f_{-i}) + \int_{E} U_{M(s, \mathbf{f'})}(s) \varphi(s)ds.
    \end{align*}

By assumption, $U_{M(s, \mathbf{f'})}(s)>0$ and $\varphi(s)>0$ whenever $s>0$. We must have $\int_{E} U_{M(s, \mathbf{f'})}(s) \varphi(s)ds>0$ since $\Pr(S\in E)>0$. Therefore, 
$u(f'_i, f_{-i}) > u(f_i, f_{-i})$.
This is a contradiction to $\bbf$ being an NE, and hence it must be that  $\Pr(f_i(S) = 1) = c$ for all $i\in [N]$.

We will then prove the \textbf{second condition}. Our goal is to show that there exist thresholds $\tau_1\leq \cdots \tau_N$ such that $M(s;\bbf) = m;s \in [\tau_m, \tau_{m+1})$ for all $m \in [\Mmax(\bbf)]$. We will prove the statement by constructing a sequence of thresholds $\tau_1,\cdots \tau_N$ that satisfy the requirements. Define the thresholds as follows: 
\begin{align*}
    \tau_m = 
\begin{cases} 
\min \{s\in [0,1]; M(s;\bbf)\geq m\} & \text{if } m\leq \Mmax(\bbf), \\
1 & \text{if } \Mmax(\bbf)<m\leq N.
\end{cases}
\end{align*}
We will show $M(s;\bbf) = m;s \in [\tau_m, \tau_{m+1})$ for all $m\in [\Mmax(\bbf)]$.

We will first prove the statement for $m=0$. By assumption, $\tau_1 = \min \{s\in [0,1]; M(s;\bbf)\geq 1\}$. Therefore, for any $s\in [0,\tau_1)$, we must have $M(s;\bbf)=0$. Hence, $M(s;\bbf) = 0$ for all $s \in [0, \tau_{1})$. 

Next, we will show that for any $0<m\leq \Mmax(\bbf)$, $M(s;\bbf) = m$ for any $s \in [\tau_m, \tau_{m+1})$. We will first show that $M(s;\bbf) \leq m$ for all $s \in [\tau_m, \tau_{m+1})$. Consider the following two cases: $m<\Mmax(\bbf)$ and $m=\Mmax(\bbf)$. If $m<\Mmax(\bbf)$, by assumption, $\tau_{m+1} = \min \{s\in [0,1]; M(s;\bbf)\geq m+1\}$, and, as a result, we must have $M(s;\bbf) \leq m$ for all $s<\tau_{m+1}$. If $m= \Mmax(\bbf)$, then clearly $M(s;\bbf) \leq m$ for all $s<\tau_{m+1}$. Hence, $M(s;\bbf) \leq m$ for any $s \in [\tau_m, \tau_{m+1})$. Next, we will show $M(s,\bbf)\geq m$ for $s\in [\tau_m, \tau_{m+1})$. Suppose, by contradiction, $M(s,\bbf)< m$ for some $s_0\in [\tau_m, \tau_{m+1})$. Since for each strategy the applicants interviewed can be written as a union of a finite number of intervals, there exists an interval $s_0\in [a_1,b_1]\subset [\tau_m, \tau_{m+1}]$ with $a_1<b_1$ such that $M(s;\bbf) < m$ on this interval. Similarly, since $\tau_{m} = \min \{s\in [0,1]; M(s;\bbf)\geq m\}$, there exists an interval $[\tau_m, b_2]$ with $b_2>\tau_m$ such that $M(s;\bbf) = m$ on this interval. We notice that $b_2<a_1$ by construction and $[\tau_m, b_2]$ and $[a_1, b_1]$ are disjoint intervals. There are strictly fewer firms competing for applicants with score $s\in [a_1, b_1]$ then applicants with score $s\in[\tau_m, b_2]$. Therefore, there must exist some firm $i$ with $K_i = \{s\in [0,1]; f_i(s)=1\}$ such that firm $i$ interviews some applicants in $[\tau_m, b_2]$ but not all applicants in $[a_1, b_1]$. Mathematically, $K_i\cap  [\tau_m,b_2)\neq \emptyset$ and $K_i^c\cap [a_1,b_1]\neq \emptyset$. We also notice that  
\begin{align*}
    U_{M(s;\bbf)+1}(s)\geq U_{m}(s)>U_m(a_1)>U_m(b_2)
\end{align*}
for $s\in(a_1,b_1]$ since $U_n$ is strictly increasing in $s$ and decreasing in $n$. This implies that if firm $i$ interviews an applicant with score $s\in[a_1,b_1]$, the resulting utility is strictly higher than the utility firm $i$ can derive by interviewing an applicant with score $s\in [\tau_m, b_2]$. Therefore, firm $i$ can deviate from $f_i$ by moving support from $K_i\cap  [\tau_m,b_2)$ to $K_i^c\cap [a_1,b_1]$, where firm $i$ would earn strictly more utility. This is a contradiction to $\bbf$ being a Nash equilibrium, and hence it must be that $M(s,\bbf)\geq m$ for all $s\in[\tau_m, \tau_{m+1}]$. Since $M(s,\bbf)\leq m$ and $M(s,\bbf)\geq m$ for all $s\in[\tau_m, \tau_{m+1})$, we must have $M(s,\bbf)= m$ for all $s\in[\tau_m, \tau_{m+1})$.

We will then prove the \textbf{third condition}. Assume $\bbf$ is a Nash equilibrium and the thresholds are defined as in condition 2. We will show that $U_m(\tau_m) \leq U_{m+1}(\tau_{m+1})$ for all $m \leq \Mmax(\bbf)-1$. Consider two cases: $c=1$ and $c<1$. If $c=1$, then each firm will interview all applicants, and we have $\tau_1=\cdots = \tau_N = 0$. Accordingly, $U_1(\tau_1)=\cdots = U_N(\tau_N) = 0$. We get the desired result.

If $c<1$, we will first show that $\tau_1>0$. Suppose, by contradiction, $\tau_1=0$. Then there exists a firm $i$ and $\delta>0$ such that firm $i$ interviews all applicants with score in $[0, \delta]$. Since $c<1$, there exists an interval $[a,b]\subset [\delta, 1]$ with $a<b$ such that $f_i(s)=0$ for $s\in[a,b]$. Since $U_n(s)$ is continuous and monotone, $\min_{s\in[a,b]}U_{M(s;\bbf)+1}(s)>0$ and $U_n(0) = 0$, there exists $0<\epsilon\leq \delta$ such that $U_{M(s;\bbf)}(s) < \min_{s\in[a,b]}U_{M(s;\bbf)+1}(s)$ for $s\in [0,\epsilon]$. By construction, if firm $i$ interviews an applicant with score $s\in [a,b]$, the utility is strictly higher than interviewing any applicant with score $s\in[0, \epsilon]$. Therefore, firm $i$ can deviate from $f_i$ by moving support from $[0,\epsilon]$ to within $[a,b]$, where firm $i$ would earn strictly more utility. This is a contradiction to $\bbf$ being a Nash equilibrium. Hence, we must have $\tau_1>0$.

Next, we will show that $\tau_m \neq \tau_{m+1}$ for all $m=1, \cdots, \Mmax(\bbf)$. Suppose, by contradiction, $\tau_m = \tau_{m+1}$ for some $m$. Since $M(s,\bbf)<m$ for $s<\tau_m$ and $M(s,\bbf)\geq m+1$ for $s\geq\tau_{m+1}$, there exists firm $i$, $\delta_1>0$, and $\delta_2>0$ such that $f_i(s)=0$ if $s\in[\tau_m-\delta_1, \tau_m]$ and $f_i(s)=1$ if $s\in[\tau_{m+1}, \tau_{m+1}+\delta_2]$. We notice that $U_{m}(\tau_m)>U_{m+1}(\tau_{m}) = U_{m+1}(\tau_{m+1})$ since $U_n(s)$ is strictly decreasing in n. Since $U_n(s)$ is continuous, there exists $0<\epsilon_1<\delta_1$ and $0<\epsilon_2<\delta_2$ such that 
\begin{align*}
    U_{m}(\tau_m)>U_{m}(\tau_m-\epsilon_1)>U_{m+1}(\tau_{m+1}+\epsilon_2)>U_{m+1}(\tau_{m+1}).
\end{align*}
Therefore, the utility firm $i$ can derive by interviewing an applicant with score $s\in[\tau_m-\epsilon_1, \tau_m]$ is strictly higher than the utility it can derive by interviewing an applicant with score $s\in[\tau_{m+1}, \tau_{m+1}+\epsilon_2]$. As a result, firm $i$ can deviate from $f_i$ by moving support from $[\tau_{m+1}, \tau_{m+1}+\epsilon_2]$ to $[\tau_m-\epsilon_1, \tau_m]$. This is a contradiction to $f$ being a Nash equilibrium. Hence, we must have $\tau_m \neq \tau_{m+1}$ for all $m=1, \cdots, \Mmax(\bbf)$.

We will then show that $U_m(\tau_m) \leq U_{m+1}(\tau_{m+1})$ for all $m< \Mmax(\bbf)$. Suppose, by contradiction, $U_m(\tau_m) > U_{m+1}(\tau_{m+1})$ for some $m=1, \cdots, \Mmax(\bbf)-1$. Since $\tau_1>0$ and $U_n(s)$ is strictly increasing in $s$, there exists $s_1<\tau_m$ and $s_2\in [\tau_{m+1}, \tau_{m+2}]$ such that $U_m(\tau_m)>U_m(s_1)>U_{m+1}(s_2)>U_{m+1}(\tau_{m+1})$. We also notice that there exists at least one firm $i$ such that, under $\bbf$, firm $i$ interviews some applicants with score in $[\tau_{m+1}, s_2]$ but not all applicants with score in $[s_1, \tau_m]$. If firm $i$ interviews an applicant with score $s\in[s_1, \tau_m]$, the utility firm $i$ can derive is at least $U_{m}(s_1)$, which is strictly higher than the utility derived by interviewing an applicant with score $s\in [\tau_{m+1}, s_2]$. Therefore, firm $i$ can deviate from $f_i$ by moving support from $[\tau_{m+1}, s_2]$ to within $[s_1, \tau_m]$ to earn strictly more utility. This contradicts to the assumption that $f$ is a Nash equilibrium. Therefore, we must have $U_m(\tau_m) \leq U_{m+1}(\tau_{m+1})$ for all $m < \Mmax(\bbf)$.

Finally, we will prove the \textbf{fourth condition}. Suppose $U_n(\tau_n) < U_m(\tau_m)$ for some $n<m\leq \Mmax(\bbf)+1$. Let firm $i$ be a firm such that $f_i(s_1)=1$ and $U_n(s_1)<U_m(\tau_m)$ for some score $s_1\in[\tau_n, \tau_{n+1}]$. We will show that $f_i(s_2) = 1$ for any $s_2\in[\tau_{m-1}, \tau_m]$ with $U_m(s_2)>U_n(s_1)$. Suppose, in contradiction, $f_i(s_2) = 0$ for some $s_2\in[\tau_{m-1}, \tau_m]$ with $U_m(s_2)>U_n(s_1)$. Since the support of firm $i$ can be written as a union of intervals and $U_n(s)$ is continuous, there exists an interval such that $s_1\in [a_1, b_1]$ with $a_1<b_1$ and $f_i(s)=1$ for $s\in[a_1, b_1]$. Similarly, there exists an interval such that $s_2\in[a_2,b_2]$ with $a_2<b_2$, $f_i(s)=0$ for $s\in[a_2, b_2]$, and $U_m(s)>U_n(b_1)$ for $s\in[a_2,b_2]$. We notice that the utility firm $i$ can derive by interviewing an applicant with score $s\in [a_1, b_1]$ is strictly less than the utility firm $i$ can derive by interviewing an applicant with score $s\in [a_2, b_2]$. Therefore, firm $i$ can increase its utility by moving some support from $[a_1, b_1]$ to $[a_2, b_2]$. This is a contradiction to $\bbf$ being a Nash equilibrium. Hence, we must have $f_i(s_2) = 1$ for any $s_2\in[\tau_{m-1}, \tau_m]$ with $U_m(s_2)>U_n(s_1)$.

We will prove the \textbf{other direction} of the statement which states that any strategy profile $\bbf$ that satisfies the four conditions listed is a Nash equilibrium. We want to show that any firm $i$ cannot earn a higher utility by deviating from its current strategy profile. Let $f_i$ be the current strategy profile of firm $i$ and $f'_i$ be an arbitrary alternative strategy profile. We will show that $u(f_i, f_{-i})\geq u(f'_i, f_{-i})$ if $\bbf $ satisfies the four conditions.

Denote $K_i = \{s\in[0,1], f_i(s)=1\}$. Under $\bbf$, we can express the utility of firm $i$ as 
\begin{align}
    u(f_i, f_{-i}) &= \int_{K\cap K'} U_{M(s, \mathbf{f})}(s) \varphi(s)ds + \int_{K\setminus K'}U_{M(s, \mathbf{f})}(s)\varphi(s)ds.  \label{equ:18}
\end{align}
Let $\bbf'$ denote the strategy profile after firm $i$ deviates from $f_i$, and denote $K'_i = \{s\in[0,1], f'_i(s)=1\}$. We notice that $M(s, \mathbf{f'})=M(s, \mathbf{f})+1$ if $s\in K'\setminus K$, and $M(s, \mathbf{f'})=M(s, \mathbf{f})$ if $s\in K\cap K'$. Therefore, the utility under the alternative strategy is 
\begin{align}
    u(f'_i, f_{-i}) & = \int_{K\cap K'}U_{M(s, \mathbf{f'})}(s) \varphi(s)ds + \int_{K'\setminus K}U_{M(s, \mathbf{f'})}(s) \varphi(s)ds \nonumber \\
    & = \int_{K\cap K'}U_{M(s, \mathbf{f})}(s) \varphi(s)ds + \int_{K'\setminus K}U_{M(s, \mathbf{f})+1}(s) \varphi(s)ds.  \label{equ:20}
\end{align}

To show $u(f_i, f_{-i})\geq u(f'_i, f_{-i})$, we only need to show the second part of equation (\ref{equ:20}) is less or equal to the second part of equation (\ref{equ:18}). To show this, we will prove that for any $s\in K_i\setminus K'_i$, $U_{M(s, \mathbf{f})}(s)\geq U_{M(s', \mathbf{f})+1}(s')$ for all $s\in K_i'\setminus K_i$. Let $s\in K_i\setminus K'_i$, $s'\in K_i'\setminus K_i$ be arbitrary. Suppose $s\in [\tau_m, \tau_{m+1}]$ and $s'\in [\tau_n, \tau_{n+1}]$ for some $m,n \in [\Mmax(\bbf)]$. Consider the following two cases: $U_{m}(\tau_m) \geq U_{n+1}(\tau_{n+1})$ or $U_{m}(\tau_m) < U_{n+1}(\tau_{n+1})$.

If $U_{\tau_m}(\tau_m) \geq U_{\tau_{n+1}}(\tau_{n+1})$, by monotonicity of $U_n(s)$, we have 
\begin{align*}
    U_{m}(s)\geq U_{m}(\tau_m)\geq U_{n+1}(\tau_{n+1})\geq U_{n+1}(s').
\end{align*}
Hence $U_{m}(s) \geq U_{n+1}(s')$. 

If $U_{m}(\tau_m) < U_{n+1}(\tau_{n+1})$, by condition 3, we must have $m
<n+1$. Suppose, by contradiction, $U_{m}(s) < U_{n+1}(s')$. By condition 4, given that $s\in K_i\setminus K'_i$,  we must have $s'\in K_i$ as well. This is a contradiction to $s'\in K_i'\setminus K_i$. Hence we must have $U_{m}(s) \geq U_{n+1}(s')$.

Therefore, we conclude that for any $s\in K_i\setminus K'_i$, $s'\in K_i'\setminus K_i$ , we must have $U_{M(s, \mathbf{f})}(s)\geq U_{M(s', \mathbf{f})+1}(s')$. As a result, $\min_{s\in K_i\setminus K'_i}U_{M(s, \mathbf{f})}(s)\geq \max_{s'\in K_i'\setminus K_i}U_{M(s', \mathbf{f})+1}(s')$. Furthermore, by condition 1, $\Pr(S\in K_i) = \Pr(S\in K'_i) = c$. Hence, 
\begin{align*}
    \Pr(S\in K_i\setminus K'_i) &=\Pr(S\in K_i) - \Pr(S\in K_i\cap K'_i)\\
    &= Pr(S\in K'_i)- \Pr(S\in K_i\cap K'_i)\\
    & = \Pr(S\in K'_i\setminus K_i)
\end{align*}
We can rewrite equation (\ref{equ:20}) as follows:
\begin{align*}
    u(f'_i, f_{-i}) & \leq \int_{K_i\cap K'_i}U_{M(s, \mathbf{f})}(s) \varphi(s)ds + \Pr(S\in K'_i\setminus K_i)(\max_{s'\in K_i'\setminus K_i}U_{M(s', \mathbf{f})+1}(s') )\\
    & \leq \int_{K_i\cap K'_i}U_{M(s, \mathbf{f})}(s) \varphi(s)ds + \Pr(S\in K_i\setminus K'_i)(\min_{s\in K_i\setminus K'_i}U_{M(s, \mathbf{f})}(s))\\
    & \leq \int_{K\cap K'} U_{M(s, \mathbf{f})}(s) \varphi(s)ds + \int_{K\setminus K'}U_{M(s, \mathbf{f})}(s)\varphi(s)ds\\
    & = u(f_i, f_{-i}).
\end{align*}

Therefore, $u(f'_i, f_{-i}) \leq u(f_i, f_{-i})$ for any arbitrary firm $i$ and alternative strategy $f'_i$. We conclude that none of the firms is able to increase their utility by deviating from the current strategy profile, and $\bbf$ is at Nash equilibrium.

\end{myproof}

\subsection{Proof of Proposition 3.1}
\begin{myproof}

Fix an instance $I=(N,c,D,\theta)$ and suppose that the utility functions
$\{U_m\}_{m=1}^N$ satisfy Assumption~2.1. We will show that we can always find an equal-utility Nash equilibrium strategy profile.

Fix a scalar $y \in [0,U_1(1)]$. For each $m\in\{1,\dots,N\}$ define
\[
\tau_m(y)
\;\equiv\;
\inf\{s\in[0,1]: U_m(s)\ge y\}.
\]
By Assumption~2.1, each $U_m$ is continuous and strictly increasing with
$U_m(0)=0$, so $\tau_m(y)$ is well-defined and satisfies
$U_m(\tau_m(y))=y$ whenever $\tau_m(y)\in(0,1)$. Moreover, since utilities are
weakly decreasing in $m$, we have
\[
\tau_1(y)\le \tau_2(y)\le \cdots \le \tau_N(y).
\]
Let $\tau_0(y)=0$, $\tau_{N+1}(y)=1$ and  $\Phi(s)=\int_0^s \varphi(t)\,dt$ be the score CDF. If applicants with
scores in $[\tau_m(y),\tau_{m+1}(y))$ are interviewed by exactly $m$ firms, the
total interview mass is
\[
T(y)
\;\equiv\;
\sum_{m=1}^N
m\bigl(\Phi(\tau_{m+1}(y))-\Phi(\tau_m(y))\bigr).
\]
Equivalently,
\[
T(y)=\sum_{k=1}^N \Pr\bigl(S\ge \tau_k(y)\bigr).
\]

Each mapping $y\mapsto \tau_m(y)$ is continuous, and hence $T(y)$ is continuous.
Moreover, $T(y)$ is weakly decreasing in $y$. At the endpoints,
\[
T(0)=N,
\qquad
T(U_1(1))=0.
\]
Since $Nc\in(0,N]$, by the Intermediate Value Theorem there exists
$y^\star\in[0,U_1(1)]$ such that
\[
T(y^\star)=Nc.
\]
Define $\tau_m^\star\equiv\tau_m(y^\star)$.

Next, we will construct a measurable strategy profile
$\bbf=(f_1,\dots,f_N)$ such that:
\begin{enumerate}
\item For each $m$, applicants with scores
$s\in[\tau_m^\star,\tau_{m+1}^\star)$ are interviewed by exactly $m$ firms.
\item Each firm interviews total mass exactly $c$.
\end{enumerate}
For each interval $[\tau_m^\star,\tau_{m+1}^\star)$, partition it into $N$
measurable subintervals of equal $\varphi$-mass, and on subinterval $j$ assign
the interviewing set $\{j,j+1,\dots,j+m-1\}$ modulo $N$. Then each applicant in
that interval is interviewed by exactly $m$ firms, and symmetry implies that
each firm receives a fraction $m/N$ of the mass of that interval. Therefore,
each firm $i$ satisfies
\[
\Pr(f_i(S)=1)
=
\frac{1}{N}\sum_{m=1}^N
m\bigl(\Phi(\tau_{m+1}^\star)-\Phi(\tau_m^\star)\bigr)
=
\frac{T(y^\star)}{N}
=
c.
\]

Since the conditions in Theorem \ref{thm:equilibrium} are all satisfied, $\bbf$ is an equal-utility Nash equilibrium.

\end{myproof}

\subsection{Proof of Proposition 3.2}
\begin{myproof}
    Let $\cD$ be an arbitrary distribution and fix the capacity $c\in (0,1]$. Suppose $\cI(N) = (N, c, \cD, \corr)$ is an instance with $N$ firms and $\mathbf{f_N}$ is a Nash equilibrium for the instance. 
    
    We will first show that $\tau_m \geq m \tau_1$ for all $m\in [\Mmax(\mathbf{f_N})]$. For $m =1$, we have $\tau_m = \tau_1$ and the inequality holds for $m=1$. For $1<m\leq \Mmax(\mathbf{f_N})$, the third condition in Theorem \ref{thm:equilibrium} gives $U_1(\tau_1)\leq U_m(\tau_m)$.  Under the correlated decision rule, the utility function can be expressed as  $U_n(s) = s/n$. Therefore, the above condition can be written as $\tau_1\leq \frac{\tau_m}{m}$. Hence, $m\tau_1\leq \tau_m$ for all $1<m\leq \Mmax(\bbf)$. Therefore, we can conclude that $\tau_m \geq m \tau_1$ for all $m\in [\Mmax(\mathbf{f_N})]$. 

    Next, we want to show that if $\mathbf{f_N}$ is an equal-utility Nash equilibrium, then $\tau_m = m \tau_1$ for all $m\in [\Mmax(\mathbf{f_N})]$. For $m =1$, we have $\tau_m = \tau_1$ and the equality holds. For $1<m\leq \Mmax(\mathbf{f_N})$, by Definition \ref{def:types_of_NE}, $U_1(\tau_1)= U_m(\tau_m)$. Since the instance is under the correlated decision rule, we have $m\tau_1= \tau_m$ for all $1<m\leq \Mmax(\mathbf{f_N})$. Hence, we've shown that $\tau_m = m \tau_1$ for all $m\in [\Mmax(\mathbf{f_N})]$.
\end{myproof}

\subsection{Derivation of $U_n(s)$ when $\theta = \indep$} \label{deri:U_n_indep}
When hiring decisions are independent, the probability that a firm successfully hires an applicant of score $s$ if they interview them is

\begin{align}
    U_n(s) &= s(\sum_{i=0}^{n-1} \binom{n-1}{i} s^i (1-s)^{n-1-i} \frac{1}{i+1}) \nonumber\\
    & = \frac{s}{ns}(\sum_{i=0}^{n-1} \binom{n}{i+1} s^{i+1} (1-s)^{n-1-i})\nonumber\\
    & = \frac{1}{n}(1-(1-s)^n) \label{equ:Un(s)_indep}.
\end{align}

\subsection{Proof of Proposition \ref{prop:NE_indep}}
Before proving Proposition \ref{prop:NE_indep}, let's first prove the following lemmas:

\begin{lemma} \label{lemma:3.2.1}
     For any distribution $\cD$, capacity $c \in (0, 1)$ and decision rule $\theta \in \{\indep, \corr\}$,
    let $\cI(N) = (N, c, \cD, \theta)$ is the instance parameterized by  number of firms $N$. Let $\mathbf{f_N}$ be a Nash equilibrium for instance $\cI(N)$. Then, $lim_{N\rightarrow \infty} \Mmax(\mathbf{f_N}) = \infty$.
\end{lemma}

\begin{myproof}
    We will first show that $Nc\leq \Mmax(\mathbf{f_N})$. By condition 1 of Theorem \ref{thm:equilibrium}, $\Pr(f_i(S)=1)=c$ for each firm $i$. As a result, the total capacity of $N$ firms is $\sum_{i=1}^N \Pr(f_i(S)=1) = Nc$. By condition 2 of Theorem \ref{thm:equilibrium}, there exists thresholds $\tau_1, \tau_2, \cdots, \tau_{\Mmax(\mathbf{f_N})}$ such that $M(s; \mathbf{f_N}) = m$ for $s\in [\tau_m, \tau_{m+1})$ for $m = [\Mmax(\mathbf{f_N})]$. Therefore, the total capacity of $N$ firms can also be written as follows:

    \begin{align*}
        Nc & = \Pr(S\in [\tau_1, \tau_2]) + 2 \Pr(S\in [\tau_2, \tau_3])\cdots + \Mmax(\mathbf{f_N}) \Pr(S\in [\tau_{\Mmax(\bbf)}, 1])\\
        &= \int_{\tau_1}^{\tau_2} \varphi(s)ds + 2\int_{\tau_2}^{\tau_3} \varphi(s)ds  \cdots +\Mmax(\mathbf{f_N})\int_{\tau_{\Mmax(\mathbf{f_N})}}^{1} \varphi(s)ds\\
        & \leq \Mmax(\mathbf{f_N})\int_{\tau_1}^{\tau_2} \varphi(s)ds + \Mmax(\mathbf{f_N})\int_{\tau_2}^{\tau_3} \varphi(s)ds  \cdots +\Mmax(\mathbf{f_N})\int_{\tau_{\Mmax(\bbf)}}^{1} \varphi(s)ds\\
        & = \Mmax(\mathbf{f_N})\int_{\tau_1}^{1} \varphi(s)ds \\
        & \leq \Mmax(\mathbf{f_N})\int_{0}^{1} \varphi(s)ds\\
        & = \Mmax(\mathbf{f_N}).
    \end{align*}

Hence, $Nc\leq \Mmax(\mathbf{f_N})$. As $N$ goes to infinity, for the inequality to hold, we must have $\Mmax(\mathbf{f_N})$ go to infinity as well. Therefore, $\lim_{N\rightarrow \infty} \Mmax(\mathbf{f_N}) = \infty$.     
\end{myproof}

\begin{lemma} \label{lemma:3.2.2}
    For any distribution $\cD$ and capacity $c \in (0, 1)$,
    let $\cI(N) = (N, c, \cD, \indep)$ is the instance parameterized by  number of firms $N$. Let $\mathbf{f_N}$ be a Nash equilibrium for instance $\cI(N)$. Then $\lim_{N\rightarrow \infty} \tau_{\Mmax(\bbf)-1}= 0$.
\end{lemma}

\begin{myproof}
    Suppose $\bbf$ is a Nash equilibrium. $U(\tau_{\Mmax(\mathbf{f_N})-1}, \Mmax(\bbf)-1) =t$. When the hiring decisions are independent, by equation \ref{equ:Un(s)_indep} we have 
    \begin{align*}
        U(\tau_{\Mmax(\mathbf{f_N})-1}, \Mmax(\mathbf{f_N})-1) = \frac{1}{\Mmax(\mathbf{f_N})-1}(1-(1-\tau_{\Mmax(\mathbf{f_N})-1})^{\Mmax(\mathbf{f_N})-1}) = t.
    \end{align*}
    As a result, we can express $\tau_{\Mmax-1}$ as
    \begin{align*}
        \tau_{\Mmax(\mathbf{f_N})-1} = 1-(1-(\Mmax(\mathbf{f_N})-1)t)^{\frac{1}{\Mmax(\mathbf{f_N})-1}}.
    \end{align*}
    Define $f(t) = 1-(1-(\Mmax(\mathbf{f_N})-1)t)^{\frac{1}{\Mmax(\mathbf{f_N})-1}}$, we will show $f(t)$ is strictly decreasing in $t$. 
    
    We will first show that $t\leq\frac{1}{\Mmax(\mathbf{f_N})}$. Suppose, by contradiction, $t>\frac{1}{\Mmax(\mathbf{f_N})}$. Since $\frac{1}{\Mmax(\mathbf{f_N})}=U(1,\Mmax(\mathbf{f_N}))$, we must have $U(\tau_{\Mmax(\mathbf{f_N})-1}, \Mmax(\mathbf{f_N})-1)>U(1,\Mmax(\mathbf{f_N}))$. By condition 4 of Theorem \ref{thm:equilibrium}, $\mathbf{f_N}$ is not a Nash equilibrium, which leads to a contradiction. Therefore, we must have $t\leq\frac{1}{\Mmax(\mathbf{f_N})}$.

    Next, we will show that $f(t)$ is strictly increasing in t. The derivative of $f(t)$ is

    \begin{align*}
        f'(t) &= (\Mmax(\mathbf{f_N})-1)\frac{1}{\Mmax(\mathbf{f_N})-1}(1-(\Mmax(\mathbf{f_N})-1)t)^{\frac{2-\Mmax(\mathbf{f_N})}{\Mmax(\mathbf{f_N})-1}}\\
        & =(1-(\Mmax(\mathbf{f_N})-1)t)^{\frac{2-\Mmax(\mathbf{f_N})}{\Mmax(\mathbf{f_N})-1}}.
    \end{align*}

    Given that $t\leq\frac{1}{\Mmax(\bbf)}$, we have 
    \begin{align*}
        1-(\Mmax(\mathbf{f_N})-1)t\geq 1-(\Mmax(\mathbf{f_N})-1)\frac{1}{\Mmax(\mathbf{f_N})} = \frac{1}{\Mmax(\mathbf{f_N})} > 0.
    \end{align*}

    Hence, $f'(t)>0$ and $f(t)$ increases with $t$. Given that $t\leq\frac{1}{\Mmax(\mathbf{f_N})}$, we get

    \begin{align*}
        \tau_{\Mmax(\mathbf{f_N})-1} \leq 1-(1-\frac{\Mmax(\mathbf{f_N})-1}{\Mmax(\mathbf{f_N})})^{\frac{1}{\Mmax(\mathbf{f_N})-1}} = 1-\frac{1}{\Mmax(\mathbf{f_N})^{\frac{1}{\Mmax(\mathbf{f_N})-1}}}.
    \end{align*}

    Taking the logarithm of $\Mmax(\mathbf{f_N})^{\frac{1}{\Mmax(\mathbf{f_N})-1}}$ we can show that $\lim_{\Mmax(\mathbf{f_N})\rightarrow \infty} \Mmax(\mathbf{f_N})^{\frac{1}{\Mmax(\mathbf{f_N})-1}} = 1$. Hence, $\lim_{\Mmax(\mathbf{f_N})\rightarrow \infty} \tau_{\Mmax(\mathbf{f_N})-1}= 0$. By Lemma \ref{lemma:3.2.1}, $\Mmax(\mathbf{f_N})\rightarrow \infty$ as $N \rightarrow \infty$, we must have
    \begin{align*}
        \lim_{N\rightarrow \infty} \tau_{\Mmax(\bbf)-1}= 0.
    \end{align*}
\end{myproof}

\begin{myproof} [\text{Proof of Proposition \ref{prop:NE_indep}}]
    We first notice that $|M(s; \mathbf{f_N}) - \Mmax(\mathbf{f_N})|>1$ if and only if $M(s; \mathbf{f_N})< \Mmax(\mathbf{f_N})-1$ for any $s\in [0,1]$. By condition 2 of Theorem \ref{thm:equilibrium}, there exists thresholds $\tau_1, \tau_2, \cdots, \tau_{\Mmax(\bbf)}$ such that $M(s; \mathbf{f_N}) = m$ for $s\in [\tau_m, \tau_{m+1}]$ for $m = [\Mmax(\bbf)]$. Therefore, $M(s; \mathbf{f_N})< \Mmax(\mathbf{f_N})-1$ is equivalent to $S\in[0,\tau_{\Mmax(\mathbf{f_N})-1})$.
    \begin{align*}
        \Pr(|M(S; \mathbf{f_N}) - \Mmax(\mathbf{f_N})|>1) 
        & = \Pr(M(S; \mathbf{f_N})< \Mmax(\mathbf{f_N})-1)\\
        & = \Pr(S\in[0,\tau_{\Mmax(\mathbf{f_N})-1})).
    \end{align*}
    By lemma \ref{lemma:3.2.2}, $\lim_{N\rightarrow \infty} \tau_{\Mmax(\bbf)-1}= 0$. As a result, we must have $\lim_{N\rightarrow \infty} \Pr(S\in[0,\tau_{\Mmax(\mathbf{f_N})-1})) = 0$ and $\lim_{N\rightarrow \infty} \Pr(|M(S; \mathbf{f_N}) - \Mmax(\mathbf{f_N})|>1) = 0$ as desired.
\end{myproof}

\subsection{Proof of Theorem \ref{thm:SW1}}

\begin{myproof}
Let $(f_1^{Naive},\dots,  f_N^{Naive})$ be the strategy profile for the Naive solution, and let  $(f_1^{NE},\dots,  f_N^{NE})$ be a Nash equilibrium solution.
We will show $u(f_i^{NE}, f_{-i}^{NE}) > u(f_i^{Naive}, f_{-i}^{Naive})$ for every firm $i \in [N]$, which will imply $\SW_{\text{NE}} > \SW_\text{naive}$.
Let $s_c = \min \{s \in [0, 1]: f_1^{Naive}(s) = 1\}$ be the score threshold in which the naive strategy interviews everyone above $s_c$.

Fix a firm $i \in [N]$.
We will show the following two inequalities:
\begin{align} \label{eq:twoinequalities}
u(f_i^{NE}, f_{-i}^{NE}) 
\geq u(f_i^{Naive}, f_{-i}^{NE})
\geq u(f_i^{Naive}, f_{-i}^{Naive}).
\end{align}
The first inequality holds by definition of a Nash equilibrium.
The second inequality holds since if firm $i$ chooses the naive strategy $f_i^{Naive}$, its utility is the lowest when all other firms choose the naive strategy.
This is because for any score $s \geq s_c$, the total number of firms competing for applicant $s$ is $N$ under $(f_i^{Naive}, f_{-i}^{Naive})$, but may be smaller under $(f_i^{Naive}, f_{-i}^{NE})$.

We will show that at least one of the two inequalities in \eqref{eq:twoinequalities} is strict. Suppose, by contradiction, that both are equalities.
We will first show that $f_{-i}^{NE} = f_{-i}^{Naive}$.
Under $(f_i^{Naive}, f_{-i}^{Naive})$, every applicant that firm $i$ interviews is interviewed by $N$ firms in total.
For the equality $u(f_i^{Naive}, f_{-i}^{NE}) = u(f_i^{Naive}, f_{-i}^{Naive})$ to hold, it must be that under the strategy profile $(f_i^{Naive}, f_{-i}^{NE})$, every applicant that firm $i$ interviews is also interviewed by $N$ firms in total. If there was an applicant interviewed by strictly fewer firms, then it would be that $u(f_i^{Naive}, f_{-i}^{NE}) > u(f_i^{Naive}, f_{-i}^{Naive})$, since the utility function $U_n(s)$ are strictly monotonic in $n$.
Therefore, it must be that $f_{-i}^{NE} = f_{-i}^{Naive}$.

We will now show that firm $i$ can deviate from $f_i^{NE}$ to strictly improve their utility.
Since we assume $u(f_i^{NE}, f_{-i}^{NE}) 
= u(f_i^{Naive}, f_{-i}^{NE})$, and that $f_{-i}^{NE} = f_{-i}^{Naive}$, we will show that firm $i$ can strictly improve from $f_i^{Naive}$, when all other firms are at $f_{-i}^{Naive}$.
Let $s' < s_c$ be a score such that $U_{1}(s') > U_N(s_c)$, which exists since $U_{1}(s_c) > U_{N}(s_c)$ and $U_{1}(s)$ is a continuous and strictly increasing function in $s$.
Under $(f_i^{Naive}, f_{-i}^{Naive})$, no firm is interviewing scores in $[s', s_c]$.
Then, if firm $i$ interviews an applicant with score $s \in [s', s_c]$, their utility is strictly higher than $U_N(s_c)$.
Therefore, there exists an $\epsilon > 0$ where firm $i$ can deviate from the naive strategy by moving support from $[s_c, s_c+\epsilon]$ to within $[s', s_c]$, where they would earn strictly more utility.
This is a contradiction to $(f_i^{NE}, f_{-i}^{Ne})$ being a Nash equilibrium, and hence it must be that one of the two inequalities in \eqref{eq:twoinequalities} is strict.
Therefore, $\SW_{\text{NE}} > \SW_\text{naive}$.
\end{myproof}

\subsection{Proof of Theorem \ref{thm:SW_PoNS2}}
We will first prove the following lemma before proving the theorem.

\begin{lemma} \label{lemma:4.2.1}
For any distribution $\cD$ and decision rule $\theta\in\{\corr, \indep\}$, let $\cI(c,N) = (N, c, \cD, \theta)$ be the instance parameterized by capacity $c$ and number of firms $N$, and $\mathbf{f_c}$ be a Nash equilibrium for instance $\cI(c,N)$. If $Nc<\int_{0.5}^1 \varphi(s)ds$, we must have $\Mmax(\mathbf{f_c}) = 1$ and $\Pr(S\in[\tau_1, 1]) = Nc$. The social welfare under $\mathbf{f_c}$ is 
    \begin{align*} 
        \SW_\text{NE}(\cI(c)) & =\int_{\tau_1}^{1}U_1(s)\varphi(s)ds. 
    \end{align*}
\end{lemma}
\begin{myproof}
    Suppose $Nc<\int_{0.5}^1\varphi(s)ds$ and $\mathbf{f_c}$ is a Nash equilibrium. The lowest possible utility a firm can derive without double-interviewing any applicant is 0.5, which is higher than the highest possible utility a firm can derive by double-interviewing any applicant. As a result, we must have $\Mmax(\mathbf{f_c}) = 1$. Since $\mathbf{f_c}$ is a Nash equilibrium, by condition 1 and 2 of Theorem \ref{thm:equilibrium}, $M(s, \mathbf{f_c}) =1$ on $[\tau_1, 1]$ and $\Pr(S\in[\tau_1, 1]) = Nc$. Therefore, the social welfare under $\mathbf{f_c}$ is 
    \begin{align*} 
        \SW_\text{NE}(\cI(c)) & =\int_{\tau_1}^{1}U_1(s)\varphi(s)ds 
    \end{align*} as desired. 
    \end{myproof}

\begin{myproof} [\text{Proof of Proposition \ref{thm:SW_PoNS2}}]
    Let $\mathbf{f_c}$ be an Nash equilibrium of the instance $\cI(c)$. We will first show $\lim_{c \to 0^+} \PoNS(\cI(c)) = N$. If $N$ is fixed, as $c \to 0^+$, eventually we will have $c<\int_{0.5}^1\varphi(s)ds$. By lemma \ref{lemma:4.2.1}, $\lim_{c \to 0^+} \Mmax(\mathbf{f_c})  = 1$ and $\Pr(S\in[\tau_1, 1]) = Nc$. The social welfare under $\mathbf{f_c}$ is 
    \begin{align*} 
        \SW_\text{NE}(\cI(c)) & =\int_{\tau_1}^{1}U_1(s)\varphi(s)ds. 
    \end{align*}
    Let $s_c = \min \{s \in [0, 1]: f_1^{\text{naive}}(s) = 1\}$ be the score threshold in which the naive strategy interviews everyone above $s_c$. The social welfare under the naive solution is 
    \begin{align*}
        \SW_\text{naive}(\cI(c)) & =N \int_{s_c}^1 U_N(s)\varphi(s)ds. 
    \end{align*}
    
    We also notice that $\lim_{c \to 0^+}U_n(s) = \frac{1}{n}$. Therefore, when $c \to 0^+$ we have 
    \begin{align*} 
        \lim_{c \to 0^+}\PoNS(\cI(c))  = \lim_{c \to 0^+}\frac{\SW_\text{NE}(\cI(c))}{\SW_\text{naive}(\cI(c))} = \frac{\int_{\tau_1}^{1} \varphi(s)ds}{N\int_{s_c}^{1} \frac{1}{N} \varphi(s)ds} = \frac{\Pr(S\in[\tau_1, 1])}{\Pr(S\in[s_c, 1])}=\frac{Nc}{c} = N.
    \end{align*}

    Next, we will show $\lim_{c \to 1^-} \PoNS(\cI(c)) = 1$. As $c \to 1^-$, by condition 1 of Theorem \ref{thm:equilibrium}, $\lim_{c \to 1^-}\Pr(f_i(S)=1) = 1$. Therefore, each firm will interview all the applicants, and we must have $\tau_1 = \cdots = \tau_N = 0$. Hence,  $\lim_{c \to 1^-} \SW_\text{NE}(\cI(c)) = \int_{0}^{1} U_N(s) \varphi(s)ds$. Similarly, $s_c = 0$ and $\lim_{c \to 1^-} \SW_\text{naive}(\cI(c)) = \int_{0}^{1} U_N(s) \varphi(s)ds$. Therefore, we have
    \begin{align*}
        \lim_{c \to 1^-} \PoNS(\cI(c))  = \lim_{c \to 1^-}\frac{\SW_\text{NE}(\cI(c))}{\SW_\text{naive}(\cI(c))} =  \frac{\int_{0}^{1} U_N(s) \varphi(s)ds}{\int_{0}^{1} U_N(s) \varphi(s)ds} = 1.
    \end{align*}
\end{myproof}

\subsection{Proof of Theorem \ref{thm:SW_PoNS4}}
\begin{myproof}
    Let $\cI(N) = (N, c, \cD, \indep)$ be an instance with $N$ firms and the independent decision rule. Let $s_c = \min \{s \in [0, c]: f_1^{\text{naive}}(s) = 1\}$ be the score threshold in which the naive strategy interviews everyone above $s_c$. 
    \begin{align*}
        \PoNS(\cI(N)) &= \frac{\SW_{\text{NE}}(\cI(N))}{\SW_{\text{naive}}(\cI(N))} \\
        &= \frac{\int_{\tau_1}^{\tau_2} U_1(s)\varphi(s)ds + 2\int_{\tau_2}^{\tau_3}U_2(s) \varphi(s)ds +\cdots n\int_{\tau_n}^{1}U_{\Mmax(\bbf)}(s) \varphi(s)ds}{N\int_{s_c}^1 U_N(s)\varphi(s)ds}.
    \end{align*}
    By lemma \ref{lemma:3.2.2}, $\lim_{N\rightarrow \infty}\tau_{\Mmax(\bbf)-1}\rightarrow 0$ and $U_N(s)\rightarrow \frac{1}{N}$. Therefore  
    \begin{align*}
        \lim_{N\rightarrow \infty}\PoNS(\cI(N)) 
        &= \frac{(\Mmax(\bbf)-1)\int_{0}^{\tau_{\Mmax(\bbf)}}\frac{1}{\Mmax(\bbf)-1} \varphi(s)ds + \Mmax(\bbf)\int_{\tau_{\Mmax(\bbf)}}^{1}\frac{1}{\Mmax(\bbf)} \varphi(s)ds}{N\int_{s_c}^1 \frac{1}{N}\varphi(s)ds}\\
        & = \frac{\int_{0}^1 \varphi(s)ds}{\int_{s_c}^1 \varphi(s)ds}\\
        &= \frac{1}{c}.
    \end{align*}

\end{myproof}

\subsection{Proof of Theorem \ref{thm:SW_PoNS3}}
Before proving Theorem \ref{thm:SW_PoNS3}, we first present the following intermediate result. 

\begin{lemma} \label{lemma:4.4.1}
    For any distribution $\cD$ and decision rule $\theta = \{\corr, \indep\}$, let $\cI(N) = (N, c, \cD, \corr)$ is the instance with $N$ firms and $\mathbf{f_N}$ be the nash equilibrium of the instance. Let $\tau_1$ be the threshold defined in Theorem \ref{thm:equilibrium}, $\lim_{N\rightarrow \infty} \tau_1 = 0$.
\end{lemma}
\begin{myproof} 
    By condition 3 of Theorem \ref{thm:equilibrium}, $U_1(\tau_1)\leq U_{\Mmax(\mathbf{f_N})}(\tau_{\Mmax(\mathbf{f_N})})$. Since $U_n(s)$ is strictly increasing in $s$, we have $U_{\Mmax(\mathbf{f_N})}(\tau_{\Mmax(\mathbf{f_N})})<U_{\Mmax(\mathbf{f_N})}(1)$. Given that $U_1(\tau_1) = \tau_1$ and $U_{\Mmax(\mathbf{f_N})}(1) = \frac{1}{\Mmax(\mathbf{f_N})}$, we have
    \begin{align*}
        \tau_1 \leq U_{\Mmax(\mathbf{f_N})}(\tau_{\Mmax(\mathbf{f_N})}) < \frac{1}{\Mmax(\mathbf{f_N})}.
    \end{align*}
    By lemma \ref{lemma:3.2.1}, $\lim_{N\rightarrow \infty} \frac{1}{\Mmax(\mathbf{f_N})} = 0$. As a result, $\lim_{N\rightarrow \infty} \tau_1 = 0$.
\end{myproof}

Next, we present the following convergence result for general distribution $\cD$.
\begin{lemma}
    \label{lemma:SW_PONS3}
    For any distribution $\cD$, capacity $c$, if $\cI(N) = (N, c, \cD, \corr)$ is the instance with $N$ firms and the correlated decision rule, then $\PoNS(\cI(N))$ increases when $N$ increases, and $\lim_{N\rightarrow \infty}\PoNS(\cI(N)) = \frac{\mathbb{E}S}{\int_{s_c}^1 s\varphi(s)ds}$.
\end{lemma}

    \begin{myproof}
        Let $\cI(N) = (N, c, \cD, \corr)$ be an instance with $N$ firms and the correlated decision rule. 

    We will first show that $\PoNS(\cI(N))$ increases when $N$ increases. Under the correlated decision rule, the utility function takes the form $U_n(s) = s/n$. The social welfare under the Nash equilibrium can be written as:
    \begin{align} 
        \SW_\text{NE}(\cI(N)) & =\int_{\tau_1}^{\tau_2}U_1(s)\varphi(s)ds + \cdots + \Mmax \int_{\tau_{\Mmax}}^{1}U_{\Mmax}(s)\varphi(s)ds \nonumber\\
        & = \int_{\tau_1}^{\tau_2} s \varphi(s)ds + \cdots + \Mmax \int_{\tau_{\Mmax}}^{1}\frac{s}{\Mmax}\varphi(s)ds \nonumber\\
        & =  \int_{\tau_1}^{\tau_2} s \varphi(s)ds + \cdots +  \int_{\tau_{\Mmax}}^{1}s\varphi(s)ds \nonumber\\
        & =  \int_{\tau_1}^{1} s \varphi(s)ds. \label{eq:SW_corr}
    \end{align}
    Let $s_c = \min \{s \in [0, c]: f_1^{\text{naive}}(s) = 1\}$ be the score threshold in which the naive strategy interviews everyone above $s_c$. The social welfare under the naive solution is 
    \begin{align} \label{sw_correlated_naive_formula}
        \SW_\text{naive}(\cI(N)) & =\int_{s_c}^1 U_N(s)\varphi(s)ds 
         =  N\int_{s_c}^{1} \frac{s}{N}\varphi(s)ds
         =  \int_{s_c}^{1} s \varphi(s)ds.
    \end{align}
    
    Therefore, we can express $\PoNS(\cI(N))$ as follows:
    \begin{align} \label{eq:PoNS}
        \PoNS(\cI(N))  = \frac{\SW_\text{NE}(\cI(N))}{\SW_\text{naive}(\cI(N))} = \frac{\int_{\tau_1}^{1} s \varphi(s)ds}{\int_{s_c}^{1} s \varphi(s)ds}.
    \end{align}

    By lemma \ref{lemma:4.4.1}, $\tau_1$ decreases as $N$ increases. Since $s\varphi(s)>0$,  $\int_{\tau_1}^{1} s \varphi(s)ds$ increases as $\tau_1$ decreases. Given that the denominator of  (\ref{eq:PoNS}) is fixed, we can conclude that $\PoNS(\cI(N)) $ increases as $N$ increases. 

    Next, we want show that in the limit when $N$ goes to infinity,  $\PoNS(\cI(N))$ converges to $\frac{\mathbb{E}S}{\int_{s_c}^1 s\varphi(s)ds}$. By lemma \ref{lemma:4.4.1}, $\lim_{N\rightarrow \infty} \tau_1 = 0$. Therefore, the numerator of the equation (\ref{eq:PoNS}) converges to $\int_{0}^{1} s \varphi(s)ds = \mathbb{E} S$. We conclude that 
    \begin{align}\label{equ:pons_IN}
        \lim_{N\rightarrow \infty}\PoNS(\cI(N)) = \frac{\mathbb{E}S}{\int_{s_c}^1 s\varphi(s)ds}.
    \end{align}
    \end{myproof}

\begin{myproof} [\text{Proof of Theorem \ref{thm:SW_PoNS3}}]
    Next, we will consider the special case where $\cD = \unif$. Under $\cD = \unif$, $\mathbb{E}S = 0.5$ and equation (\ref{equ:pons_IN}) can be written as 
    \begin{align*}
        \lim_{N\rightarrow \infty}\PoNS(\cI(N)) = \frac{\frac{1}{2}}{\int_{1-c}^1 sds} = \frac{1}{2c-c^2}. 
    \end{align*}
    Therefore, under the uniform distribution, $\lim_{N\rightarrow \infty}\PoNS(\cI(N)) = \frac{1}{2c-c^2} $. We only need to show that if $Nc\leq 0.5$, then $\PoNS(\cI(N)) = 1.5$. By lemma \ref{lemma:4.2.1}, when $Nc\leq \int_{0.5}^1 \varphi(s)ds = 0.5$, $\Mmax(\mathbf{f_N}) = 1$ and $\Pr(S\in [\tau_1, 1]) = Nc$. Since $\cD = \unif$, $\tau_1= 1-Nc$. By equation (\ref{eq:PoNS}), 
    \begin{align*}
        \PoNS(\cI(N))  &= \frac{1-\tau_1^2}{1-(1-c)^2} \\
        & = \frac{2Nc-N^2c^2}{2c-c^2}\\
        & = \frac{2N-N^2c}{2-c}\\
        &\geq \frac{2N-N/2}{2-c} \\
        &= \frac{3N}{2(2-c)} \\
        &\geq \frac{6}{2(2-c)} \\
        &\geq \frac{3}{2} = 1.5.
    \end{align*}
    Hence, under the uniform distribution, if $Nc\leq 0.5$, $ \PoNS(\cI(N))\geq 1.5$.
\end{myproof}

\subsection{Proof of Theorem \ref{thm:PoA}}
We will first derive an expression for the highest possible social welfare before proving the theorem.
\begin{lemma} \label{lemma:sw_max}
For any distribution $\cD$,
    let $\cI(c, N) = (N, c, \cD, \corr)$ be the instance parameterized by the capacity $c$ and the number of firms $N$. Suppose $\Pr(S\in [s_{Nc},1])=Nc$, then
    \begin{align*}
    \SW_{\max}(\cI(c,N))=
    \begin{cases}
        \int_{s_{Nc}}^{1} s \varphi(s)ds, & \text{if } Nc\leq 1\\
        \E[S].  &o.w.
        \end{cases}
    \end{align*}
    
\end{lemma}
\begin{myproof}
    Let $\cI(c, N) = (N, c, \cD, \theta)$ be the instance parameterized by the capacity $c$ and number of firms $N$. Let $s_{Nc}\in [0,1]$ such that $ \Pr(S\in [s_{Nc}, 1]) = Nc$. Our goal is to show that for any arbitrary strategy profile with the same capacity and number of firms, the social welfare is at most $\int_{s_{Nc}}^{1} s \varphi(s)ds$ when $Nc\leq1$, and $\E[S]$ when $Nc> 1$.
    
    Let's first consider the case when $Nc\leq 1$. We will first show that $\SW_{\text{max}}(\cI(c,N))\geq \int_{s_{Nc}}^1s\varphi(s)ds$. First, since  $\Pr(S\in [s_{Nc}, 1]) = Nc$ and $Nc\leq 1$, there exists $s_{Nc} = s_1<s_2\cdots <s_N<s_{N+1}=1$ such that $\Pr(S\in[s_i, s_{i+1}]) = c$. Let $\bbf$ be a strategy such that $f_i(s)=1$ if and only if $s\in [s_i, s_{i+1}]$ for $i\in [N]$. Under $\bbf$, we have $M(s, \bbf) = 1$ for $s\in[s_{Nc}, 1]$ and  $M(s, \bbf) = 0$ otherwise. Therefore, the social welfare is $ \int_{s_{Nc}}^1s\varphi(s)ds$. Hence we must have $\SW_{\text{max}}(\cI(c,N))\geq \int_{s_{Nc}}^1s\varphi(s)ds$.

    We will then show $\SW_{\text{max}}(\cI(c,N))\leq \int_{s_{Nc}}^1s\varphi(s)ds$. Let $\mathbf{f_\text{arb}}$ be an arbitrary strategy profile with $N$ firms each with capacity $c$. Define $S_n = \{s\in[0,1], M(s;\mathbf{f_\text{arb}}) = n\}$. The social welfare under $\mathbf{f_\text{arb}}$ is
    \begin{align}
        \SW_{\text{arb}}(\cI(c,N)) & = N\int_{S_N} U_N(s) \varphi(s)ds + (N-1)\int_{S_{N-1}} U_{N-1}(s) \varphi(s)ds +\cdots + \int_{S_1} s \varphi(s)ds \nonumber\\
        &= N\int_{S_N} \frac{s}{N} \varphi(s)ds + (N-1)\int_{S_{N-1}} \frac{s}{N-1} \varphi(s)ds +\cdots + \int_{S_1} s \varphi(s)ds \nonumber\\ 
        & = \int_{\cup S_i}s \varphi(s)ds. \label{equ:SWmax_corr}
    \end{align}
    Since $\Pr(S\in \cup S_i)\leq Nc$ and $s$ is strictly increasing, we must have $\SW_{\text{arb}}(\cI(c,N))\leq \int_{s_{Nc}}^1s\varphi(s)ds$. Therefore, $\SW_{\text{max}}(\cI(c,N))\leq \int_{s_{Nc}}^1s\varphi(s)ds$. We've shown that $\SW_{\text{max}}(\cI(c,N))\geq \int_{s_{Nc}}^1s\varphi(s)ds$ and $\SW_{\text{max}}(\cI(c,N))\leq \int_{s_{Nc}}^1s\varphi(s)ds$. As a result, $\SW_{\text{max}}(\cI(c,N))= \int_{s_{Nc}}^1s\varphi(s)ds$ when $Nc\leq 1$. 

    Next, consider the case when $Nc>1$. By equation (\ref{equ:SWmax_corr}), for any arbitrary strategy profile $\mathbf{f_\text{arb}}$ the social welfare is 
    \begin{align*}
        \SW_{\text{arb}}(\cI(c,N)) = \int_{\cup S_i}s \varphi(s)ds \leq \int_0^1s \varphi(s)ds = \E [S].
    \end{align*}
    Hence, $\SW_{\text{max}}(\cI(c,N))\leq \E [S]$. We will then show $\SW_{\text{max}}(\cI(c,N))\geq \E [S]$. Since $Nc>1$, there exists $0 = s_1\leq s_2\cdots \leq s_N \leq s_{N+1}=1$ with $p\leq N$ such that $\Pr(S\in[s_i, s_{i+1}]) \leq c$. Let $\bbf$ be a strategy such that $f_i(s) = 1$ if and only if $s\in [s_i, s_{i+1}]$ for $i\in [N]$. In this case, $\cup S_i = [0,1]$ and $\Pr(S\in \cup S_i) = 1$. Hence, by equation (\ref{equ:SWmax_corr}), the social welfare under $\bbf$ is $\int_0^1s \varphi(s)ds = \E [S]$. Hence, $\SW_{\text{max}}(\cI(c,N))\geq \E [S]$. As a result, we must have $\SW_{\text{max}}(\cI(c,N))= \E [S]$ when $Nc>1$.

\end{myproof}

\begin{lemma} \label{lemma:sw_max_indep}
For any distribution $\cD$,
    let $\cI(c, N) = (N, c, \cD, \indep)$ be the instance parameterized by the capacity $c$ and the number of firms $N$. Suppose $\Pr(S\in [s_{Nc},1])=Nc$, then 
    \begin{align*}
        \lim_{c\rightarrow 0} \SW_{\max}(\cI(c,N)) = \int_{s_{Nc}}^{1} s \varphi(s)ds.
    \end{align*}
\end{lemma}
\begin{myproof}
    Let $\mathbf{f_\text{max}}$ be the strategy profile for instance $\cI(c, N)$ that maximizes the social welfare. We will first show that $\Mmax(\mathbf{f_\text{max}}) =1$. Define $S_n = \{s\in[0,1], M(s;\mathbf{f_\text{max}}) = n\}$. We want to show that $S_i = \emptyset$ for all $i>1$. Suppose, by contradiction, $S_i \neq \emptyset$ for some $i>1$. We want to show that there exist some set $S'_1\notin \cup_i S_i$ such that $i\int_{S_i}\varphi(s)ds = \int_{S'_1}\varphi(s)ds$ and $i\int_{S_i}U_i(s)\varphi(s)ds < \int_{S'_1}s\varphi(s)ds$.  Under the utility constraint, we have $\sum_{i=1}^{\Mmax(\mathbf{f_\text{max}})} i\Pr(S\in S_i)\leq Nc$. Hence, $\Pr(S\in \cup_i S_i)\leq \sum_{i=1}^{\Mmax(\mathbf{f_\text{max}})} \Pr(S\in S_i)\leq \sum_{i=1}^{\Mmax(\mathbf{f_\text{max}})} i\Pr(S\in S_i)\leq Nc$. Since $c\rightarrow 0$, $Nc\rightarrow 0$ as well. When $Nc$ is small enough, we can always find some $S'_1 \subset (0.5, 1]\setminus \cup_i S_i$ such that $i\int_{S_i}\varphi(s)ds = \int_{S'_1}\varphi(s)ds$. Therefore, the total social welfare derived by interviewing applicants in $S_i$ is
    \begin{align*}
        i\int_{S_i}U_i(s)\varphi(s)ds \leq  i\int_{S_i}\frac{1}{i}\varphi(s)ds
        = \frac{1}{i} \int_{S'_1}\varphi(s)ds
        \leq \frac{1}{2} \int_{S'_1}\varphi(s)ds
        < \int_{S'_1}s\varphi(s)ds,
    \end{align*}
    which is strictly less than the social welfare that can be derived by interviewing applicants in $S'_1$. Therefore, for each $S_i \neq \emptyset$, we can always move the support from $S_i$ to some set $S'_1\notin \cup_i S_i$ to yield a higher utility. This contradicts the assumption that $\mathbf{f_\text{max}}$ maximizes the social welfare. Hence, $S_i=\emptyset$ for all $i>1$ and $\Mmax(\mathbf{f_\text{max}}) =1$. Since $\Pr(S\in S_1)\leq Nc$, $U_1(s) = s$ is strictly increasing and $\varphi(s)>0$, we have $\int_{S_1}s\varphi(s)ds \leq \int_{s_{Nc}}^1s\varphi(s)ds$ where $\Pr(S\in[s_{Nc},1]) = Nc$. Therefore, $\SW_{\max}(\cI(c,N)) = \int_{s_{Nc}}^{1} s \varphi(s)ds$ when $c$ is small enough. Hence, we conclude that $\lim_{c\rightarrow 0} \SW_{\max}(\cI(c,N)) = \int_{s_{Nc}}^{1} s \varphi(s)ds$.
    
\end{myproof}

\begin{myproof}[Proof of Theorem \ref{thm:PoA}.]

We first show the first part of the statement is true. That is, for any $N$, we have $\lim_{c\rightarrow 0} \PoA(\cI(c, N)) =1$. Let $\cI(c,N)$ be an instance with $N$ fixed. Let $\mathbf{f_c}$ be a Nash equilibrium of the instance $\cI(c,N)$ with thresholds $\tau_1, \cdots, \tau_{\Mmax(\mathbf{f_c})}$ defined as in Theorem \ref{thm:equilibrium}. As $c\rightarrow 0$, we have $Nc\rightarrow 0$. Hence, we must have $Nc<\int_{0.5}^1 \varphi(s)ds$ for small enough $c$. By lemma \ref{lemma:4.2.1}, the social welfare under $\mathbf{f_c}$ is $\lim_{c\rightarrow 0}\SW_\text{NE}(\cI(c,N))  =\int_{\tau_1}^{1}s\varphi(s)ds$.
and $\Pr(S\in[\tau_1, 1]) = Nc$. Next we will show $\lim_{c\rightarrow 0} \SW_{\max}(\cI(c,N)) = \int_{\tau_1}^{1}s\varphi(s)ds$ for $\theta = \{\corr, \indep\}$. If $\theta = \corr$, since $Nc\rightarrow 0 $ and $\Pr(S\in[\tau_1, 1]) = Nc$, by lemma \ref{lemma:sw_max}, $\lim_{c\rightarrow 0}\SW_{\max}(\cI(c,N)) = \int_{\tau_1}^{1} s \varphi(s)ds$ as desired. If $\theta = \indep$, since $\Pr(S\in[\tau_1, 1]) = Nc$, by lemma \ref{lemma:sw_max_indep}, $\lim_{c\rightarrow 0}\SW_{\max}(\cI(c,N))  = \int_{\tau_1}^{1} s \varphi(s)ds$ as well. Hence, $\lim_{c\rightarrow 0} \SW_{\max}(\cI(c,N)) = \int_{\tau_1}^{1}s\varphi(s)ds$ for $\theta = \{\corr, \indep\}$. Therefore,
\begin{align*}
    \lim_{c\rightarrow 0} \PoA(\cI(c)) = \lim_{c\rightarrow 0}\frac{\SW_{\max}(\cI(c))}{\SW_\text{NE}(\cI(c))} = \frac{\int_{\tau_1}^{1}s\varphi(s)ds }{\int_{\tau_1}^{1} s \varphi(s)ds} = 1.
\end{align*}

Next, we will show for any $c \in (0, 1)$, we have $\lim_{N\rightarrow \infty} \PoA(\cI(c, N)) =1$. We will first prove the case when $\theta = \corr$. By equation (\ref{eq:SW_corr}), when hiring decisions are correlated, $\SW_\text{NE}(\cI(c,N))  = \int_{\tau_1}^{1} s \varphi(s)ds$. Moreover, by lemma \ref{lemma:4.4.1}, $\lim_{N\rightarrow \infty}\tau_1 = 0$. Therefore, we must have 
    \begin{align*}
        \lim_{N\rightarrow \infty}\SW_\text{NE}(\cI(c)) =\lim_{\tau_1\rightarrow 0} \int_{\tau_1}^{1} s \varphi(s)ds = \int_{0}^{1} s \varphi(s)ds = \E [S].
    \end{align*}
When $N$ is large enough, we must have $Nc>1$. By lemma \ref{lemma:sw_max}, $\SW_{\max}(\cI(c)) = \E [S]$. Therefore, 

    \begin{align*}
    \lim_{N\rightarrow \infty} \PoA(\cI(c,N)) = \lim_{c\rightarrow 0}\frac{\SW_{\max}(\cI(c))}{\SW_\text{NE}(\cI(c))} = \frac{\E [S] }{\E [S]} = 1.
    \end{align*}
If $\theta = \indep$, by lemma \ref{lemma:3.2.2}, $\lim_{N\rightarrow \infty}\tau_{\Mmax(\mathbf{f_c})-1}\rightarrow 0$ and $U_N(s)\rightarrow \frac{1}{N}$. Therefore  
    \begin{align*}
        &\lim_{N\rightarrow \infty}\SW_\text{NE}(\cI(c,N))
        \\&=(\Mmax(\mathbf{f_c})-1)\int_{0}^{\tau_{\Mmax(\mathbf{f_c})}}\frac{1}{\Mmax(\mathbf{f_c})-1} \varphi(s)ds + \Mmax(\mathbf{f_c})\int_{\tau_{\Mmax(\mathbf{f_c})}}^{1}\frac{1}{\Mmax(\mathbf{f_c})} \varphi(s)ds\\
        & = \int_0^1 \varphi(s)ds = 1.
    \end{align*}
We will then show $\lim_{N\rightarrow \infty}\SW_{\max}(\cI(c, N))=1$ as well. We will first show $\SW_{\max}(\cI(c, N))\leq1$. Let $\mathbf{f_\text{arb}}$ be an arbitrary strategy profile and define $S_n = \{s\in[0,1], M(s;\mathbf{f_\text{arb}}) = n\}$ for $n\in [\Mmax(\mathbf{f_\text{arb}})]$. The social welfare under $\mathbf{f_\text{arb}}$ is 
    \begin{align}
        \SW_{\text{arb}}(\cI(c,N)) & = N\int_{S_N} U_N(s) \varphi(s)ds + (N-1)\int_{S_{N-1}} U_{N-1}(s) \varphi(s)ds +\cdots + \int_{S_1} s \varphi(s)ds \nonumber\\
        &\leq N\int_{S_N} \frac{1}{N} \varphi(s)ds + (N-1)\int_{S_{N-1}} \frac{1}{N-1} \varphi(s)ds +\cdots + \int_{S_1} \varphi(s)ds \nonumber\\
        & = \int_{\cup_{i=1}^N S_i} \varphi(s)ds \nonumber\\
        &= \Pr(S\in \cup_{i=1}^N S_i). \label{eq:SWmax_indep_ub}
    \end{align}
    By equation (\ref{eq:SWmax_indep_ub}), we have $\SW_{\text{arb}}(\cI(c,N))\leq \Pr(S\in \cup_{i=1}^N S_i)\leq 1$. Hence, $\SW_{\text{max}}(\cI(c,N))\leq 1$. As a result, $\lim_{N\rightarrow \infty}\SW_{\text{max}}(\cI(c,N))\leq 1$. Moreover, by definition, $\lim_{N\rightarrow \infty}\SW_{\text{max}}(\cI(c,N))\geq \lim_{N\rightarrow \infty}\SW_{\text{NE}}(\cI(c,N))=1$. Therefore, $\lim_{N\rightarrow \infty}\SW_{\text{max}}(\cI(c,N))\geq 1$. As a result, we conclude $\lim_{N\rightarrow \infty}\SW_{\text{max}}(\cI(c,N))=1$. Therefore,
    \begin{align*}
    \lim_{N\rightarrow \infty} \PoA(\cI(c,N)) = \lim_{N\rightarrow \infty}\frac{\SW_{\max}(\cI(c,N))}{\SW_\text{NE}(\cI(c,N))} = 1
    \end{align*} 
    for $\theta = \{\corr, \indep\}$.

\end{myproof}

\subsection{Proof of Theorem \ref{thm:congestion}}
We will first prove that the best response dynamics terminates at a Nash equilibrium for the discretized version of the game.

\begin{lemma}[Convergence of Serial Best-Response in Discretized Game]
\label{thm:discrete-convergence}
Under Assumption~\ref{ass:quantization}, the game is a finite exact potential game. The best-response dynamics terminates in finitely many steps at a Nash equilibrium.
\end{lemma}

\begin{myproof}
We first observe that the strategy set of each firm is finite. Each firm selects at most $L$ bins out of $K$, so the number of strategies per firm is finite:
\[
\sum_{\ell=0}^{L} \binom{K}{\ell} \le (K+1)^L.
\]
With $n$ firms, the total number of profiles $S = (S_1,\dots,S_n)$ is finite.

    We will then verify that the game is an exact potential game. Suppose there are $N$ firms in the market and the strategy profile is $\bbf$. Each firm selects a subset $S_i \subseteq \{1,\dots,K\}$ of at most $L$ bins to interview, satisfying the capacity constraint. For any strategy profile $S = (S_1,\dots,S_n)$, let $M_k(S)$ denote the number of firms selecting bin $k$.

Then the utility of firm $i$ is
\[
u_i(S) = \sum_{k \in S_i} \int_{I_k} U_{M_k(S)}(s)\, \varphi(s)\, ds,
\]
and define the potential function
\[
\Phi_K(S) = \sum_{k=1}^{K} \sum_{j=1}^{M_k(S)} \int_{I_k} U_j(s)\, \varphi(s)\, ds.
\]
Suppose firm $i$ unilaterally deviates from $S_i$ to $S_i'$ (all other firms fixed). Let $S' = (S_i', S_{-i})$ be the new profile. The only bins whose occupancy changes are those in $D := S_i \cup S_i'$. Therefore, we have
\begin{align*}
    \Delta \Phi_K &= \sum_{k \in D} \left( \sum_{j=1}^{M_k(S')} \int_{I_k} U_j(s)\, \varphi(s)\, ds - \sum_{j=1}^{M_k(S)} \int_{I_k} U_j(s)\, \varphi(s)\, ds \right)\\
    & = \sum_{k \in S_i\setminus S'_i} \left( \sum_{j=1}^{M_k(S')} \int_{I_k} U_j(s)\, \varphi(s)\, ds - \sum_{j=1}^{M_k(S)} \int_{I_k} U_j(s)\, \varphi(s)\, ds\right) \\
    &+ \sum_{k \in S'_i\setminus S_i} \left( \sum_{j=1}^{M_k(S')} \int_{I_k} U_j(s)\, \varphi(s)\, ds - \sum_{j=1}^{M_k(S)} \int_{I_k} U_j(s)\, \varphi(s)\, ds\right)\\
    & = \sum_{k \in S_i\setminus S'_i} \left( \sum_{j=1}^{M_k(S)-1} \int_{I_k} U_j(s)\, \varphi(s)\, ds - \sum_{j=1}^{M_k(S)} \int_{I_k} U_j(s)\, \varphi(s)\, ds\right) \\
    &+ \sum_{k \in S'_i\setminus S_i} \left( \sum_{j=1}^{M_k(S)+1} \int_{I_k} U_j(s)\, \varphi(s)\, ds - \sum_{j=1}^{M_k(S)} \int_{I_k} U_j(s)\, \varphi(s)\, ds\right)\\
    & = \sum_{k \in S'_i\setminus S_i}\int_{I_k} U_{M_k(S')}(s)\, \varphi(s)\, ds - \sum_{k \in S_i\setminus S'_i}\int_{I_k} U_{M_k(S)}(s)\, \varphi(s)\, ds\\
    & = u_i(S') -u_i(S).
\end{align*}

We've verified that $\Phi_K$ is an exact potential., and the game is an exact potential game. Each time a firm changes strategy, it strictly increases its utility, and thus strictly increases $\Phi_K$. 

We also observe that the potential function can only take finitely many values: for each such $k$ and $j \le M_k(S)$, the value
\[
\int_{I_k} U_j(s)\, \varphi(s)\, ds
\]
is a fixed constant (depending only on the functions $U_j$ and $\varphi$, and the partition). Therefore, for fixed $K$, the set of possible values of $\Phi_K(S)$ is finite (at most $(n+1)^K$ values).

Since the strategy space is finite and $\Phi_K$ takes only finitely many values, the process must terminate in finitely many steps. When the process terminates, no firm has a strict profitable deviation. Therefore, the final profile is a Nash equilibrium.
\end{myproof}
\begin{corollary}[$\varepsilon$-Approximation to Continuous Game]
 For any $\varepsilon > 0$, there exists $K_0 \in \mathbb{N}$ such that for all $K \ge K_0$, the following holds:

Let $S = (S_1,\dots,S_n)$ be any pure-strategy Nash equilibrium in the $K$-quantized game. Define the corresponding continuous strategy profile $f = (f_1,\dots,f_n)$ by setting $f_i = \bigcup_{k \in S_i} I_k$.

Then $f$ is an $\varepsilon$-Nash equilibrium in the original continuous game: for all firms $i$ and all admissible deviations $f_i' \subseteq [0,1]$ with $\int_{f_i'} \varphi(s) ds \le c$,
\[
u_i(f_i, f_{-i}) \ge u_i(f_i', f_{-i}) - \varepsilon.
\]
\end{corollary}

\begin{myproof}
Since $\varphi$ and $U_j$ are continuous on $[0,1]$ and the interval is compact, they are uniformly continuous and bounded. The bins $I_k$ shrink in width as $K \to \infty$, and the Riemann sums
\[
\sum_{k \in S_i} \int_{I_k} U_{r_k(i)}(s) \varphi(s)\,ds
\]
converge uniformly to the integral
\[
\int_{f_i} U_{M_k(s)}(s) \varphi(s)\,ds = u_i(f).
\]
Similarly, for any continuous deviation $f_i'$, we can approximate it by a union of quantized bins $f_i^{(K)}$ such that the utility difference is less than $\varepsilon/2$.

Because $S$ is a pure Nash equilibrium of the quantized game, firm $i$ has no profitable deviation among the quantized bin unions. Therefore, the gain from any deviation in the continuous game is at most $\varepsilon$ for $K$ sufficiently large.
\end{myproof}

\subsection{Proof of Theorem \ref{thm:convergence_corr}}
    We first prove some intermediate results before proving the main theorem. 
    \begin{lemma}\label{lemma:conv1}
        Suppose there is a $\delta > 0$ such that $\varphi(s) \geq \delta \; \forall s \in [0,1]$ and there are $N$ firms in the market that form an equal-utility Nash equilibrium $\mathbf{f_N}$ with $\theta = \corr$. Let $0=\tau_0\leq \cdots \leq\tau_{N+1}=1$ be the set of thresholds corresponds to $\mathbf{f_N}$. For a new firm $N+1$, as long as $c \leq 0.5 \delta$ for every firm $i$, there exists $0=\tau'_0\leq \cdots\leq \tau'_{N+2}=1$ such that $\tau_m < \tau'_{m+1}\leq \tau_{m+1}$ for each $i\in [\Mmax(\mathbf{f_N})]$, $\sum_{m=0}^N \int_{\tau'_{m+1}}^{\tau_{m+1}} \varphi(s)ds = c$ and $U_n(\tau'_n) = U_m(\tau'_m)$ for all $0<\tau'_m<\tau'_n<1$. Let $f_{N+1}(s) = 1$ if and only if $s\in\cup_{m=0}^{N} [\tau'_{m+1}, \tau_{m+1})$. Then $f_{N+1}$ is the unique best response for firm $N+1$ and the strategy profile forms an equal-utility nash equilibrium.  
        
    \end{lemma}

    \begin{myproof}

        Let $f(t) = \sum_{m=0}^N \int_{(m+1)t}^{\tau_{m+1}} \varphi(s)ds$ for $t\in[0, \tau_1]$. We will show that for all $t\in[\frac{\tau_{\Mmax(\mathbf{f_N})}}{\Mmax(\mathbf{f_N})+1}, \tau_1)$, thresholds defined by 
        \begin{align}
        \tau'_{m} = 
        \begin{cases}\label{eq:def_thresholds}
            mt,&\text{if } mt\leq  1\\
            1, &o.w. 
        \end{cases}
        \end{align}
        satisfy $\tau_m < \tau'_{m+1}\leq \tau_{m+1}$ for each $m\in [\Mmax(\mathbf{f_N})]$. Let $m\in [\Mmax(\mathbf{f_N})]$ be arbitrary, we will first check the inequality when $t^* = \frac{\tau_{\Mmax(\mathbf{f_N})}}{\Mmax(\mathbf{f_N})+1} = \frac{\Mmax(\mathbf{f_N})\tau_1}{\Mmax(\mathbf{f_N})+1}$. If $\tau_{\Mmax(\mathbf{f_N})+1}<1$, then for any $m\in [\Mmax(\mathbf{f_N})]$
        \begin{align}\label{eq:check_inequality}
            \frac{\tau'_{m+1}}{\tau_m} & > \frac{(m+1)\frac{\Mmax(\mathbf{f_N})\tau_1}{\Mmax(\mathbf{f_N})+1}}{m\tau_1} \nonumber\\
            & = \frac{m+1}{\mathbf{\Mmax(\mathbf{f_N})}+1}\frac{\Mmax(\mathbf{f_N})}{m}\nonumber\\
            & =(1+\frac{1}{m})(1-\frac{1}{\Mmax(\mathbf{f_N})+1}) \nonumber\\
            & \geq (1+\frac{1}{\Mmax(\mathbf{f_N})})(1-\frac{1}{\Mmax(\mathbf{f_N})+1}) \nonumber\\
            & = \frac{\Mmax(\mathbf{f_N})+1}{\Mmax(\mathbf{f_N})}\frac{\Mmax(\mathbf{f_N})}{\Mmax(\mathbf{f_N})+1} = 1. 
        \end{align}
        If $\tau_{\Mmax(\mathbf{f_N})+1}=1$, then clearly, $\tau_{\Mmax(\mathbf{f_N})}<\tau'_{\Mmax(\mathbf{f_N})+1}\leq\tau_{\Mmax(\mathbf{f_N})+1}$. The equation (\ref{eq:check_inequality}) holds for $m<[\Mmax(\mathbf{f_N})]$. Therefore, the inequality holds for all $m$. We conclude that $\tau^*_{m+1}> \tau_m$ if $t^* = \frac{\tau_{\Mmax(\mathbf{f_N})}}{\Mmax(\mathbf{f_N})+1}$ for all $m\in [\Mmax(\mathbf{f_N})]$. For any $t'\in[t^*, \tau_1)$, if $\tau'_{m+1}<1$, then $\tau'_{m+1}=(m+1)t'>(m+1)t^* =\tau^*_{m+1}\geq \tau_m$. If $\tau'_{m+1}=1$, then clearly $\tau'_{m+1}> \tau_{m}$ since by definition $\tau_{m}<1$ for all $m\in [\Mmax(\mathbf{f_N})]$. Moreover, $mt\leq m\tau_1 = \tau_m$ for any $m\in [\Mmax(\mathbf{f_N})+1]$. Therefore, $\tau_m < \tau'_{m+1}\leq \tau_{m+1}$ for each $m\in [\Mmax(\mathbf{f_N})]$. 

        Next, we want to show $f(t)$ is continuous. Note that each integral  $\int_{(m+1)t}^{\tau_{m+1}} \varphi(s)ds$ is continuous in $t$, and as a result, we must have $f(t)$ is continuous as well. Moreover, by proposition \ref{prop:NE_shared}, if $\theta = \corr$, we must have $\tau_m = m\tau_1$ for each $m\in [\Mmax(\mathbf{f_N})]$. Hence, $\lim_{t\rightarrow \tau_1}f(t) = \sum_{m=0}^N \int_{\tau_{m+1}}^{\tau_{m+1}} \varphi(s)ds=0$. Furthermore, 
        \begin{align*}
            \lim_{t\rightarrow t^*}f(t) &= \sum_{m=0}^{\Mmax(\mathbf{f_N})} \int_{\tau'_{m+1}}^{\tau_{m+1}} \varphi(s)ds\\
            &\geq \sum_{m=0}^{\Mmax(\mathbf{f_N})} \int_{\tau'_{m+1}}^{\tau_{m+1}}\delta ds\\
            & = \delta \sum_{m=0}^{\Mmax(\mathbf{f_N})} (\tau_{m+1} - \tau'_{m+1})\\
            & = \delta (\sum_{m=0}^{\Mmax(\mathbf{f_N})-1} (\frac{(m+1)\tau_{\Mmax(\mathbf{f_N})}}{\Mmax(\mathbf{f_N})} - \frac{(m+1)\tau_{\Mmax(\mathbf{f_N})}}{\Mmax(\mathbf{f_N})+1}) + (1-\tau'_{\Mmax(\mathbf{f_N})+1}))\\
            & = \delta (\sum_{m=1}^{\Mmax(\mathbf{f_N})} \frac{m\tau_{\Mmax(\mathbf{f_N})}}{\Mmax(\mathbf{f_N})(\Mmax(\mathbf{f_N})+1)}+ (1-\tau'_{\Mmax(\mathbf{f_N})+1}))\\
            & = \delta (\frac{\Mmax(\mathbf{f_N})(\Mmax(\mathbf{f_N})+1)}{2}\frac{\tau_{\Mmax(\mathbf{f_N})}}{\Mmax(\mathbf{f_N})(\Mmax(\mathbf{f_N})+1)}+ (1-\tau'_{\Mmax(\mathbf{f_N})+1}))\\
            & = \delta(1 + \frac{\tau_{\Mmax(\mathbf{f_N})}}{2} - \tau'_{\Mmax(\mathbf{f_N})+1})\\
            & \geq \delta(1 + \frac{\tau_{\Mmax(\mathbf{f_N})}}{2} - \tau_{\Mmax(\mathbf{f_N})})\\
            & =\delta( 1-\frac{\tau_{\Mmax(\mathbf{f_N})}}{2}) > 0.5\delta.
        \end{align*}
        Since $0<c\leq 0.5\delta$, by intermediate value theorem, there exists $t_0\in[t^*, \tau_1)$ such that $f(t_0) = c$. In addition, the thresholds defined by (\ref{eq:def_thresholds}) satisfy the desired inequalities.

        Let $\mathbf{f_{N+1}} = (f_1, \cdots, f_N, f'_{N+1})$ be the strategy profile after firm $N+1$ makes the above response. We want to show that $f_{N+1}$ is the best response strategy by showing $\mathbf{f_{N+1}}$ is an equal-utility Nash equilibrium. Since $\tau_m < \tau'_{m+1}\leq \tau_{m+1}$ for each $m\in [\Mmax(\mathbf{f_N})]$, $\cup_{m=1}^{N+1} [\tau'_{m+1}, \tau_{m+1})$ is the union of disjoint sets. Hence, $\Pr(s\in[0,1], f_{N+1}=1)=\cup_{m=1}^{N+1} \Pr(S\in[\tau'_{m+1}, \tau_{m+1})) =c$. As a result, $\mathbf{f_{N+1}}$ satisfies the first condition of Theorem \ref{thm:equilibrium}. In addition, since $\tau_m < \tau'_{m+1}\leq \tau_{m+1}$ and the strategies $(f_1, \cdots, f_N)$ remain unchanged, we have $M(s, \mathbf{f_{N+1}}) = M(s, \mathbf{f_{N}})+1$ for $s\in \cup_{m=1}^{N+1} [\tau'_{m+1}, \tau_{m+1})$. Hence, $M(s, \mathbf{f_{N+1}}) = m$ for $s\in [\tau'_m, \tau'_{m+1})$ and the thresholds satisfy the second condition of Theorem \ref{thm:equilibrium}. Next, the thresholds defined by (\ref{eq:def_thresholds}) satisfy $U_m(\tau'_m) = U_n(\tau'_n)$ for all $0<\tau'_m<\tau'_n<1$. Therefore, $\mathbf{f_{N+1}}$ satisfies all the conditions of Theorem \ref{thm:equilibrium} and is an equal-utility Nash equilibrium. By definition of a Nash equilibrium, $f'_{N+1}$ is the best response strategy. We only need to show this strategy is unique. Let $K = \{s\in[0,1], f_{N+1}(s) = 1\}$. We claim that the utility firm $N+1$ can derive by interviewing an applicant with score $s\notin K$ is strictly less than the worst utility it can derive by interviewing an applicant with in $K$. Given that the thresholds satisfy $\tau_m < \tau'_{m+1}\leq \tau_{m+1}$ for each $m\in [\Mmax(\mathbf{f_N})]$ and $U_m(\tau'_m) = U_n(\tau'_n)$ for all $0<\tau'_m<\tau'_n<1$, $\min_{s\in K}U_{M(s, \mathbf{f_{N+1}})}(s) = t_0$. Since $U_n(s)$ is strictly increasing in $s$, for any applicant with $s\notin K$, the utility firm $N+1$ can derive by interviewing this applicant is $U_{M(s, \mathbf{f_{N+1}}+1)}(s)<t_0$. Hence, $U_{M(s, \mathbf{f_{N+1}}+1)}(s)<\min_{s\in K}U_{M(s, \mathbf{f_{N+1}})}(s)$. This implies that moving any support from $K$ to $K^c$ will lead to a strictly smaller utility. Therefore, $f_{N+1}$ yields a strictly higher utility than any other strategy. Hence, $f_{N+1}$ is the unique best response.
    \end{myproof}

    \begin{myproof} [\text{Proof of Theorem \ref{thm:convergence_corr}}]
        We will prove the theorem by induction. When $n=1$, the best response strategy is $f_1(s) = 1$ if and only if $s\in [s_c, 1]$ with $\Pr(S\in [s_c, 1]) = c$. Since $c<0.5\delta$, we have $\int_{0.5}^1\varphi(s)ds\geq \int_{0.5}^1\delta ds = 0.5\delta$. Hence, $c<\int_{0.5}^1\varphi(s)ds$. Let $\bbf_1$ denote the strategy profile. $\bbf_1$ satisfies the conditions established in Theorem \ref{thm:equilibrium}, and thus it is at an equal-utility Nash equilibrium. 

        Next, suppose the strategy profile $\bbf_k$ with $\Mmax(\bbf_k)=k$ is at an equal-utility Nash equilibrium and there are $n$ firms. By lemma \ref{lemma:conv1}, since the best response strategy is unique, the new strategy profile $\bbf'_k$ after firm $n+1$ makes the best response is still at an equal-utility Nash equilibrium. Therefore, the one-turn best response dynamics converges to an equal-utility Nash equilibrium. 
    \end{myproof}

\subsection{Proof of Theorem \ref{thm:convergence_indep}}
To show the statement is true, we first present some intermediate results. 

\begin{lemma}\label{lemma:5.3.1}
    Suppose $\tau_1\in [0,1]$ and $\tau_m = 1-(1-m\tau_1)^{\frac{1}{m}}$ such that $\tau_m\leq 1$, then $\tau_m$ is strictly decreasing when $\tau_1$ decreases.  Fix an integer $n\geq 1$, $\tau_{n+1}-\tau_n$ is strictly decreasing when $\tau_1$ decreases. 
\end{lemma}
\begin{myproof}
    Let $f(t)= 1-(1-mt)^{\frac{1}{m}}$ and $m$ be fixed. We will show that $f(t)$ decreases when $t$ decreases. When $t$ decreases,  $(1-mt)^{\frac{1}{m}}$ is strictly increasing and therefore $f(t)$ is strictly decreasing. Then the conclusion follows. Next, we fix an integer $n\geq 1$, and we want to show $\tau_{n+1}-\tau_n$ is strictly decreasing with $\tau_1$. Let $g(t) = \tau_{n+1}-\tau_n = (1-nt)^{\frac{1}{n}}-(1-(n+1)t)^{\frac{1}{n+1}}$. The derivative is $g'(t) = \frac{1}{(1-(n+1)t)^{1-\frac{1}{n+1}}} - \frac{1}{(1-nt)^{1-\frac{1}{n}}}$. Since $0<1-\frac{1}{n}<1-\frac{1}{n+1}<1$ and $1-nt>1-(n+1)t$, we must have $(1-(n+1)t)^{1-\frac{1}{n+1}}<(1-nt)^{1-\frac{1}{n}}$. Therefore, $g'(t) = \frac{1}{(1-(n+1)t)^{1-\frac{1}{n+1}}} - \frac{1}{(1-nt)^{1-\frac{1}{n}}} >0$. We conclude that $\tau_{n+1}-\tau_n$ is strictly decreasing with $\tau_1$.
\end{myproof}

\begin{lemma}\label{lemma:conv2}
        Fix the total number of firms $N_0$. Suppose there is a $\delta > 0$ such that $\varphi(s) \geq \delta \; \forall s \in [0,1]$ and there are $N< N_0$ firms in the market that form an equal-utility Nash equilibrium $\mathbf{f_N}$ with $\theta = \indep$, $\tau_1>\frac{1}{\Mmax(\mathbf{f_N})+1}$. Let $0=\tau_0\leq \cdots \leq\tau_{N+1}=1$ be the set of thresholds corresponds to $\mathbf{f_N}$. Let $c_0 = \delta (\min_{m\in[N_0]} ((1-\frac{m}{N_0+1})^{\frac{1}{m}} -(1-\frac{m+1}{N_0+1})^{\frac{1}{m+1}}))$. For a new firm $N+1$, as long as $c \leq c_0$ for every firm $i$, there exists $0=\tau'_0\leq \cdots\leq \tau'_{N+2}=1$ such that $\tau_m < \tau'_{m+1}\leq \tau_{m+1}$ for each $m\in [\Mmax(\mathbf{f_N})]$, $\sum_{m=0}^N \int_{\tau'_{m+1}}^{\tau_{m+1}} \varphi(s)ds = c$ and $U_m(\tau'_m) = U_n(\tau'_n)$ for all $0<\tau'_m<\tau'_n<1$. Let $f_{N+1}(s) = 1$ if and only if $s\in\cup_{m=0}^{N} [\tau'_{m+1}, \tau_{m+1})$. $f_{N+1}$ is the unique best response for firm $N+1$, and the strategy profile $\mathbf{f'_N}$ defined by $\mathbf{f'_N} = (f_1, \cdots, f_N, f'_{N+1})$ forms an equal-utility nash equilibrium with $\tau'_1>\frac{1}{\Mmax(\mathbf{f'_N})+1}$. 
    \end{lemma}

    \begin{myproof}
        By equation $(\ref{equ:Un(s)_indep})$, given $\tau_1$, under $\theta =\indep$, we have $\tau_m = 1-(1-\tau_1 m)^{1/m}$ for $m\leq \Mmax(\mathbf{f_N})$. Let $t^* = \max_{m\in [\Mmax(\mathbf{f_N})]}U_{m+1}(\tau_m)$. We will show that for all $t\in(t^*, \tau_1)$, thresholds defined by
        \begin{align}
        \tau'_{m} = 
        \begin{cases} \label{eq:thresholds_indep}
            1-(1-m t)^{\frac{1}{m}},&\text{if } 1-(1-m t)^{\frac{1}{m}}\leq  1\\
            1, &o.w. 
        \end{cases}
        \end{align}
        satisfy $\tau_m < \tau'_{m+1}\leq \tau_{m+1}$ for each $m\in [\Mmax(\mathbf{f_N})]$. First, we note that $U_{m+1}(\tau'_{m+1})=t$ under construction. Suppose, by contradiction, $\tau_m \geq \tau'_{m+1}$ for some $m$, then we have $U_{m+1}(\tau_m)\geq U_{m+1}(\tau'_{m+1})=t$. This contradicts the assumption that $t^*=\max_{m\in [\Mmax(\mathbf{f_N})]}U_{m+1}(\tau_m)<t$. Moreover, $\tau'_{m+1}\leq \tau_{m+1}$ since $t<\tau_1$ by lemma \ref{lemma:5.3.1}. Therefore, we conclude that $\tau_m < \tau'_{m+1}\leq \tau_{m+1}$ for each $m\in [\Mmax(\mathbf{f_N})]$.
        
       Next, let $f(t) = \sum_{m=0}^N \int_{1-(1-t m)^{1/m}}^{\tau_{m+1}} \varphi(s)ds$ for $t\in[0, \tau_1]$. We want to show that there exists $t$ such that $f(t) = c$. Since $\int_{1-(1-t m)^{1/m}}^{\tau_{m+1}} \varphi(s)ds$ is continuous for each $m$, $f(t)$ is continuous. We also notice that $\lim_{t\rightarrow \tau_1}f(t) = \sum_{m=0}^N \int_{\tau_{m+1}}^{\tau_{m+1}} \varphi(s)ds=0$. Define $c_0 = \delta (\min_{m\in[N_0]} ((1-\frac{m}{N_0+1})^{\frac{1}{m}} -(1-\frac{m+1}{N_0+1})^{\frac{1}{m+1}}))$. By lemma \ref{lemma:5.3.1}, for any $\tau_1>\frac{1}{N+1}$, we must have $\tau_{m+1}-\tau_{m}\geq (1-\frac{m}{N+1})^{\frac{1}{m}} -(1-\frac{m+1}{N+1})^{\frac{1}{m+1}}\geq \min_{m\in[N_0]} ((1-\frac{m}{N_0+1})^{\frac{1}{m}} -(1-\frac{m+1}{N_0+1})^{\frac{1}{m+1}})$. Since $\mathbf{f_N}$ is an equal-utility nash equilibrium, by assumption we have $\tau_1>\frac{1}{N+1}$. Furthermore, we have
        \begin{align*}
            \lim_{t\rightarrow t^*}f(t) &= \sum_{m=0}^{\Mmax(\mathbf{f_N})} \int_{\tau'_{m+1}}^{\tau_{m+1}} \varphi(s)ds\\
            &\geq \sum_{m=0}^{\Mmax(\mathbf{f_N})} \int_{\tau'_{m+1}}^{\tau_{m+1}}\delta ds\\
            & = \delta \sum_{m=0}^{\Mmax(\mathbf{f_N})} (\tau_{m+1} - \tau'_{m+1})\\
            & > \delta \sum_{m=0}^{\Mmax(\mathbf{f_N})} (\tau_{m+1} - \tau_m)\\
            & > \delta (\max_{m\in [\Mmax(\mathbf{f_N})]} \tau_{m+1} - \tau_m)\\
            & > c_0.
        \end{align*}
        Since $c\leq c_0$, by intermediate value theorem, there exists $t_0\in[t^*, \tau_1)$ such that $f(t_0) = c$. In addition, the thresholds defined by $\tau_m = 1-(1-\tau_1 m)^{1/m}$ such that $\tau_m\leq 1$ satisfy the desired inequalities.

        Let $\mathbf{f_{N+1}} = (f_1, \cdots, f_N, f'_{N+1})$ be the strategy profile after firm $N+1$ makes the above response. We want to show that $f_{N+1}$ is the best response strategy by showing $\mathbf{f_{N+1}}$ is an equal-utility Nash equilibrium. Since $\tau_m < \tau'_{m+1}\leq \tau_{m+1}$ for each $m\in [\Mmax(\mathbf{f_N})]$, $\cup_{m=1}^{N+1} [\tau'_{m+1}, \tau_{m+1})$ is the union of disjoint sets. Hence, $\Pr(s\in[0,1], f_{N+1}=1)=\cup_{m=1}^{N+1} \Pr(S\in[\tau'_{m+1}, \tau_{m+1})) =c$. As a result, $\mathbf{f_{N+1}}$ satisfies the first condition of Theorem \ref{thm:equilibrium}. In addition, since $\tau_m < \tau'_{m+1}\leq \tau_{m+1}$ and the strategies $(f_1, \cdots, f_N)$ remain unchanged, we have $M(s, \mathbf{f_{N+1}}) = M(s, \mathbf{f_{N}})+1$ for $s\in \cup_{m=1}^{N+1} [\tau'_{m+1}, \tau_{m+1})$. Hence, $M(s, \mathbf{f_{N+1}}) = m$ for $s\in [\tau'_m, \tau'_{m+1})$ and the thresholds satisfy the second condition of Theorem \ref{thm:equilibrium}. Next, the thresholds defined by (\ref{eq:thresholds_indep}) satisfy $U_m(\tau'_m) = U_n(\tau'_n)$ for all $0<\tau'_m<\tau'_n<1$. Therefore, $\mathbf{f_{N+1}}$ satisfies all the conditions of Theorem \ref{thm:equilibrium} and is an equal-utility Nash equilibrium. Furthermore, $\tau'_1>\frac{1}{\Mmax(\mathbf{f'_N})+1}$ by construction. By definition of a Nash equilibrium, $f'_{N+1}$ is the best response strategy. We only need to show this strategy is unique. Let $K = \{s\in[0,1], f_{N+1}(s) = 1\}$. We claim that the utility firm $N+1$ can derive by interviewing an applicant with score $s\notin K$ is strictly less than the worst utility it can derive by interviewing an applicant with in $K$. Given that the thresholds satisfy $\tau_m < \tau'_{m+1}\leq \tau_{m+1}$ for each $m\in [\Mmax(\mathbf{f_N})]$ and $U_m(\tau'_m) = U_n(\tau'_n)$ for all $0<\tau'_m<\tau'_n<1$, $\min_{s\in K}U_{M(s, \mathbf{f_{N+1}})}(s) = t_0$. Since $U_n(s)$ is strictly increasing in $s$, for any applicant with $s\notin K$, the utility firm $N+1$ can derive by interviewing this applicant is $U_{M(s, \mathbf{f_{N+1}}+1)}(s)<t_0$. Hence, $U_{M(s, \mathbf{f_{N+1}}+1)}(s)<\min_{s\in K}U_{M(s, \mathbf{f_{N+1}})}(s)$. This implies that moving any support from $K$ to $K^c$ will lead to a strictly smaller utility. Therefore, $f_{N+1}$ yields a strictly higher utility than any other strategy. Hence, $f_{N+1}$ is the unique best response.
    \end{myproof}
    
    \begin{myproof} [\text{Proof of Theorem \ref{thm:convergence_indep}}]
        Fix the total number of firms $N$, let $c_0 = \min (0.5\delta,\delta (\min_{m\in[N]} ((1-\frac{m}{N+1})^{\frac{1}{m}} -(1-\frac{m+1}{N+1})^{\frac{1}{m+1}}))$. We will prove the theorem by induction. When $n=1$, the best response strategy is $f_1(s) = 1$ if and only if $s\in [s_c, 1]$ with $\Pr(S\in [s_c, 1]) = c$. Since $c<0.5\delta$, we have $\int_{0.5}^1\varphi(s)ds\geq \int_{0.5}^1\delta ds = 0.5\delta$. Hence, $c<\int_{0.5}^1\varphi(s)ds$ and $\tau_1=s_c>0.5$. Let $\bbf_1$ denote the strategy profile. $\bbf_1$ satisfies the conditions established in Theorem \ref{thm:equilibrium}, and thus it is at an equal-utility Nash equilibrium with $\tau_1>0.5$. 

        Next, suppose the strategy profile $\bbf_k$ with $\Mmax(\bbf_k)=k$ and n firms is at an equal-utility Nash equilibrium with $\tau_1>\frac{1}{\Mmax(\bbf_k)+1}$. By lemma \ref{lemma:conv2}, since the best response strategy is unique, the new strategy profile $\bbf'_k$ after firm $n+1$ makes the best response is still at an equal-utility Nash equilibrium. Therefore, the one-turn best response dynamics converges to an equal-utility Nash equilibrium. 
    \end{myproof}

\subsection{Proof of Proposition \ref{prop:fixedW_corr}}
\begin{myproof}
    We will start with the first part of the statement. We notice that the function $f(x) = N \int_x^1 U_N(s)\varphi(s)ds$ is continuous in $x$ and $f(0) =  N \int_0^1 U_N(s)\varphi(s)ds$, $f(1) = 0$. Moreover, since $s\varphi(s)>0$ for $s>0$, $f(x)$ is strictly decreasing in $x$. Therefore, given $W\in(0,N \int_0^1 U_N(s)\varphi(s)ds)$, there exist a unique $x^*\in(0,1)$ such that $f(x^*) = W$. Let $c = \Pr(S\in [x^*, 1])$. The social welfare under the naive solution is 

    \begin{align*}
        \SW_\text{naive} =N\int_{x^*}^1 U_N(s) \varphi(s) ds = W
    \end{align*}
    as desired. 

    Now suppose $W\leq \int_{0.5}^{1}s\varphi(s)ds$.  Let $\cI(N) = (N, c/N, \cD, \theta)$ be the instance parameterized by the number of firms $N$. The total capacity of the N firms is $c$, and there exists disjoint sets $K_1, \cdots, K_N$ such that $\cup K_i = [x^*, 1]$ and $\Pr(S\in K_i) = c/N$. Define the strategy profile $\mathbf{f_N} = (f_1, \cdots, f_N)$ such that $K_i = \{s\in[0,1], f_i(s) = 1\}$. We will show that $\mathbf{f_N}$ is a Nash equilibrium with social welfare $W$.  We will first show $x^*\geq 0.5$. Given that $\int_{x^*}^1 s \varphi(s) ds=W\leq \int_{0.5}^{1}s\varphi(s)ds$ and $f(x)$ is strictly decreasing in $x$. We must have $x^*\geq 0.5$. By construction, we have the thresholds $\tau_0 = 0,\tau_1 = x^*, \tau_2 =\cdots=\tau_{N+1}= 1$. Moreover, $M(s, \mathbf{f_N}) = 1$ for $s\in [\tau_1, 1]$. Hence, $\mathbf{f_N}$ satisfies all the conditions in Theorem \ref{thm:equilibrium}. As a result,  $\mathbf{f_N}$ is a Nash equilibrium. The social welfare under $\mathbf{f_N}$ is $\SW_\text{NE} =\int_{x^*}^1 s \varphi(s) ds = W$ as desired. 
\end{myproof}

\subsection{Derivation of $U_{n_B}(s, \bbf_A)$ when $\theta = \indep$}\label{appendix:deri_indep_tiered}
\begin{myproof}
   We use $M(s, \mathbf{f_A})$ to denote the number of tier A firms interviewing an applicant with score $s$ when the tier A firms are playing the action profile $\mathbf{f_A}$. Then the utility derived by each of the $n_B$ tier B firms can be expressed as 
   \begin{align*}
       &U_{n_b}(s, \bbf_A) \\
       &= \mathbb{P}(\text{applicant accepts the offer when there are $n_B$ tier B offers} | \text{no tier A offer})\mathbb{P}(\text{no tier A offer})\\
       & = U_{n_B}(s)(1-s)^{M(s, \mathbf{f_A})}\\
       & = (1-(1-s)^{n_B})(1-s)^{M(s, \mathbf{f_A})}/n_B.
   \end{align*}

\end{myproof}

\subsection{Proof of Proposition \ref{prop:tiered_eq_corr}}
\begin{myproof}
    Suppose there are $N_A$ tier A firms and $N_B$ tier B firms. We will first prove the forward direction. Suppose the strategy profile $\mathbf{f}_A\cup \mathbf{f}_B$ is at Nash equilibrium. We will first show that $\mathbf{f}_A$ must itself be a Nash equilibrium. 
    
   Assume, for contradiction, $\mathbf{f}_A$ is not a Nash equilibrium when tier B firms are absent. Then there is at least one tier A firm $i$ that can increase its utility by changing its strategy. However, the utility function of a tier A firm is unaffected by the presence or absence of tier B firms, since the payoff depends only on its own threshold and the applicant distribution. This implies that, even with the presence of tier B firms, firm $i$ should still be able to improve its utility by changing its strategy, so $\mathbf{f}_A\cup \mathbf{f}_B$ cannot be a Nash equilibrium. Hence, $\mathbf{f}_A$ must be at equilibrium, and by Theorem \ref{thm:equilibrium}, there are thresholds that satisfy the four conditions in the theorem. Hence, we've proved the first condition in the proposition. 
   
   Next, fix $\bbf_A$ and the corresponding thresholds identified in the previous condition $\tau_1, \cdots, \tau_{N_A}$. Under $\theta = \corr$, if a tier B firm interviews an applicant that is also interviewed by another tier A firm, then the applicant will receive an offer from both firms or neither of the firms. In this case, the applicant will never go to the tier B firm. As a result, the utility gained by the tier B firm is strictly zero. Since tier B firms get zero utility by interviewing an applicant that is also interviewed by any tier A firm when $\theta = \corr$, tier B firms will never interview any applicant with a score higher than $\tau_1$. Following the proof of Theorem \ref{thm:equilibrium}, we can show that the second statement of the proposition holds. 

    Conversely, suppose $f_A$ and $f_B$ each satisfy the four equilibrium conditions in Theorem~\ref{thm:equilibrium} 
with thresholds 
\[
0 = \tau_0 \le \tau'_1 \le \cdots \le \tau'_{N_B} \le \tau'_{N_B+1} 
  = \tau_1 \le \tau_2 \le \cdots \le \tau_{N_A+1} = 1.
\]
We show that the combined strategy profile $f_A \cup f_B$ forms a Nash equilibrium under $\theta = \text{CORR}$.

Because each firm's utility depends only on applicants within its assigned threshold interval, and these intervals are disjoint across tiers by construction, no firm---whether tier A or tier B---can gain from deviating. 
Each tier A firm maximizes its utility given the thresholds $\tau_i$ according to Theorem~\ref{thm:equilibrium}, 
while tier B firms obtain zero utility from interviewing applicants with scores above $\tau_1$. 
Any deviation by a tier B firm to higher-score applicants would therefore either violate its capacity constraint or strictly reduce its expected utility. 

Consequently, no unilateral deviation is profitable for any firm, and the combined strategy profile $f_A \cup f_B$ constitutes a Nash equilibrium. 
This completes the proof.
    
\end{myproof}

\subsection{Proof of Proposition \ref{prop:tiered_eq_indep}}
\begin{myproof}
Suppose there are $N_A$ tier A firms and $N_B$ tier B firms. We assume that $\mathbf{f_A} \cup \mathbf{f_B}$ satisfies the listed conditions. For tier A firms, since the strategies of tier B firms do not affect the utility of tier A firms, satisfying the conditions listed in Theorem \ref{thm:equilibrium} implies that tier A firms cannot increase their utility by deviating from their current strategies. For tier B firms, we will show that any tier B firm $i$ cannot earn a higher utility by deviating from its current strategy profile. Let $f^B_i$ be the current strategy profile of firm $i$ and $f'^B_i$ be an arbitrary alternative strategy profile. We will show that $u(f^B_i, f^B_{-i})\geq u(f'^B_i, f^B_{-i})$ if $M(s, \mathbf{f_B}) $ satisfies the second conditions. Denote $K_i = \{s\in[0,1], f^B_i(s)=1\}$. Under $\bbf_B$, we can express the utility of firm $i$ as 
\begin{align}
    u(f^B_i, f^B_{-i}) &= \int_{K\cap K'} U_{M(s, \mathbf{f_B})}(s, \mathbf{f_A}) \varphi(s)ds + \int_{K\setminus K'}U_{M(s, \mathbf{f_B})}(s, \mathbf{f_A})\varphi(s)ds.  \nonumber
\end{align}
Let $\bbf'_B$ denote the strategy profile after firm $i$ deviates from $f^B_i$, and denote $K'_i = \{s\in[0,1], f'^B_i(s)=1\}$. We notice that $M(s, \mathbf{f'_B})=M(s, \mathbf{f_B})+1$ if $s\in K'\setminus K$, and $M(s, \mathbf{f'})=M(s, \mathbf{f})$ if $s\in K\cap K'$. Therefore, the utility under the alternative strategy is 
\begin{align}
    u(f'^B_i, f^B_{-i}) & = \int_{K\cap K'}U_{M(s, \mathbf{f'_B})}(s, \mathbf{f_A}) \varphi(s)ds + \int_{K'\setminus K}U_{M(s, \mathbf{f'_B})}(s, \mathbf{f_A}) \varphi(s)ds \nonumber \\
    & = \int_{K\cap K'}U_{M(s, \mathbf{f_B})}(s, \mathbf{f_A}) \varphi(s)ds + \int_{K'\setminus K}U_{M(s, \mathbf{f_B})+1}(s, \mathbf{f_A}) \varphi(s)ds.  \label{equ:22}
\end{align}

To show $u(f^B_i, f^B_{-i})\geq u(f'^B_i, f^B_{-i})$, we only need to show the second part of the above equation is less or equal to the second part of the equation. To show this, we will prove that for any $s\in K_i\setminus K'_i$, $U_{M(s, \mathbf{f_B})}(s, \mathbf{f_A})\geq U_{M(s', \mathbf{f_B})+1}(s', \mathbf{f_A})$ for all $s\in K_i'\setminus K_i$. Let $s\in K_i\setminus K'_i$, $s'\in K_i'\setminus K_i$ be arbitrary. First, since the applicant with score $s$ is interviewed by firm $i$, we must have $M(s, \mathbf{f_B})>0$. By the third condition we assumed, there must exist a constant $H$ such that  $U_{M(s, \mathbf{f_B})}(s, \mathbf{f_A})\geq H$. If $M(s', \mathbf{f_B}) = 0$, then we must have $U_{M(s', \mathbf{f_B})+1}(s', \mathbf{f_A}) = U_{1}(s', \mathbf{f_A})<H$ by assumption. If $M(s', \mathbf{f_B}) > 0$, then by assumption we must have $U_{M(s', \mathbf{f_B})+1}(s', \mathbf{f_A})<H$. Therefore, we must have $U_{M(s, \mathbf{f_B})}(s, \mathbf{f_A})\geq U_{M(s', \mathbf{f_B})+1}(s', \mathbf{f_A})$ for all $s\in K_i'\setminus K_i$, $s\in K_i\setminus K'_i$.

Furthermore, by condition 2, $\Pr(S\in K_i) = \Pr(S\in K'_i) = c$. Hence, 
\begin{align*}
    \Pr(S\in K_i\setminus K'_i) &=\Pr(S\in K_i) - \Pr(S\in K_i\cap K'_i)\\
    &= Pr(S\in K'_i)- \Pr(S\in K_i\cap K'_i)\\
    & = \Pr(S\in K'_i\setminus K_i).
\end{align*}
We can rewrite equation (\ref{equ:22}) as follows:
\begin{align*}
    u(f'^B_i, f^B_{-i}) & \leq \int_{K_i\cap K'_i}U_{M(s, \mathbf{f_B})}(s, \mathbf{f_A}) \varphi(s)ds + \Pr(S\in K'_i\setminus K_i)(\max_{s'\in K_i'\setminus K_i}U_{M(s', \mathbf{f_B})+1}(s', \mathbf{f_A}) )\\
    & \leq \int_{K_i\cap K'_i}U_{M(s, \mathbf{f_B})}(s, \mathbf{f_A}) \varphi(s)ds + \Pr(S\in K_i\setminus K'_i)(\min_{s\in K_i\setminus K'_i}U_{M(s, \mathbf{f_B})}(s, \mathbf{f_A}))\\
    & \leq \int_{K\cap K'} U_{M(s, \mathbf{f_B})}(s, \mathbf{f_A}) \varphi(s)ds + \int_{K\setminus K'}U_{M(s, \mathbf{f_B})}(s, \mathbf{f_A})\varphi(s)ds\\
    & = u(f^B_i, f^B_{-i}).
\end{align*}

Therefore, $u(f'^b_i, f^B_{-i}) \leq u(f^B_i, f^B_{-i})$ for any arbitrary firm $i$ and alternative strategy $f'^B_i$. We conclude that none of the tier B firms is able to increase their utility by deviating from the current strategy profile. 

Since none of the tier A or tier B firms can increase their utility by changing their strategy, the strategy profile $\mathbf{f_A}\cup\mathbf{f_B}$ is at Nash equilibrium. 
    
\end{myproof}

\subsection{Proof of Proposition \ref{prop:extension_sw}}
\begin{myproof}
We assume that there are $N_A$ tier A firms and $N_B$ tier B firms. Let $\mathbf{f}_A$  and $\mathbf{f}_B$ be the strategy profile of the tier A and tier B firms, respectively. Under the Nash equilibrium, there are thresholds $\tau'_{1}\leq \tau'_2\cdots\leq \tau'_{n_A}\leq \tau_1 \leq \tau_2\cdots\tau_{N_B}\leq \tau_{n_B+1}=1$. The total social welfare is
\begin{align*}
    SW_{\mathrm{tiered}} &= \int_{\tau'_1}^{\tau'_2}s\varphi(s)ds + 2\int_{\tau'_2}^{\tau'_3}\frac{s}{2}\varphi(s)ds + \cdots + N_B \int_{\tau_{N_B}}^{1}\frac{s}{N_B}\varphi(s)ds\\
    & = \int_{\tau'_1}^{1}s\varphi(s)ds.
\end{align*}
Similarly, suppose there are $N$ firms that form a Nash equilibrium and let $\tau''_1, \cdots, \tau''_N$ be the thresholds and $\mathbf{f}$ be the strategy profile. The total social welfare is 
\begin{align*}
    SW_{\mathrm{non\text{-}tiered}} = \int_{\tau''_1}^{1}s\varphi(s)ds.
\end{align*}
To compare $SW_{\mathrm{tiered}}$ and $SW_{\mathrm{non\text{-}tiered}}$, 
it suffices to compare the lowest thresholds $\tau_1'$ and $\tau_1''$. 
We show that $\tau_1' < \tau_1''$. 

Suppose, for contradiction, $\tau_1'\geq \tau''_1$. We first consider the case where all firms are at an equal-utility Nash equilibrium for both the tiered and non-tiered settings. Recall that under the correlated decision rules, we must have $\tau'_{i^{\text{equal}}}=i\tau'_{1^{\text{equal}}}\geq i\tau''_{1^{\text{equal}}}=\tau''_{i^{\text{equal}}}$ for all $i$. That is, $\tau'_{i^{\text{equal}}}\geq \tau''_{i^{\text{equal}}}$ for all $i$. Thus, we have $M(s, \mathbf{f}^{\text{equal}})\geq M(s, \mathbf{f}_A^{\text{equal}})$ on the interval $[\tau''_{1^{\text{equal}}}, \tau_{1^{\text{equal}}}]$. Next, note that there is at least one $N_B$ firm, so we must have $\tau_{1^{\text{equal}}}>\tau_1'^{\text{equal}}\geq \tau''_{1^{\text{equal}}}$. Hence, $\tau_{i^{\text{equal}}}=i\tau_{1^{\text{equal}}}>i\tau''_{1^{\text{equal}}}=\tau''_{i^{\text{equal}}}$, and $M(s, \mathbf{f}^{\text{equal}})\geq M(s, \mathbf{f}_B^{\text{equal}})$ on the interval $[\tau_{1^{\text{equal}}}, 1]$. Moreover, there exist sub-intervals of $[\tau_{1^{\text{equal}}}, 1]$ such that on those sub-intervals we have $M(s, \mathbf{f}^{\text{equal}})> M(s, \mathbf{f}_B^{\text{equal}})$. We also observe that the total capacity of all firms is the same under the tiered and the non-tiered case. That is, 
\begin{align*}
    Nc &= \int_{\tau''_{1^{\text{equal}}}}^{\tau_{1^{\text{equal}}}}M(s, \mathbf{f}^{\text{equal}})\varphi(s)ds + \int_{\tau_{1^{\text{equal}}}}^{1}M(s, \mathbf{f}^{\text{equal}})\varphi(s)ds \\
    & =\int_{\tau'_{1^{\text{equal}}}}^{\tau_{1^{\text{equal}}}}M(s, \mathbf{f}_B^{\text{equal}})\varphi(s)ds + \int_{\tau_{1^{\text{equal}}}}^{1}M(s, \mathbf{f}_A^{\text{equal}})\varphi(s)ds = Nc,
\end{align*}
which is impossible. Therefore, we must have $\tau_1'^{\text{equal}}<\tau''_{1^{\text{equal}}}$. 

Now suppose either the tier A firms or the tier B firms are at a variable-utility Nash equilibrium under the tiered case. We will show that $\tau_1^\text{var}<\tau_{1^{\text{equal}}}$. Suppose, for contradiction, $\tau_{1^\text{var}}\geq\tau_{1^{\text{equal}}}$. Since we have $U_n(s)<U_m(s)$ for some $n<m$ under the variable-utility setting, we must have $\tau'_{i^\text{var}}>\tau'_{i^{\text{equal}}}$ or $\tau_i^\text{var}>\tau_{i^{\text{equal}}}$ for some index $i$. That is, $M(s, \mathbf{f}_{B^{\text{equal}}})> M(s, \mathbf{f}^{\text{var}_B})$ or $M(s, \mathbf{f}_{A^{\text{equal}}})> M(s, \mathbf{f})^{\text{var}_A})$ for some $s$. Then we have 
\begin{align*}
    Nc &= \int_{\tau'_{1^\text{var}}}^{\tau_{1^\text{var}}}M(s, \mathbf{f}^\text{var}_A)\varphi(s)ds + \int_{\tau_{1^\text{var}}}^{1}M(s, \mathbf{f}^\text{var}_B)\varphi(s)ds \\
    & <\int_{\tau'_{1^{\text{equal}}}}^{\tau_{1^{\text{equal}}}}M(s, \mathbf{f}^{\text{equal}}_A)\varphi(s)ds + \int_{\tau_{1^{\text{equal}}}}^{1}M(s, \mathbf{f}^{\text{equal}}_B)\varphi(s)ds = Nc,
\end{align*}
which again is a contradiction. Therefore, we must have $\tau'_{1^\text{var}}<\tau'_{1^{\text{equal}}}<\tau''_{1^{\text{equal}}}$. 

Finally, we only need to show the case where the non-tiered firms form a variable-utility Nash equilibrium. That is, we will show $\tau'_{1^{\text{equal}}}<\tau''_{1^\text{var}}$. Suppose, for contradiction,  $\tau'_{1^{\text{equal}}}\geq \tau''_{1^\text{var}}$. We will first define the following set of thresholds: $\{\tau''_{i^\text{ext}}\}$ with $i\in [n^\text{var}]$ such that $\tau''_{i^\text{ext}} = i\tau''_{i^\text{var}}$. Define $\mathbf{f}^\text{ext}$ to be the strategy profile corresponding to the set of thresholds $\tau''_{i^\text{ext}}$. That is, there are exactly $i$ firms on the interval $[\tau''_{i^\text{ext}}, \tau''_{{i+1}^\text{ext}]}$. 

Next, let's first consider the case where $\tau''_{2^\text{var}}< \tau_{1^{\text{equal}}}$. Then on the interval $[\tau_{1^{\text{equal}}},1]$, we have $M(s, \mathbf{f}^\text{var})-M(s, \mathbf{f}_A^{\text{equal}})\geq 1$, and the inequality is strictly on some sub-intervals of $[\tau_{1^{\text{equal}}},1]$. Therefore, we must have 
\begin{align*}
    \int_{\tau_{1^{\text{equal}}}}^{1}M(s, \mathbf{f}^\text{var})\varphi(s)ds - \int_{\tau_{1^{\text{equal}}}}^{1}M(s, \mathbf{f}^{\text{equal}}_A)\varphi(s)ds > \int_{\tau_{1^{\text{equal}}}}^{1}\varphi(s)ds = c,
\end{align*}
because there is at least one tier A firm. Therefore, the non-tiered NE has at least c more capacity on this interval than the tiered NE under this setting. Next, consider the interval $[\tau''_{1^\text{var}}, \tau_{1^{\text{equal}}}]$. By assumption, we have $\tau'_{i^{\text{equal}}}=i\tau'_{1_{\text{equal}}}>i\tau''_{1^\text{var}}=\tau''_{i^\text{ext}}$. Let $I$ be the set of indices such that $\tau''_{i^\text{var}}>\tau''_{i^\text{ext}}$. Then the total capacity under the tiered setting can exceed the total capacity under the non-tiered setting on the interval $[\tau''_{1^\text{var}}, \tau_{1^{\text{equal}}}]$ by at most
\begin{align*}
    &\sum_{i\in I, i\leq n_B}\int_{\tau'_{i^{\text{equal}}}}^{\tau''_{i^\text{var}}} \varphi(s)ds - \sum_{i\notin I, i\leq n_B}\int_{\tau''_{i^\text{var}}}^{\tau'_{i^{\text{equal}}}} \varphi(s)ds \\
    & \leq \sum_{i\in I, i\leq n_B}\int_{\tau'_{i^{\text{equal}}}}^{\tau''_{i^\text{var}}} \varphi(s)ds\\
    & \leq \sum_{i\in I, i\leq n_B}\int_{\tau''_{i^\text{ext}}}^{\tau''_{i^\text{var}}} \varphi(s)ds\\
    & \leq  \sum_{i\in I}\int_{\tau''_{i^\text{ext}}}^{\tau''_{i^\text{var}}} \varphi(s)ds \leq c.
\end{align*}
The last inequality holds because for the firm that interviews the applicant with score $\tau''_{1^\text{var}}$ must also interview everyone in $\cup_{i\in I}[\tau''_{i^\text{ext}}, \tau''_{i^\text{var}}]$ and its capacity cannot exceed $c$. Therefore, combining the two results above, we conclude that the tiered setting will have strictly less total capacity than the non-tiered setting, which is not possible. 

Finally, we consider the case where $\tau''_{2^\text{var}}\geq \tau_{1^{\text{equal}}}>\tau''_{1^\text{var}}$. Again, assume that for contradiction, $\tau'_{1^{\text{equal}}}\geq \tau''_{1^\text{var}}$. We first notice that we've already shown that $\tau'_{1^{\text{equal}}}<\tau''_{1^{\text{equal}}}$. Then, under the assumption we made,  we have $\tau''_{1^\text{ext}}<\tau'_{1^{\text{equal}}}<\tau''_{1^{\text{equal}}}$. The total capacity under $\mathbf{f}^\text{var}$ exceeds the total capacity under ${\mathbf{f}_A}^{\text{equal}}\cup {\mathbf{f}_B}^{\text{equal}}$ by at least
\begin{align}
    &\sum_{1\leq i\leq n_B}\int_{\tau''_{i^\text{ext}}}^{\tau'_{i^{\text{equal}}}} \varphi(s)ds + \sum_{i\geq 2}\int_{\tau''_{i^\text{ext}}}^{\tau_{i^{\text{equal}}}} \varphi(s)ds - \sum_{i\geq 1}\int_{\tau''_{i^\text{ext}}}^{\tau''_{i^\text{var}}} \varphi(s)ds  \nonumber\\
    & = \sum_{1\leq i\leq n_B}\int_{\tau''_{i^\text{ext}}}^{\tau'_{i^{\text{equal}}}} \varphi(s)ds -\int_{\tau''_{1^\text{ext}}}^{\tau''_{1^\text{var}}} \varphi(s)ds + \sum_{i\geq 2}\int_{\tau''_{i^\text{var}}}^{\tau_{i^{\text{equal}}}} \varphi(s)ds  \nonumber\\
    & = \sum_{1\leq i\leq n_B}\int_{\tau''_{i^\text{ext}}}^{\tau'_{i^{\text{equal}}}} \varphi(s)ds -\int_{\tau''_{1^\text{ext}}}^{\tau''_{1^\text{var}}} \varphi(s)ds
    -\int_{\tau''_{1^\text{var}}}^{\tau_{1^{\text{equal}}}} \varphi(s)ds
    + \sum_{i\geq 1}\int_{\tau''_{i^\text{var}}}^{\tau_{i^{\text{equal}}}} \varphi(s)ds \nonumber\\
    & = \sum_{1\leq i\leq n_B}\int_{\tau''_{i^\text{ext}}}^{\tau'_{i^{\text{equal}}}} \varphi(s)ds -\int_{\tau''_{1^\text{ext}}}^{\tau_{1^{\text{equal}}}} \varphi(s)ds + \text{total capacity of tier B firms} \nonumber\\
    & = \sum_{2\leq i\leq n_B}\int_{\tau''_{i^\text{ext}}}^{\tau'_{i^{\text{equal}}}} \varphi(s)ds -\int_{\tau'_{1^{\text{equal}}}}^{\tau_{1^{\text{equal}}}} \varphi(s)ds + \text{total capacity of tier B firms} \label{eq:capacity_prop}
\end{align}
We observe that the $\text{total capacity of tier B firms}\geq \int_{\tau'_{1^{\text{equal}}}}^{\tau_{1^{\text{equal}}}} \varphi(s)ds$. If $n_B > 1$, then clearly the equation \ref{eq:capacity_prop} is greater than 0. If $n_B = 1$, then we have $M(s, \mathbf{f}^\text{var})\geq M(s, \mathbf{f}_A^{\text{equal}}\cup \mathbf{f}_B^{\text{equal}}$) for all $s\in [0,1]$. Specifically, since $\tau''_{1^\text{var}}<\tau'_{1^{\text{equal}}}$, the inequality is strict. Hence, $\mathbf{f}^\text{var}$ must have strictly more capacity. This is again impossible. Hence, we've shown that in all the possible cases, we must have $\tau'_{1^{\text{equal}}}<\tau''_{1^\text{var}}$. Since we've shown $\tau'_{1^\text{var}}\leq\tau'_{1^{\text{equal}}}$, we conclude that $\tau'_1<\tau''_1$ under any possible scenario. Hence, the social welfare under the tiered setting is always higher than the social welfare under the non-tiered setting. 
\end{myproof}

\section{Extension: Mixed Strategies}\label{appendix:B}

In our model, we have assumed each firm $i$ decides on a deterministic function $f_i(s):[0,1]\rightarrow\{0,1\}$. In this section, we extend our model to allow for randomized strategies. Specifically, instead of consistently choosing a single action $f_i$, Firm $i$ now selects an action $f_{i_j}$ from the finite set of pure strategies $\{f_{i_1}, f_{i_2},\cdots, f_{i_{K_i}}\}$ with probability $p_j$ where $j\in [K_i]$, $K_i\in \mathbb{Z}$, and $\sum_{j=1}^K p_j = 1$. 

We show that even in this more general setting with randomized firm decisions, a pure strategy Nash equilibrium always exists.  Thus, our earlier results naturally extend to this broader framework.

\begin{theorem} \label{thm:mix_strategy}
    In any Nash equilibrium, each firm's strategy must be deterministic. 
\end{theorem}

Before proving the theorem, we first present the following lemma:
\begin{lemma}\label{lemma:mix_strategy}
    Let $\mathbf{f}$ be an arbitrary action profile. For any constant t and distribution $\mathcal{D}$, $\mathbb{P}_{S\sim \mathcal{D}}( U_{M(s, \mathbf{f})}(s) = t) = 0$.
\end{lemma}
\begin{myproof}
Let $t$ be an arbitrary constant and $S = \{s\in [0,1]: U_{M(s, \mathbf{f})}(s) = t\}$. We will show that the set $S$ contains finitely many points. Suppose $\{f_{i_1}, f_{i,2}\cdots, f_{i,K_i}\}$ is the strategy set for firm $i$. By assumption, under each strategy $f_{i_j}$, applicants interviewed by firm $i$ can be written as a union of finite intervals. Let $S_{i_j}=\{s\in[0,1]:f_{i_j}(s)=1\}$, $\cup_{j=1}^{K_i}S_{i_j}$ is the union of a finite number of intervals. Similarly, $\cup_{i=1}^N\cup_{j=1}^{K_i}S_{i_j}$ is also a union of finitely many intervals. Recall that $M(s, \mathbf{f})$ is the expected number of firms that interview an applicant of score $s$. We conclude that there are finitely many possible values for $M(s, \mathbf{f})$. We note that for each fixed $n$, $U_n(s)$ is strictly increasing in $s$. Therefore, for each value of $M(s, \mathbf{f})$, $U_{M(s, \mathbf{f})}(s) = t$ has a unique solution. Since there are finitely many values for $M(s, \mathbf{f})$, the set $S$ is also finite. Therefore, we must have $\mathbb{P}(s\in S)=0$. 
\end{myproof}

\begin{myproof}[\text{Proof of Theorem }\ref{thm:mix_strategy}]
    We want to show that each firm's action must be deterministic under Nash equilibrium. Suppose,  by contradiction, there is an action profile $\mathbf{f} = (f_1, \cdots, f_N)$ that is a Nash equilibrium with at least one firm's action $f_i$ being non-deterministic. For the randomized strategy $f_i$, we will show that there exists some pure strategy $f'_i$ that yields strictly higher utility. Let the mixed strategy of firm $i$ be given by the finite set of pure strategies 
$\{ f_{i1}, f_{i2}, \ldots, f_{iK_i} \}$, 
where firm $i$ plays $f_{ij}$ with probability $p_j$, 
$j \in [K_i]$, $\sum_{j=1}^{K_i} p_j = 1$. 

Fix all other firms' strategies $f_{-i}$, and consider firm $i$’s utility under each pure strategy $f_{ij}$. 
Let $S_j = \{ s \in [0,1] : f_{ij}(s) = 1 \}$ 
and $\hat{S} = \bigcup_{j \in [K_i]} S_j$ be the union of all score intervals where firm $i$ interviews applicants under at least one of its pure strategies. 
Define the threshold score
\[
t^* = \max \{ t : \mathbb{P}(s \in\hat S : U_{M(s,f)}(s) > t) =c\}
\]
and the set $$S^* = \{s\in S: U_{M(s, \mathbf{f})}(s)>t^*\}.$$
We then define a pure strategy 
$f_i'(s) = 1$ whenever $U_{M(s,f)}(s) \ge t^*$ 
and $f_i'(s) = 0$ otherwise.
Intuitively, firm $i$ interviews exactly those applicants whose expected marginal utility exceeds $t^*$.

Next, we will show that $u(f'_i, f_{-i})\geq u(f_{i_j}, f_{-i})$ for each $j\in[k]$.  We first observe that, following the proof of condition 1 in Theorem \ref{thm:equilibrium} , any strategy profile $S_j$ that is a Nash equilibrium must satisfy $P(s\in S_j) = c$. Otherwise, we can always expand the set $S_j$ to increase the utility. As a result, we must have $\mathbb{P}(s\in S^*) = \mathbb{P}(s\in S_j)$ for any arbitrary $j\in[K_i]$. Equivalently, we have $\mathbb{P}(s\in S^*\setminus S_j) = \mathbb{P}(s\in S_j \setminus S^*)$ for any $j\in[K_i]$. Therefore,

\begin{align*}
    u(f'_i, f_{-i}) &= \int_{S^*}U_{M(s, \mathbf{f})}(s) \varphi(s) ds\\
    & =  \int_{S^*\cap S_j}U_{M(s, \mathbf{f})}(s) \varphi(s) ds +  \int_{S^*\setminus S_j}U_{M(s, \mathbf{f})}(s) \varphi(s) ds\\
     & \geq  \int_{S^*\cap S_j}U_{M(s, \mathbf{f})}(s) \varphi(s) ds +  \Pr(s\in S^*\setminus S_j)t^*\\
     & \geq  \int_{S^*\cap S_j}U_{M(s, \mathbf{f})}(s) \varphi(s) ds +  \int_{S_j\setminus S^*}U_{M(s, \mathbf{f})}(s) \varphi(s) ds\\
     & = u(f_{i_j}, f_{-i})
\end{align*}
For the equality to hold, we must have $U_{M(s, \mathbf{f})}(s) = t^*$ for both $S^*\setminus S_j$ and $S_j\setminus S^*$. The inequality is strict whenever $\varphi(s) > 0$ on 
$(S^* \setminus S_j) \cup (S_j \setminus S^*)$, 
which must occur if $f_i$ is non-deterministic, since by Lemma~\ref{lemma:mix_strategy} 
the set $\{ s : U_{M(s,f)}(s) = t^* \}$ has measure zero. 
Hence there exists at least one $j$ such that
\[
u(f_i', f_{-i}) > u(f_{ij}, f_{-i}).
\]

The total utility from firm $i$’s mixed strategy is
\[
u(f_i, f_{-i}) 
= \sum_{j \in [K_i]} p_j \, u(f_{ij}, f_{-i})
< \sum_{j \in [K_i]} p_j \, u(f_i', f_{-i})
= u(f_i', f_{-i}).
\]
Thus, the pure strategy $f_i'$ yields strictly higher utility than any mixed combination, contradicting the assumption that $f_i$ is a best response.
\end{myproof}